%% file: ms.tex
\pdfoutput=1
\documentclass[11pt]{article}
\input{headers_and_macros}

\title{Worst-case Optimal Binary Join Algorithms under General $\ell_p$ Constraints}

\author{Sai Vikneshwar Mani Jayaraman$^1$, Corey Ropell$^2$ and Atri Rudra$^3$}
\date{%
	$^1$AWS Redshift\\
	$^2$Amazon, Inc\\
	$^3$University at Buffalo\\
}

\begin{document}
\addcontentsline{toc}{section}{Appendix}
\addtocontents{toc}{\protect\setcounter{tocdepth}{0}}
\maketitle
\input{abstract}

\input{introduction}
\input{generic_algorithm}
\input{prelims_and_notation}
\input{proof_for_worst_case_optimality}
\input{cycle_main}
\input{panda_main}
\input{related_work}
\input{conclusion}
\bibliographystyle{acm}
\bibliography{atri}
\addtocontents{toc}{\protect\setcounter{tocdepth}{2}}
\appendix
\section*{Acknowledgements}
We are greatly indebted to Szymon Toru\'nczyk for many fruitful discussions during the early stages of this project. We thank Shi Li for giving us the results in Appendix~\ref{app:shi}. We thank Mahmoud Abo Khamis, Oliver Kennedy, Shi Li, Hung Ngo and Dan Suciu for helpful discussions. Finally, we thank NSF for their generous support through the grant CCF-1763481 and Amazon, where a part of this work was done.

\tableofcontents
\input{appendix_sec1}

\input{appendix_sec2}
\input{appendix_sec3}
\input{appendix_structural_result}

\input{appendix_sec4}
\input{appendix_sec4_1}
\input{appendix_sec4_2}
\input{appendix_sec5}
\input{appendix_conclusion}
\end{document}

%% file: headers_and_macros.tex
\usepackage{amsmath,amssymb,amsthm}
\usepackage[T1]{fontenc} 

\usepackage{fourier}
\usepackage{fullpage}

\newcommand\xqed[1]{%
	\leavevmode\unskip\penalty9999 \hbox{}\nobreak\hfill
	\quad\hbox{#1}}
\newcommand\demo{\xqed{$\blacktriangleleft$}}

\usepackage{authblk}
\usepackage{color,xspace}
\usepackage{todonotes}
\usepackage{algorithm}
\usepackage[noend]{algpseudocode}
\usepackage{caption}
\usepackage{blkarray}
\usepackage{multirow}

\usepackage{colortbl}
\usepackage{enumitem}

\usepackage{hyperref}

\usepackage{tikz,ifthen}
\usetikzlibrary{arrows.meta}
\usepackage{arydshln}

\newtheorem{thm}{Theorem}[section]
\newtheorem{conjecture}{thm}[section]
\newtheorem{lemma}[thm]{Lemma}
\newtheorem{corollary}[thm]{Corollary}
\newtheorem{claim}[thm]{Claim}
\newtheorem{lemm}[thm]{Lemma}
\newtheorem{cor}[thm]{Corollary}

\newtheorem{defn}[thm]{Definition}
\newtheorem{problem}[thm]{Problem}
\newtheorem{property}[thm]{Property}
\newtheorem{Example}[thm]{Example}
\newtheorem{assumption}[thm]{Assumption}
\newtheorem{ques}[thm]{Question}

\newcommand{\vx}{\mathbf{x}}

\newcommand{\abs}[1]{\left|{#1}\right|}

\newcommand{\cA}{\mathcal{A}}
\newcommand{\cI}{\mathcal{I}}

\newcommand{\cC}{\mathcal{C}}
\newcommand{\cV}{\mathcal{V}}
\newcommand{\cE}{\mathcal{E}}

\newcommand{\yell}[1]{\textcolor{red}{\textbf{#1}}}

\newcommand{\coreyyell}[1]{\textcolor{orange}{#1}}

\newcommand{\cH}{\mathcal{H}}
\newcommand{\cT}{\mathcal{T}}
\newcommand{\cD}{\mathcal{D}}

\newcommand{\eps}{\epsilon}

\newcommand{\bd}{\mathbf{d}}
\newcommand{\bard}{\mathbf{d}}

\newcommand{\cP}{\mathcal{P}}

\usepackage{amsthm,amsmath,amssymb}

\newcommand{\vy}{\mathbf{y}}

\newcommand{\set}[1]{\left\{#1\right\}}

\renewcommand{\tilde}{\widetilde}

\newcommand{\ceil}[1]{{\left\lceil#1\right\rceil}}
\newcommand{\floor}[1]{{\left\lfloor#1\right\rfloor}}

\newcommand{\Dom}{\mathrm{Dom}}

\newcommand{\mvec}[1]{\mathbf{#1}}

\newcommand{\vd}{\mvec{d}}
\newcommand{\inset}[1]{\left\{#1\right\}}
\newcommand{\inparen}[1]{\left(#1\right)}
\renewcommand{\Join}{\bowtie}

\mathchardef\mhyphen="2D

\newcommand{\vt}{\mathbf{t}}

\newcommand{\norm}[1]{\left\lVert#1\right\rVert}

\newcommand\numberthis{\addtocounter{equation}{1}\tag{\theequation}}

\newcommand{\B}{\mathcal{B}}

\newcommand{\edgeud}[2]{\ensuremath{({#1}, {#2})}}
\newcommand{\edgeu}[2]{\ensuremath{\set{{#1} , {#2}}}}
\newcommand{\edge}[2]{\ensuremath{({#1} \rightarrow {#2})}}
\newcommand{\edged}[2]{\ensuremath{({#1} \rightarrow {#2}, {\mathbf{d}})}}
\newcommand{\edgep}[2]{\ensuremath{({#1} \rightarrow {#2}, {p})}}

\newcommand{\edgeinfty}[2]{(\ensuremath{#1} \rightarrow \ensuremath{#2}, \infty)}
\newcommand{\LP}{\mathrm{LP}}

\newcommand{\rank}{\mathrm{rank}}

\newcommand{\inc}{\mathrm{In}}
\newcommand{\out}{\mathrm{out}}

\newcommand{\usum}[1]{\underset{#1}{\sum}}
\newcommand{\uprod}[1]{\underset{#1}{\prod}}
\newcommand{\ujoin}[1]{\underset{#1}{\Join}}
\newcommand{\ucup}[1]{\underset{#1}{\cup}}
\newcommand{\umin}[1]{\underset{#1}{\min}}
\newcommand{\umax}[1]{\underset{#1}{\max}}

\newcommand{\ty}{\tilde{\mathrm{y}}}

\newcommand{\vz}{\mathbf{z}}

\newcommand{\val}{\mathrm{val}}

\newcommand{\TJ}{\mathrm{J_G^{(I)}}}

\newcommand{\poly}{\mathrm{poly}}
\allowdisplaybreaks

\definecolor{darkgray}{rgb}{0.4, 0.4, 0.4}

\makeatletter
\newenvironment{breakablealgorithm}
{
	\begin{center}
		\refstepcounter{algorithm}
		\hrule height.8pt depth0pt \kern2pt
		\renewcommand{\caption}[2][\relax]{
			{\raggedright\textbf{\ALG@name~\thealgorithm} ##2\par}%
			\ifx\relax##1\relax 
			\addcontentsline{loa}{algorithm}{\protect\numberline{\thealgorithm}##2}%
			\else 
			\addcontentsline{loa}{algorithm}{\protect\numberline{\thealgorithm}##1}%
			\fi
			\kern2pt\hrule\kern2pt
		}
	}{
		\kern2pt\hrule\relax
	\end{center}
}
\makeatother

\newcommand{\panda}{\mathsf{PANDA}}
\newcommand{\lftj}{\mathsf{LFTJ}}
\newcommand{\toporder}{\mathsf{TopologicalOrdering}}

\newcommand{\mypar}[1]{\smallskip\noindent\textbf{{#1}.}}

%% file: abstract.tex
\begin{abstract}
Worst-case optimal join algorithms have so far been studied in two broad contexts -- $(1)$ when we are given input relation sizes [Atserias et al., FOCS 2008, Ngo et al., PODS 2012, Velduizhen et. al, ICDT 2014]  $(2)$ when in addition to size, we are given a degree bound on the relation [Abo Khamis et al., PODS 2017]. To the best of our knowledge, this problem has not been studied beyond these two statistics even for the case when input relations have arity (at most) two. 

In this paper, we present a worst-case optimal join algorithm when are given $\ell_{p}$-norm size bounds on input relations of arity at most two for $p \in (1, 2]$. ($p=1$ corresponds to relation size bounds and $p=\infty$ correspond to the degree bounds.) The worst-case optimality holds any fixed $p \in (2, \infty)$ as well (as long as the join query graph has large enough girth). Our algorithm is {\em simple}, does not depend on $p$ (or) the $\ell_{p}$-norm bounds and avoids the (large) poly-log factor associated with the best known algorithm $\panda$ [Abo Khamis et al., PODS 2017] for the size and degree bounds  setting of the problem. In this process, we (partially) resolve two open question from [Ngo, 2018 Gems of PODS].
 We believe our algorithm has the {\em potential} to pave the way for practical worst-case optimal join algorithms beyond the case of size bounds. 

\end{abstract}

%% file: introduction.tex
\section{Introduction} \label{sec:norm-intro}
Over the last decade or so, there has been a surge of interest in designing {\em worst-case optimal} join algorithms where the goal is to design algorithms that compute the natural join query in time that is linear  in worst-case size bounds on the join output based on some statistics about the input relations. The first such results were based on statistics on sizes of input relations -- the celebrated Atserias-Grohe-Marx (AGM) result proved the tight worst-case bounds on the join output size~\cite{AGM},  which were later realized via an algorithmic result by Ngo et al.~\cite{NPRR} (also see~\cite{V14}). That result has since been extended to handle {\em degree} bounds on the relations in addition to size bounds by Abo Khamis et al.~\cite{panda,panda-pods20}, though these results have drawbacks in the sense that their algorithm $\panda$ -- $(1)$ depends on size of the input relations and the degree bounds and $(2)$ loses a large multiplicative factor poly-logarithmic in the input sizes in its runtime analysis, which makes $\panda$ impractical (unlike the algorithm in~\cite{NPRR}, which has neither of these drawbacks).

In this paper, we mainly focus on the join processing problem for relations with arity (at most) two, which has applications in graph databases~\cite{semih-survey} (which we will discuss in detail in a bit), particularly for pattern matching in SPARQL for RDF data~\cite{sparql} and columnar databases~\cite{Rowid} (a connection we discuss in Appendix~\ref{app:columnar}). In particular, for any binary relation $R(A,B)$, we assume an arbitrary direction between the attributes say $\edge{A}{B}$.  Then, for any constant $a\in \Dom(A)$, the degree of $a$ in $R$ is the number of  tuples $(a,b)$ that are  in $R$. Now, consider the `degree vector' $\vd_R$ that we get by collecting the degree of every constant  in $\Dom(A)$ (or rather the effective  domain of $A$). Let  $\norm{R}_p$ denote the $\ell_p$ norm of  this  degree vector i.e., 
\begin{align*}
\norm{R}_p  \stackrel{\text{def}}{=}\norm{\vd_R}_p =\sqrt[p]{\usum{a\in\Dom(A)} \inparen{\vd_R[a]}^p }.
\end{align*}
When $p  = 1$, the above gives us the usual size bound and when $p = \infty$, we get the degree bound (in the chosen direction). Given a join query with $m$ relations $(R_{i})_{i \in [m]}$, we define a directed join query graph $G = (V, E)$, where each edge in $G$ corresponds to the schema of $R_i$ (for every $i \in [m]$) and each vertex corresponds to an attribute in the join query. We define a natural join $\ujoin{e \in E} R_e$, where every tuple $\vt \in \ujoin{e \in E}R_e$ satisifies $\pi_{e}(\vt) \in R_{e}$ for every $e \in E$. Here, $\pi_{e} (\vt)$ denotes the projection of $\vt$ on to vertices/attributes in $e$.

We are now ready to state the main problem we consider in this paper. 
\begin{ques} \label{ques:main}
Fix  $p\ge 1$. Given a directed join graph $G=(V,E)$ of the corresponding natural join query  $\ujoin{e = \edge{v}{u} \in E} R_{\edge{v}{u}}$\footnote{Throughout the paper, we will use $e$ and $\edge{v}{u}$ interchangeably to denote an edge in $E$.}  such that $\norm{R_{\edge{v}{u}}}_p\le L_{\edge{v}{u}}$ for every $e \in E$,\footnote{Our results also hold for the case when the directions of tuples are not fixed upfront and instead follow from an appropriate definition of the undirected $\ell_p$-norm (we discuss this in detail in Appendix~\ref{app:orientation}).} 
 can we design an algorithm with data complexity linear in the worst-case join size bound for this setting? Here, the  worst-case join size bound is over all possible relations  that satisfy the given norm bounds $L_e$ for all $e\in E$.  We also consider the setting where we are also given degree bounds $\norm{R_e}_\infty \le d_e$.
\end{ques}

Next, we discuss the motivation of studying this problem for general $p$ (specifically $p \in (1, 2]$).
\subsection{Motivation and Background}
The main motivation for Question~\ref{ques:main} comes from the setup of graph databases\\~\cite{semih-survey}, where a common goal is to enumerate all occurrences of a specific directed subgraph $G = (V, E)$ in a large directed graph $\cH = (\cV, \cE)$ (see Tables 7.(a) and 9 in~\cite{semih-survey} for specific examples). We note that this is a special case of Question~\ref{ques:main} where for each $\edge{u}{v}\in E$, we have:
\begin{align*}
\Dom(u)=\Dom(v)=\cV, \quad R_{\edge{v}{u}} = \cE, \quad \norm{R_{\edge{v}{u}}}_{p}  \le L.
\end{align*}
In other words, $R_{\edge{v}{u}}$ encodes $\cH$ as a bipartite graph with bipartition $\cV\times \cV$, where each $\edge{x}{y}\in \cE$ is in $R_{\edge{v}{u}}$ with $x\in\Dom(u)$ and $y\in\Dom(v)$. (Note that all input relations have exactly the same set of tuples and as a result, the same $\ell_{p}$-norm upper bound $L$.)

As mentioned earlier, this specific problem has been studied in two closely related contexts for {\em any} $G$ (including when $G$ is a hypergraph, which we do not quite cover here).
\begin{itemize}
\item {When $p = 1$, the seminal result of Atserias et. al~\cite{AGM} derived combinatorially tight size bounds for this problem, which were later leveraged to obtain worst-case optimal join algorithms~\cite{NPRR,V14}.}
\item {When $p = 1$ and a degree bound is given, combinatorially tight size bounds (the so-called {\em polymatroid bound}) for the arity two case\footnote{Their results also hold for a more general class of degree constraints for hypergraphs but not all of them.} were derived in a line of beautiful works by Abo Khamis et al.~\cite{panda,panda-pods20}. However, the corresponding algorithm $\panda$~\cite{panda} has a runtime that matches the polymatroid bound up to a multiplicative factor of \\ $O\inparen{\inparen{\log{N}}^{\inparen{\inparen{2^{|V|}}!}}}$.\footnote{Here, $N$ is an upper bound on size of $R_{\edge{v}{u}}$ for every $\edge{v}{u} \in E(G)$. To be precise, the best known bound on the multiplicative factor that can be proven in~\cite{panda} is $O\inparen{\inparen{\log{N}}^{\inparen{\inparen{2^{|V|}}!}\cdot \poly\inparen{2^{\abs{V}}}}}$ but for simplicity, we'll ignore the factor of $\poly\inparen{2^{\abs{V}}}$ in the exponent.}}
\end{itemize}
To the best of our knowledge, the above problem for values of $p\in (1,\infty)$ and specifically for $p\in (1,2]$ has not been studied before.

A natural question is if we gain anything by going beyond these two statistics in the first place? To answer this, we will begin with the case when $G$ is a triangle (which we will also use as a running example in this and the next section).

When $G$ is a triangle, the final bound obtained from both the above settings is $\min (N^{3/2}, N d)$, where $\umax{\edge{v}{u} \in E} |R_{\edge{v}{u}}|=N$ and $\umax{\edge{v}{u} \in E} d_{\edge{v}{u}}=d$. We consider the case of $p = 2$, where we have
\begin{align*}
\norm{R_{\edge{v}{u}}}_{2} = \sqrt[2]{\usum{v \in \cV(\cH)} \deg(v)^2} \le L \quad \text{ for every } \edge{v}{u} \in E(G).
\end{align*}

In this paper, we prove an upper bound of $L^2$ on the number of triangles with the above $\ell_2$ norm bound. We first present some numbers based on real-world benchmarks from SNAP~\cite{snap} (chosen based on a spectrum of number of edges) below. As we show in Table~\ref{table:joins-with-stat}, our bound based on $\ell_2$-norm is {\em at least} $2.75$x better than both the AGM ($N^{3/2}$) and degree-based bounds ($\min(N^{3/2}, Nd)$) 
and {\em up to} $18$x better than the AGM bound and $6$x better than the degree-based bounds.
\input{table_with_stats}

In fact, we present a theoretical justification of the results in Table~\ref{table:joins-with-stat} by considering the case when we want to list the copies of a (small) graph $G$ in a large graph $H$ that satisfies the {\em power-law} or is {\em scale-free}. Recall that a scale-free graph with exponent $\alpha$ has proportional to $k^{-\alpha}$ fraction of vertices with degree $k$. Graphs in practice tend to have $\alpha\in (2,3)$, which is what we consider mainly in this treatment.\footnote{We would like to stress that we do not claim that these graphs are very prevalent in practice (in fact, by now there is considerable doubt on whether such graphs strongly capture graphs that occur in practice~\cite{no-power-law}) but these form a {\em mathematically} natural class of graphs that have been well-studied.} For this setting, we are able to show the following (details are in Appendix~\ref{app:power-law}):
\begin{itemize}
\item Our join size bounds for scale-free graphs with exponent $\alpha$ are no worse for $p\in (1,\alpha-1]$ than those from the AGM bound (i.e., based on $\ell_1$ bound) for {\em every} graph $G$. In fact, we achieve the best bounds for $p=\alpha-1$.
\item For any $n$-cycle, our bounds based on $\ell_{\alpha-1}$-norm are asymptotically better than those that follow from $\ell_1+\ell_\infty$ bounds (as well as those based on just $\ell_1$ norm bounds).
\item For certain corner cases (e.g., for the triangle query and $\alpha=3$) our bounds are tight even for scale free graphs (our lower bounds are tight in general with respect to the instance that satisfies the given norm bounds-- these in general are not scale-free graphs).
\end{itemize}

We note that the $\ell_1$ norm based bounds are better for graphs that are complete bipartite graphs while the $\ell_1+\ell_\infty$ bounds are better for graphs where the degree distribution of the graph are very closely concentrated around the larger degrees. By contrast,
we expect better bounds based on $\ell_2$-norm when the degree distribution is skewed towards smaller degrees.
 
Finally, we would like to address the cost of maintaining the $\ell_p$-norm bound. Note that we can compute (and maintain) the degree sequence and hence, exactly compute the $\ell_p$-norm bound in linear time (and update in constant time with linear space). Since statistics are typically preprocessed (or) recomputed from scratch in periodic intervals in real-time Database systems~\cite{postgres}, we believe this is a reasonable computation cost. Additionally, $\ell_p$-norm (for $p\in [1,2]$) has (theoretically) appealing approximation guarantees in the streaming model, which we discuss in Appendix~\ref{sec:streaming}.

\subsection{Our Contributions}
We would like to translate the combinatorial gains from above to worst-case optimal join algorithms and hence, answer Question~\ref{ques:main}. More importantly, it would be convenient (both conceptually and from a practicality point of view) to have an algorithm that is {\em robust} to additional statistics i.e., if we add statistics based on another $p$, the algorithm remains the same (while the analysis could be different). 

A natural choice for such an algorithm is $\panda$~\cite{panda},\footnote{This is not immediate from $\panda$ but follows from our arguments-- see  Appendix~\ref{sec:panda-algo}.} which suffers from a multiplicative factor  of $O\inparen{\log{N}^{\inparen{\inparen{2^{|V|}}!}^2}}$. While theoretically, this is ``only'' a poly-log factor away from the (optimal) polymatroid bound, the fact that this poly-log factor depends {\em doubly-exponentially} on the query size has seriously hindered its practical implementation (unlike its worst-case optimal join counterpart~\cite{NPRR}). In fact, Ngo in the survey accompanying his 2018 Gems of PODS talk, highlighted the following open question:
\begin{ques}[Open Problem 5 in~\cite{ngo-survey}] \label{ques:panda}
Is there an algorithm running within the polymatroid bound that does not impose the poly-log (data) factor as in $\panda$ for the case when $p = 1$ and $\ell_{\infty}$ bounds are given?
\end{ques}
Further, $\panda$ needs the knowledge of $p$ (and the $\ell_p$-norm bounds) to work.

In this paper, we answer Question~\ref{ques:main}:
\begin{itemize}
\item Affirmatively for the case when $p \in (1, 2]$. For $p > 2$, our results hold for $G$s with girth\footnote{Girth here refers to length of the smallest directed cycle in $G$.} at least $p + 1$.
\item Affirmatively for the case when $p =1$ and $p={\infty}$ bounds are given, assuming same $L$ and same $d$. In this process, we answer Question~\ref{ques:panda} in the affirmative for a non-trivial class of join queries.
\end{itemize}
We achieve this using a fairly straightforward worst-case optimal join algorithm (see Section~\ref{sec:tech}). We briefly discuss our main technical result here, starting with a definition of our linear program (which we call $\LP^{(+)}$).
\begin{align*}
& \min \usum{\edge{v}{u} \in E} \left( x_{\edge{v}{u}} \log(L_{\edge{v}{u}})  \right) \\
& \usum{\edge{v}{u} \in E} x_{\edge{v}{u}}  + \usum{\edge{u}{w} \in E} \frac{x_{\edge{u}{w}}}{p} \ge 1 \quad \forall u \in V   \\
& x_{\edge{v}{u}} \ge 0 \quad \forall \edge{v}{u} \in E.
\end{align*}
\begin{thm} [Informal version of Theorem~\ref{thm:lp-main}]
For any $G$ with girth at least $p + 1$, our algorithm computes $\TJ$ in time linear in 
\begin{align*}
\Theta \left( 2^{(p + 1) |V|} \cdot \uprod{\edge{v}{u} \in E}  L_{\edge{v}{u}}^{x^*_{\edge{v}{u}}} \right), 
\end{align*}
where $\vx^* = (x^*_{\edge{v}{u}})_{\edge{v}{u} \in E}$ is an optimal solution to $\LP^{(+)}$.
\end{thm}

We first consider the case when $G$ is acyclic\footnote{The notion of acyclicity here is that $G$ has no directed cycles unlike the more common notion of acyclic hypergraphs in the query processing literature.} and all the degrees in a relation are within a factor of two-- essentially any reasonable algorithm here will work but we observe that the Leapfrog-TrieJoin ($\lftj$~\cite{V14}) when applied to this special case works well. To handle the case of general $G$ (but still with the bounded degree assumption), we simply run the algorithm above for all spanning acyclic subgraphs of $G$. Finally, to handle general degrees, we simply bucket tuples in a relation $R_{\edge{u}{v}}$ based on the degrees of values in $\Dom(u)$ and then run the previous algorithm for all possible combination of degree buckets. We would like to note here that our algorithm is independent of $p$ and the corresponding $\ell_{p}$-norm bounds -- both these are used only in our analysis.

\subsubsection{Other Results}  \label{sec:other-results}
We now remove the restriction that both $G$ and $\cH$ are directed. In particular, we consider a version where both $G$ and $\cH$ are undirected and each edge in $\cH$ can be directed in either direction. We first note that ability to orient tuples both ways in a relation can bring down the $\ell_{p}$-norm of $E(\cH)$ significantly. As an example, consider the $\ell_{\infty}$ case, where it turns out that if we allow the direction to be decided at the edge level, the $\ell_{\infty}$ bound changes from the maximum directed degree to the degeneracy~\cite{degen} of the undirected $\cH$. Our results for $p \in (1, \infty)$ hold for this case as well for appropriate definition of $\ell_p$-norm, as long as $G$ has girth (now, length of the smallest undirected cycle in $G$) at least $p + 1$. Further, our results for $p = 1$ with $\ell_{\infty}$ bounds given also hold for this case. We discuss this in detail in Appendix~\ref{app:orientation}.

\subsubsection{Dependence on $p$}
One aspect that we have avoided in our discussion so far is our restriction on $p$ in designing worst-case optimal join algorithms. It turns out that when $p \in (2, \infty]$ (say $p = 3$), the hard instance for the worst-case size lower bound on $|J|$ is not Cartesian product-based even for the case when $G$ is a triangle (a fact we discuss in detail in Section~\ref{sec:overview}). This, in turn, puts a limitation on our techniques to prove the corresponding upper bound, which rely on specific structural results based on $\LP^{(+)}$.

We conclude this section by noting that our results for Question~\ref{ques:main} can be extended to acyclic hypergraphs $G$ for the setting when $\ell_{p}$ for any $p \in [1, \infty)$ and $\ell_{\infty}$ bounds that are a subset of the ones considered in~\cite{panda-pods20} are given. The main limitation in going beyond acyclic hypergraphs seems to be our analysis (and not the algorithm). We discuss more in Section~\ref{sec:concl}.
\subsection{Implications of Our Results}
We start by discussing further implications of our results here. For the $\ell_1$, $\ell_{\infty}$ case, 
Ngo in~\cite{ngo-survey} showed that the upper bound among all possible acyclic subgraphs of $G$ is {\em finite} and raised the following question:
\begin{ques}[Open Question 3 in~\cite{ngo-survey}] \label{ques:acyclic}
Can we achieve the polymatroid bound by considering the smallest polymatroid bound among all possible acyclic spanning subgraphs (i.e., we drop a subset of the $\ell_\infty$ constraints so that the remaining $\ell_\infty$ constraints are acyclic) of the original join query graph?
\end{ques}
Through the analysis of our algorithm, we answer Question~\ref{ques:acyclic} in the affirmative for the special case of $L$ and $d$ being the same.

We hope that by answering Question~\ref{ques:panda} in the positive and the simplicity of our algorithm, our work opens the way to an eventual practical implementation of a worst-case join algorithm for the case of $\ell_1$ and $\ell_\infty$ bounds. However, we would like to mention that the additional exponential factors in its runtime due to considering acyclic subgraphs and the degree-based bucketing\footnote{While this is {\em somewhat} similar to the standard hash-partitioning technique (currently used in Database engines~\cite{postgres}), it is not immediately clear how we can handle skew in our scenario.} imply that more engineering improvements will have to be made to our algorithm's implementation can be competitive with existing worst-case optimal join variants.

At this juncture, we remark that the recent results of Abo Khamis et al.~\cite{panda,panda-pods20} are based on some beautiful results on {\em entropic inequalities}, while our results are based on more basic tools that have been used in the context of worst-case optimal join algorithms, starting at least from~\cite{NPRR}. One potential roadblock in using entropic inequalities for our results, especially those on $\ell_p$ bounds for $p\in(1,\infty)$, is that it is not immediately clear to us how those bounds can be captured in terms of entropy. However, given that our simple techniques can prove the new results presented in this paper, perhaps they can be improved and strengthened with an appropriate entropy formulation -- we leave this tantalizing possibility for future work.

Finally, to the best of our knowledge, the $\ell_p$-norm bound (for $p\in(1,\infty)$) as we define here is not a statistic that is currently used in database systems to evaluate join queries. Our work shows the potential benefit of having this statistic and we speculate that it has the potential to find applications in join query processing engines.

\subsection{Paper Organization} We present our algorithm and an overview of our techniques in Section~\ref{sec:tech} and we setup preliminaries and notation in Section~\ref{sec:norm-prelims}. Then, we present our results for $\ell_{p}$-norm bounds for general $G$ (with girth at least $p + 1$) in Section~\ref{sec:overview}. Next, we present results for $\ell_1$ plus $\ell_\infty$-norm bounds in Section~\ref{sec:panda}. Finally, we survey related work in Section~\ref{sec:related} and discuss limitations and conclude with open questions in Section~\ref{sec:concl}. For the sake of readability, all proofs have been deferred to the appendix.

%% file: table_with_stats.tex
\begin{table*}[th!]
{
\small
\centering
{

	\hspace{0.5cm} {
	\begin{tabular}{c|c|c|c} 
	Dataset & Number of Edges & AGM-$\ell_2$ Bound Ratio & Deg Bound-$\ell_2$ Ratio \\
	\hline
	 ca-GrQc &  $28980$ & $10.1$ & $4.8$ \\
	ca-HepTh & $51971$ & $18.19$ & $5.19$ \\
	facebook-combined & $88234$ & $3.26$ & $2.75$ \\
	soc-Epinions1 & $508837$ & $6.63$ &  $6.63$ \\
	\hline
\end{tabular}
}
\caption{\small Comparison of our $\ell_2$-norm bounds with AGM and degree-based bound when $G$ is a triangle: The first column denotes the SNAP dataset~\cite{snap}, we chose a subset of benchmarks based on the number of edges (denoted by second column). The third and final columns denote the ratio between the tight combinatorial bounds we obtain for the $\ell_2$-norm ($L^2$)  -- $(1)$ the AGM bound $(N^{3/2})$~\cite{AGM} and $(2)$ the degree-based bound $\min(Nd, N^{3/2})$~\cite{panda}. Here, the $\ell_1$-norm bound is denoted by $N$ and the $\ell_{\infty}$ bound is denoted by $d$. 
} \label{table:joins-with-stat}
} 
}
\end{table*}

%% file: generic_algorithm.tex
\section{Our Algorithm}  \label{sec:tech}

In this section, we will present our generic algorithm for any query graph $G = (V, E)$, starting with some notation.
\subsection{Notation} \label{sec:sec2-notation}
Recall that $G$ corresponds to a join query. Given a database instance $I=\inset{R_{\edge{v}{u}}}_{\edge{v}{u} \in E}$, we denote the join output for $I$ by
\[\TJ = \ujoin{e = \edge{v}{u} \in E} R_{\edge{v}{u}}.\] 
A \textit{\textbf {degree configuration}} $\bd = (d_{\edge{v}{u}})_{\edge{v}{u} \in E}$ is where each value $d_{\edge{v}{u}}$ is a power of two that is at most $2 \cdot \norm{R_{\edge{v}{u}}}_\infty$. We define a subrelation $R^{d}_{\edge{v}{u}}$ of $R_{\edge{v}{u}}$ for any $d$ that is a power of two, as follows:
\begin{align*}
R_{\edge{v}{u}}^{d} =  \left \{ \mathbf{t}: \mathbf{t} \in  R_{\edge{v}{u}}, \frac{d}{2} < \deg_{R_{\edge{v}{u}}}(\pi_{v}(\vt)) \le d  \right \} \numberthis \label{eq:subrelation-deg},
\end{align*}
where $\pi_{v}(\vt)$ is the value in $\vt$ corresponding to $v$ and $\deg_{R_{\edge{v}{u}}}(x)$ denotes the degree of the value $x\in\Dom(v)$ in $R_{\edge{v}{u}}$.
We denote the join output for subrelations corresponding to degrees in $\bd$ by 
\[\TJ(\bd)= \ujoin{e = \edge{v}{u} \in E} R_{\edge{v}{u}}^{d_{\edge{v}{u}}}.\] 
We define $|V| = n$ and we would like to recall that the acyclicity of $G$ is defined in the directed sense (i.e., $G$ is a Directed Acyclic Graph or DAG). For any integer $m \ge 1$, $[m]$ denotes the set $\{1, \dots, m\}$. Throughout the paper, we will assume that the degree values $d_{\edge{v}{u}}$ for every $\edge{v}{u} \in E$ are powers of two.

\subsection{Our Algorithm} \label{sec:sec2-generic}
In this section, we present our algorithm with a running example of $G$ being a triangle  with $V = \{ A, B, C\}$, $E = \{ \edge{A}{B}, \edge{B}{C}, \edge{C}{A} \}$ and relations $R_{\edge{A}{B}}, S_{\edge{B}{C}}$ and $T_{\edge{C}{A}}$ (note that $G$ is acyclic). Our goal is to compute $\TJ(\bd)$ and we do so in three stages\footnote{We would like to note here that this structure was introduced in~\cite{JR-paper}.}. 

\begin{Example} \label{ex:triang-acyclic}
Consider the acyclic subquery with $R_{\edge{A}{B}}$ and  $S_{\edge{B}{C}}$ and let $\bd = (d_{\edge{A}{B}},d_{\edge{B}{C}})$ be a degree configuration for this subquery. Our goal here is to compute $J_{A, B, C} = R^{d_{\edge{A}{B}}}_{\edge{A}{B}} \Join S^{d_{\edge{B}{C}}}_{\edge{B}{C}}$.
 
We use the well-known algorithm $\lftj$~\cite{V14} for this purpose and start by considering the topological ordering $(A, B, C)$ of vertices. The set of values of $A$ in $\pi_{A}(J_{A, B, C}(\bd))$ is a subset of $\pi_{A} (R^{d_{\edge{A}{B}}}_{\edge{A}{B}})$. For a fixed $a \in \Dom(A)$, the set of values of $B$ is the intersection of\footnote{Here, $\sigma_{A = a}(R^{d_{\edge{A}{B}}}_{\edge{A}{B}})$ denotes the set of tuples in $R^{d_{\edge{A}{B}}}_{\edge{A}{B}})$ where the value corresponding to $A$ is $a$~\cite{relational-algebra}.} $\pi_{B}(\sigma_{A = a}(R^{d_{\edge{A}{B}}}_{\edge{A}{B}}))$  and $\pi_{B}(S^{d_{\edge{B}{C}}}_{\edge{B}{C}})$. Finally, for fixed values $(a, b) \in \Dom(A) \times \Dom(B)$, the set of values of $c\in\Dom(C)$ is $\pi_{C}(\sigma_{B = b}(S^{d_{\edge{B}{C}}}_{\edge{B}{C}}))$ where $B$ has value $b$. Taking the union over all such triples $(a,b,c)$ gives us $J_{A, B, C}(\bd)$, as required. \demo
\end{Example}
It turns out that we can extend the above algorithm to any acyclic $G$ (following similar ideas in~\cite{ngo-survey}), where we consider a topological ordering and for each vertex, we compute an intersection of its projections on all its `incoming' and `outgoing' relations on the vertex. Then, we take a Cartesian product of this intersection with the set of tuples computed so far and continue this process until the join output is computed.
\begin{breakablealgorithm}
	\caption{Prefix-Join $(G,\bd,I)$} \label{algo:gen-ub}
	\begin{algorithmic}[1]
		\Require{Directed Acyclic Query graph $G=(V,E)$ (such that the undirected version is connected); degree configuration $\bard = (d_{\edge{v}{u}})_{\edge{v}{u} \in E}$; database instance $I=\set{R_e}_{e\in E}$.}
		\Ensure{$\TJ(\bd)$}
		\Statex
		\State $(u_1,\dots,u_n) \gets \toporder(G)$ 
		\State $J_1 \gets  \underset{\edge{u_1}{w}\in E}{\bigcap} \pi_{u_1}\left(R_{\edge{u_1}{w}}^{d_{\edge{u_1}{w}}} \right)$ \Comment{Since $u_1$ is a source, it has only outgoing relations.}
		\For{$i=2 \dots n$}
		\State $J_i\gets\emptyset$
		\State $P_{\text{out}}(i)\gets \underset{\edge{u_i}{w}\in E}{\bigcap} \pi_{u_i}\left(R_{\edge{u_i}{w}}^{d_{\edge{u_i}{w}}} \right)$
		\ForAll{$\vt\in J_{i-1}$}
		\State $P_{\text{in}}(i,\vt)\gets \underset{\edge{v}{u_i}\in E}{\bigcap} \inset{y: (\vt[v],y)\in R_{\edge{v}{u_i}}^{d_{\edge{v}{u_i}}}}$ 
		\State $P_i(\vt)\gets P_{\text{in}}(i,\vt) \cap P_{\text{out}}(i)$
		\State $J_{i}\gets J_{i} \cup \inset{\vt}\times P_i(\vt)$ \Comment{$\inset{\vt}\times P_i(\vt)$ is the set of all valid extensions of $\vt$ in $J_i$}
		\EndFor
		\EndFor
		\State \Return{$J_{n}$}
	\end{algorithmic}
\end{breakablealgorithm}
The correctness of Algorithm~\ref{algo:gen-ub} follows directly from the correctness of $\lftj$ -- the proof is by induction on $i \in [2, n]$. In particular, at the end of iteration $i$, the set $J_i$ is indeed the join $\ujoin{\edge{v}{u} \in E} R_{\edge{v}{u}}^{d_{\edge{v}{u}}}$ {\em projected down} to $\set{u_1,\dots,u_i}$. 

Next, we extend Algorithm~\ref{algo:gen-ub} to compute $\TJ(\bd)$ for any (potentially cyclic) $G$ via a simple generalization: run Algorithm~\ref{algo:gen-ub} on all spanning acyclic subgraphs of $G$ in `parallel' and stop once we have completely processed the first spanning acyclic subgraph. The details are in Algorithm~\ref{algo:acyclic-ub}.
\begin{breakablealgorithm}
	\caption{Acyclic-Join $(G, \bd, I)$} \label{algo:acyclic-ub}
	\begin{algorithmic}[1] 
		\Require{Directed Query graph $G=(V,E)$; Database instance $I=\set{R_e}_{e\in E}$; Degree configuration $\bard = (d_{\edge{v}{u}})_{\edge{v}{u} \in E}$.}
		\Ensure{$\TJ(\bd)$}
		\State $\TJ(\bd)\gets \emptyset$
		\ForAll{Spanning Acyclic Subgraphs of $G$}
		\Comment{We start all these runs in parallel and when the first one terminates, we set $\TJ(\bd)$ as its output.}
		\State Let the spanning acyclic subgraph under consideration be $\set{G_i}_{i\in [t]}$ for some $t>0$, where each $G_i$ is connected (in the undirected sense).
		\State $J\left(\set{G_i}_{i\in [t]}, \bard, I \right) \gets \times_{i \in [t]} \mathrm{Prefix-Join}(G_i, \bard,I)$
		\EndFor
		\State Let $\set{G_i^*}_{i\in [t^*]}$ be the acyclic subgraph that finishes first.
		\State Prune $J\left(\set{G_i^*}_{i\in [t^*]}, \bard, I \right)$ against all relations to get the final $\TJ(\bd)$. \Comment{Retain $\vt$ only if $\pi_e(\vt)\in R_e$ for every $e\in E$.}
		\label{line:pruning}
		\State \Return{$\TJ(\bd)$}
	\end{algorithmic}
\end{breakablealgorithm}
We note that we can run Algorithm~\ref{algo:acyclic-ub} on a standard `serial' machine with the known trick of multiplexing the runs of Algorithm~\ref{algo:gen-ub} on all spanning acyclic subgraphs (e.g., by running one iteration of Algorithm~\ref{algo:gen-ub} at a time) and to stop once a spanning acyclic subgraph finishes. The correctness of Algorithm~\ref{algo:acyclic-ub} follows from the fact that we only consider spanning subgraphs and the correctness of Algorithm~\ref{algo:gen-ub} (as well as the pruning step in Line~\ref{line:pruning}).

Finally, to compute $\TJ$, we simply run Algorithm~\ref{algo:acyclic-ub} on all possible degree configurations:
\begin{breakablealgorithm}
	\caption{Forward-Join $(G,I)$}  \label{algo:generic-ub}
	\begin{algorithmic}[1]
		\Require{Directed Query graph $G=(V,E)$; database instance $I=\set{R_e}_{e\in E}$.}
		\Ensure{$\TJ$}
		\State{$J \gets \emptyset$}
		\ForAll{degree configurations $\bard = (d_{\edge{v}{u}})_{\edge{v}{u} \in E}$} 
		\Comment{$d_{\edge{v}{u}}$ runs over all powers of $2$ until $\norm{R_{\edge{v}{u}}}_\infty$.}
		\State{$J\gets J\cup  \mathrm{Acyclic-Join}(G, \bard,I)$}
		\EndFor
		\State \Return{$J$}
	\end{algorithmic}
\end{breakablealgorithm}
We would like to stress that our overall algorithm {\em does not use {\bf any} information about $p$ or the corresponding $\ell_p$ norm bounds.} This information is {\em only used} in its runtime analysis and our algorithm works {\em simultaneously} for all norm bounds.
By contrast, when we adapt $\panda$ to our setup (i.e., use it 
in place of Algorithm~\ref{algo:gen-ub})\footnote{We discuss this in detail in Appendix~\ref{sec:panda-algo}.}, $\panda$ {\em does} need to know the $\ell_p$ norm bounds (as well as the value of $p$). More importantly, if we do so, we will be losing a practically prohibitive multiplicative factor of $O\inparen{\inparen{\log{N}}^{\inparen{\inparen{2^{|V|}}!}^2}}$ in the process as well.

Before discussing how we avoid this multiplicative factor in the runtime analysis of our algorithm, we would like to mention here that from a practical point of view, it makes sense to have a simple algorithm even though its proof of runtime/worst-case optimality could be more technically involved (since the latter is just for the analysis). Finally, a natural question that can arise here is if we need to run Algorithm~\ref{algo:acyclic-ub} on all acyclic spanning subgraphs of $G$ and in our proofs, we show how we can pick a very specific class of acyclic spanning subgraphs that achieve the worst-case size bound (though this needs the knowledge of the norm  bounds).

\subsection{Worst-Case Optimality of Generic Algorithm} \label{sec:sec2-runtime}
To analyze our algorithm, we take our running example of $G$ being a (directed) triangle. Recall that we have $V = \{ A, B, C\}$ and $E = \{ \edge{A}{B}, \edge{B}{C}, \edge{C}{A} \}$ with relations $R_{\edge{A}{B}}, S_{\edge{B}{C}}$ and $T_{\edge{C}{A}}$\footnote{We note here that if flipped the direction of the edge $\edge{C}{A}$, then we can prove an upper bound of $L^2$ using Cauchy-Schwarz (a fact we show in Appendix~\ref{app:cauchy}). However, we have not been able to extend this argument to other $G$.}. For ease of exposition, we assume the $\ell_2$-norm case and $\umax{\edge{v}{u} \in E}\norm{R_{\edge{v}{u}}}_{2} \le L$. Recall that we claimed an upper bound of $L^2$ for this case, which we prove in the following example (we ignore a constant factor of $4$ in our analysis).
\begin{Example} \label{ex:triang-cyclic}
Consider the acyclic subquery $R_{\edge{A}{B}} \Join S_{\edge{B}{C}}$ and let $\bd = (d_{\edge{A}{B}},d_{\edge{B}{C}})$ be a degree configuration for this subquery. We compute this subquery using the algorithm we describe in Example~\ref{ex:triang-acyclic} (which is essentially Algorithm~\ref{algo:gen-ub}). We recall the topological ordering $(A, B, C)$ of vertices and note that there are $\frac{L^2}{d_{\edge{A}{B}}^2}$  choices for $a\in\pi_A\left(R_{\edge{A}{B}}^{d_{\edge{A}{B}}}\right)$.\footnote{This is because each such value of $a$ has degree at least $d_{\edge{A}{B}}$ in $R_{\edge{A}{B}}^{d_{\edge{A}{B}}}$.} In Algorithm~\ref{algo:gen-ub}, this corresponds to bounding the size of $J_1$. Then, for each such choice of $a$, we have at most $d_{\edge{A}{B}}$ choices for $b$ and for each pair of choice of $(a,b)$, we have at most $d_{\edge{B}{C}}$ choices for $c$ in the output tuple $(a,b,c)$ (these correspond to bounding the size of $P_{\text{in}}(2,(a))$ and $P_{\text{in}}(3,(a,b))$ in Algorithm~\ref{algo:gen-ub} respectively). This implies we have an upper bound of
\[\mathcal{B}_{\neg\edge{A}{C}} =\frac{L^2}{d_{\edge{A}{B}}^2}\cdot d_{\edge{A}{B}}\cdot d_{\edge{B}{C}} 
 = L^2\frac{d_{\edge{B}{C}}}{d_{\edge{A}{B}}}.\]
We remark that the above is the bound on the size of the output of Algorithm~\ref{algo:gen-ub} for the subquery $R_{\edge{A}{B}}\Join S_{\edge{B}{C}}$ (for the degree configuration $(d_{\edge{A}{B}},d_{\edge{B}{C}})$).

Similarly, we get the following bounds from the subqueries $T_{\edge{C}{A}}\Join R_{\edge{A}{B}}$ and $S\edge{B}{C}\Join T\edge{C}{A}$, to get two more bounds:
\begin{align*}
\mathcal{B}_{\neg\edge{B}{C}} = L^2\frac{d_{\edge{A}{B}}}{d_{\edge{C}{A}}},
\mathcal{B}_{\neg\edge{A}{B}} = L^2\frac{d_{\edge{C}{A}}}{d_{\edge{A}{B}}}.
\end{align*}
While none of these three bounds by themselves are enough for all degree configurations, we can take their {\em minimum} (i.e. $\min\inset{\mathcal{B}_{\neg\edge{A}{C}}, \mathcal{B}_{\neg\edge{B}{C}}, \mathcal{B}_{\neg\edge{A}{B}}}$) to be the final bound, which can be bounded as:
\begin{align*}
& \min\inset{L^2\frac{d_{\edge{B}{C}}}{d_{\edge{A}{B}}}, L^2\frac{d_{\edge{A}{B}}}{d_{\edge{C}{A}}}, L^2\frac{d_{\edge{C}{A}}}{d_{\edge{A}{B}}} }\\
& \le \sqrt[3]{L^2\frac{d_{\edge{B}{C}}}{d_{\edge{A}{B}}}\cdot L^2\frac{d_{\edge{A}{B}}}{d_{\edge{C}{A}}}\cdot L^2\frac{d_{\edge{C}{A}}}{d_{\edge{A}{B}}}} \\
& =L^2,
\end{align*}
as claimed.\footnote{We still need to sum this bound up over all possible degree configurations but two things come to our aid here -- (a) considering the $\ell_2$-norm bound corresponding to the degree bucket and (b) exploiting the fact that the degrees are powers of two.}
We remark that the above bound is a valid upper bound for the size of the output of Algorithm~\ref{algo:acyclic-ub} for the triangle query and (any) degree configuration $\bd$. \demo
\end{Example}

Next, we discuss how to extend our techniques for general $G$, which involves more technical work than the above example. Our proofs consist of two broad parts:
(i)  Showing that picking the best acyclic sub-queries still allows us to prove an optimal bound.
(ii) Avoid paying the multiplicative $O\inparen{\log{N}^{\inparen{\inparen{2^{|V|}}!}^2}}$  factor that $\panda$ pays.

We start with (ii): it turns out that since our algorithm is simple, we can bound its runtime via either a simple self-contained expression (or) via a simple generalization of the edge covering LP (that is used to prove the AGM bound) to the case of $\ell_p$ bound:
\begin{align*}
& \min \quad \usum{\edge{v}{u} \in E} \left( x_{\edge{v}{u}} \log(L_{\edge{v}{u}}) ) \right)\tag{$\LP^{(+)}$}\\
& \usum{\edge{v}{u} \in E} \left( x_{\edge{v}{u}} \right) + \usum{\edge{u}{w} \in E} \frac{x_{\edge{u}{w}}}{p} \ge 1 \quad \forall u \in V  \numberthis \label{eq:covering} \\
& x_{\edge{v}{u}} \ge 0 \quad \forall \edge{v}{u} \in E \numberthis \label{eq:primal}.
\end{align*}

It turns out that if we considered the obvious bound for each {\em degree configuration}, then we can upper bound the join output size using the LP above: i.e., for any degree configuration, the size of the output of Algorithm~\ref{algo:acyclic-ub} is bounded by $2^{\LP^{(+)}}$. However, if we summed this worst-case bound over all degree configurations, we will suffer an $O\inparen{\inparen{\log{N}}^{|E|}}$ multiplicative factor loss, which is much better than the multiplicative $O\inparen{\inparen{\log{N}}^{\inparen{\inparen{2^{|V|}}!}^2}}$ factor of $\panda$ but still not ideal. Since we have a simple closed form expression for the sub-join query for each degree configuration, we can apply H\"{o}lder's inequality to simply `push in' the sums to sum up the $\ell_p$ bounds, which allow us to get rid of the multiplicative $O\inparen{\log{N}^{|E|}}$ factor. 

Finally, for (i), we exploit the fact that any optimal basic feasible solution to our LP implies that we only need to consider very specific classes of join queries. E.g., in~\cite{NPRR}, it was shown that for the simple join query graph case (with $p=1$), we only need to handle cycles and stars. For the setting of $\panda$, we have to put in a bit more effort and handle the case of $d\le \sqrt{L}$ (in which case our LP does not provide the correct upper bound) and the case of $d>\sqrt{L}$ (when our LP is indeed a valid upper bound) separately. 

\subsection{Question~\ref{ques:acyclic} for General $\ell_{p}$ Constraints} \label{sec:sec2-acyclic}
As an interesting by-product of our proofs, we answer Question~\ref{ques:acyclic} in the affirmative for the case when $\ell_{1}$ and $\ell_{\infty}$ bounds are the same. A natural extension to consider is if it holds for $\ell_2$ (and $\ell_p$ more generally), which gets more interesting. The answer to Question~\ref{ques:acyclic} for $\ell_2$ is {\em no} (we present an example in Appendix~\ref{sec:app-acyclic-subgraph}). As a consequence, in Algorithm~\ref{algo:generic-ub}, we make the choice of spanning acyclic subgraph depending on the degree configuration and show that this is sufficient to prove tight bounds for the $\ell_p$ norm case for $p\in (1,2]$.

%% file: prelims_and_notation.tex
\section{Further Preliminaries and Notation} \label{sec:norm-prelims}
We present most relevant preliminaries here and defer a detailed version to Appendix~\ref{app:prelims}. 

For each subrelation $R_{\edge{v}{u}}^{d_{\edge{v}{u}}}$ from~\eqref{eq:subrelation-deg} for every $\edge{v}{u} \in E$, $d_{\edge{v}{u}} \le L_{\edge{v}{u}}$, we define
\begin{align*}
L_{\edged{v}{u}} & \gets \norm{R_{\edge{v}{u}}^{d_{\edge{v}{u}}}}_{p} \numberthis \label{eq:l-edged}. 
\end{align*}
We will be using the following structural result (the proof is in Appendix~\ref{sec:norm-structural}) to prove our upper bound. We first define the notion of a basic feasible solution.
\begin{defn} [Basic Feasible Solution to LP~\eqref{eq:primal}~\cite{bf-book}] \label{def:basic-feasible-main}
A basic feasible solution $\mathbf{x} = \left(x_{\edge{v}{u}} \right)_{\edge{v}{u} \in E}$ to $\LP^{(+)}$ is one that satisfies all its $|V| + |E|$ constraints with at least $|E|$ of them satisfied with equality (ones we call {\em tight}). Let $C$ denote the  $(|V| + |E|)\times |E|$ constraint matrix, where the rows are indexed by constraints and columns are indexed by variables and let $S$ denote the set of tight constraints. Then, the matrix $C$ projected down to rows (i.e., constraints) in $S$ has rank exactly $|E|$. 
\end{defn}
\begin{thm} [Based on~\cite{NPRR}] \label{thm:opt-decomp}
For any directed graph $G = (V, E)$, there exists an optimal solution $\mathbf{x}^* = \left( x^*_{\edge{v}{u}} \right)_{\edge{v}{u} \in E}$ to $\LP^{(+)}$, and a $t$ such that $G$ can be decomposed into a disjoint union of $t$ connected components (in the undirected sense) $G_i = (V_i, E_i)$ with 
$|V_i| - 1 \le |E(G_i)| \le |V_i|$, where  $x^*_{\edge{v}{u}} > 0 \quad$ for all  $\edge{v}{u} \in E(G_i)$.
and $\mathbf{x}^*_{i} = \left( x^*_{\edge{v}{u}} \right)_{\edge{v}{u} \in E(G_i)}$ is an optimal basic feasible solution for $\LP^{(+)}$ on $G_i$ for every $i \in [t]$. Further, we have $\cup_{i=1}^{t}V(G_i) = V$ and $V(G_i) \cap V(G_j) = \emptyset$ \space $\forall i, j \in [t], i \neq j$. The following is true:
\[\TJ = \times_{i \in [t]} J^{(I)}_{G_i}.\]
\end{thm}
\mypar{Other assumptions and notation} 
For all our algorithms, we assume that each relation $R_{e}$ is stored in a two level B-tree-like index structure~\cite{Btree-ref} in the ordering $\edge{v}{u}$ (see Appendix~\ref{app:data-struct}). Note that the total time of construction of these B-trees for each relation is $O(|E| L_{\edge{v}{u}} \log(L_{\edge{v}{u}}))$. We assume the RAM model of computation i.e., elements in the B-tree can be accessed in constant time. Throughout the paper, we assume that all the logarithms are base $2$ unless specified otherwise.

%% file: proof_for_worst_case_optimality.tex
\section{Worst-case Optimality of Algorithm~\ref{algo:generic-ub} for $\ell_p$-norm bounds} \label{sec:overview}
In this section, our goal is to prove that Algorithm~\ref{algo:generic-ub} is worst-case optimal (under $\ell_p$ norm bounds) for any $G$ with girth at least $p + 1$. 

We discuss briefly here about our restriction on the girth being at least $p + 1$  and it mainly has to do with our $LP$ techniques not providing the optimal lower bound in this scenario. To illustrate this, consider our running example of $G$ being a triangle with a cyclic orientation, $p = 3$, and all $\ell_3$-norm bounds are upper bounded by $L$. In this case, we get a lower bound of $L^{\frac{9}{4}}$ using our $\LP$-based result; \footnote{Based on the dual of $\LP^{(+)}$~\eqref{eq:primal} with $p = 3$, the optimal values for $y_u$ is when all $y_u$ values are all equal. For the case of $p=3$, we have that $|\Dom(u)| = L^{\frac{3}{4}}$, yielding a final bound of $L^{\frac{9}{4}}$} However, we can get a (trivial) lower bound of $L^3$ on $|J_{G}^{(I)}|$ using the instance:
\begin{align*}
   ||R_{\edge{A}{B}}||_3 = ||R_{\edge{B}{C}}||_3 = ||R_{\edge{C}{A}}||_3 = \{(i, i) : i \in [L^3] \}
\end{align*}

We start by stating the following result (the proof is in Appendix~\ref{sec:gen-acyclic}) on the runtime of Algorithm~\ref{algo:gen-ub}.
\begin{thm} \label{claim:gen-acyclic}
For any acyclic $G$, $p \in [1, \infty]$ and any degree configuration
\begin{align*}
\bd = (d_{\edge{v}{u}})_{d_{\edge{v}{u}} \le \min\left(L_{\edge{v}{u}} L_{\edgeinfty{v}{u}}\right), \edge{v}{u} \in E},
\end{align*}
Algorithm~\ref{algo:gen-ub} computes $|\TJ(\bd)|$ in time linear in $\B(\bard, G)$, where
\begin{align*}
\B(\bard, G) = \left( \uprod{u \in V} \cD_{u}(\bard) \right). \numberthis \label{eq:gen-bdg-ub}	
\end{align*}
In the above, for each $u \in V$, $\cD_{u}(\bard)$ is defined as
\begin{align*}
\min \inset{  \umin{\edge{v}{u} \in E}\inset{d_{\edge{v}{u}}},  \umin{\edge{u}{w} \in E} \inset{\frac{2^{p} \cdot L_{\edged{u}{w}}^{p}}{d_{\edge{u}{w}}^{p}}}   }.
\end{align*}
 Further, we have 
\begin{align*}
 \left|\TJ(\bd) \right| \le \B(\bard, G).
 \end{align*}
\end{thm}
The proof follows by noting that in Algorithm~\ref{algo:gen-ub}, we proceed in a topological ordering of vertices in $V$. For each vertex $u \in V$, the size of projections of incoming relations to $u$ can be upper bounded by the smallest incoming degree and outgoing relations from $u$ can be upper bounded by an effective domain size based on the $\ell_{p}$-norm bound $L_{\edge{u}{w}}$ and outgoing degree (as shown in Example~\ref{ex:triang-cyclic}) respectively.

We use Theorem~\ref{claim:gen-acyclic} to prove the worst-case optimality of Algorithm~\ref{algo:generic-ub}. It turns out that in addition to Theorem~\ref{claim:gen-acyclic}, we also need a way to pick an acyclic spanning subgraph to reason about Algorithm~\ref{algo:acyclic-ub}. We do this in two steps -- we start with the case of $G$ being acyclic and then consider the general $G$ case with girth at least $p + 1$.
\subsection{$G$ is Acyclic} \label{sec:g-acyclic}
For acyclic $G$, we consider a slightly more general scenario, where in addition to the $\ell_{p}$-norm size bounds on each $\edge{v}{u} \in E$, we are given a $\ell_{\infty}$-norm constraint in the same direction (i.e., $\norm{R_{\edge{v}{u}}}_{\infty} \le L_{\edgeinfty{v}{u}}$). We state our primal $\LP^{(+)}$, 
which is a generalization of $\LP$~\eqref{eq:primal}). 
\begin{align*}
& \min \usum{\edge{v}{u} \in E} \left( x_{\edge{v}{u}} \log(L_{\edge{v}{u}}) + z_{\edge{v}{u}} \log(L_{\edgeinfty{v}{u}}) \right) \tag{$\LP^{(+)}$} \\
& \usum{\edge{v}{u} \in E} \left( x_{\edge{v}{u}} + z_{\edge{v}{u}} \right) + \usum{\edge{u}{w} \in E} \frac{x_{\edge{u}{w}}}{p} \ge 1 \quad \forall u \in V  \numberthis \label{eq:gen-covering} \\
& x_{\edge{v}{u}}, z_{\edge{v}{u}} \ge 0 \quad \forall \edge{v}{u} \in E \numberthis \label{eq:gen-primal}.
\end{align*}
In the LP above, $x_{\edge{v}{u}}$ coresponds to the $\ell_{p}$ constraints and $z_{\edge{v}{u}}$ corresponds to $\ell_{\infty}$ constraints for every $\edge{v}{u} \in E$ (we will consider the same setting in Section~\ref{sec:panda} as well). We state the following result based on $\B(\bard, G)$ (from~\eqref{eq:gen-bdg-ub}) and $\LP^{(+)}$ defined above.
\begin{lemm} \label{thm:acyclic-bound-lp}
For any acyclic $G$, any feasible solution $(\vx, \vz) = (x_{\edge{v}{u}}, z_{\edge{v}{u}})_{\edge{v}{u} \in E}$ to $\LP^{(+)}$ on $G$ and degree configuration $\bd$, we have 
\begin{align*}
\B(\bard, G)  \le 2^{p |V|} \cdot \uprod{\edge{v}{u} \in E}   \left( d_{\edge{v}{u}}^{z_{\edge{v}{u}}} \cdot L_{\edged{v}{u}}^{x_{\edge{v}{u}}} \right).
\end{align*}
\end{lemm}
The proof proceeds by upper bounding the $\min$ in the term $\cD_{u}(\bard)$ for every $u \in V$ using a product where the exponents on the terms come from~\eqref{eq:gen-covering} as follows:
\begin{align*} 
\B(\bard, G) & = \uprod{u \in V} \min \left(  (d_{\edge{v}{u}})_{\edge{v}{u} \in E},  \left( \frac{2^{p} \cdot L_{{\edged{u}{w}}}^{p}}{d_{\edge{u}{w}}^{p}} \right)_{\edge{u}{w} \in E}   \right) \\
& \le 2^{p |V|} \cdot \uprod{u \in V} \left( \left( \uprod{\edge{v}{u} \in E} d_{\edge{v}{u}}^{x_{\edge{v}{u}} + z_{\edge{v}{u}}} \right) \cdot \left(\uprod{\edge{u}{w} \in E} \left( \frac{L_{{\edged{u}{w}}}^{p}}{d_{\edge{u}{w}}^{p}} \right)^{\frac{x_{\edge{u}{w}}}{p}}  \right) \right),
\end{align*}
which eventually proves Lemma~\ref{thm:acyclic-bound-lp} (details are deferred to Appendix~\ref{app:acyclic-bound-lp}). Recall that Algorithm~\ref{algo:generic-ub} computes $\TJ$ as the union of $\TJ(\bd)$ across all possible degree configurations $\bd$. We are now ready to state theorem for DAGs:
\begin{thm} \label{thm:gen-main}
For any DAG $G$ with $p \le |V| - 1$ and an optimal solution $(\vx^*, \vz^*) = \left(x^*_{\edge{v}{u}}, z^*_{\edge{v}{u}} \right)_{\edge{v}{u} \in E}$ to $\LP^{(+)}$ on $G$, Algorithm~\ref{algo:generic-ub} computes $\TJ$ in time linear in
\begin{align*}
2^{(p + 1) |V|} \cdot \uprod{\edge{v}{u} \in E} \left(L_{\edgeinfty{v}{u}}^{z^*_{\edge{v}{u}}} \cdot L_{\edge{v}{u}}^{x^*_{\edge{v}{u}}} \right) \numberthis \label{eq:gen-ub}
\end{align*}
for instances $\cI = \{R_{\edge{v}{u}}: ||R_{\edge{v}{u}}||_{p} \le L_{\edge{v}{u}}, ||R_{\edge{v}{u}}||_{\infty} \le L_{\edgeinfty{v}{u}}, \edge{v}{u} \in E\}$. Further, $|\TJ|$ is at most~\eqref{eq:gen-ub}. Finally, there {\em exists} an instance $I \in \cI$
such that
\begin{align*}
\left|\TJ \right| & \ge \frac{1}{2^{|V|}} \cdot \uprod{\edge{v}{u} \in E} \left( L_{\edgeinfty{v}{u}}^{z^*_{\edge{v}{u}}} \cdot L_{\edge{v}{u}}^{x^*_{\edge{v}{u}}}  \right) \numberthis \label{eq:gen-lb}.
\end{align*}
\end{thm}
Note that our upper and lower bounds differ by a factor of $2^{(p + 2) |V|}$ for $p \le |V| - 1$. The proof of~\eqref{eq:gen-ub} follows by a direct application of H\"{o}lder's inequality~\cite{ineq} using the fact that $\usum{\edge{v}{u} \in E} \left( z_{\edge{v}{u}} + \frac{x_{\edge{v}{u}}}{p} \right)$ is always at least $1$ as long as $p \le |V| - 1$ and is deferred to Appendix~\ref{sec:gen-deg-main}.\footnote{Our results hold for $p$ in $|V| - 1 < p < \infty$ as well with a slightly worse gap of $2^{p |V|}  \inparen{(p |E|)^2}^{(p |E|)^2} \cdot c^{|E|}$. The proof is deferred to Appendix~\ref{sec:genp-ub}.} Since $G$ is acyclic and connected (in the undirected sense), it is an spanning acyclic subgraph (by definition), which implies we can run Algorithm~\ref{algo:gen-ub} directly on $G$ (skipping Algorithm~\ref{algo:acyclic-ub}). Finally, the proof of~\eqref{eq:gen-lb} is similar to the proof of the AGM bound~\cite{AGM} (see Appendix~\ref{sec:gen-lb}).

%% file: cycle_main.tex
\subsection{$G$ has girth at least $p + 1$} \label{sec:lp-main}
In this section, we consider general $G$ with girth at least $p + 1$. We are now ready to state our main theorem. The LP we consider here is $\LP^{(+)}$ from Section~\ref{sec:sec2-runtime}.
\begin{thm} \label{thm:lp-main}
For any $G$ with girth at least $p + 1$ and an optimal solution $\vx^* = (x^*_{\edge{v}{u}})_{\edge{v}{u} \in E}$ to $\LP^{(+)}$ on $G$, Algorithm~\ref{algo:generic-ub} computes $\TJ$ in time linear in
\begin{align*}
2^{(p + 1) |V|} \cdot \uprod{\edge{v}{u} \in E}  L_{\edge{v}{u}}^{x^*_{\edge{v}{u}}} \numberthis \label{eq:lp-ub}
\end{align*}
for instances $\cI = \{R_{\edge{v}{u}}: ||R_{\edge{v}{u}}||_{p} \le L_{\edge{v}{u}}, \edge{v}{u} \in E\}$. Further, $|\TJ|$ is at most~\eqref{eq:lp-ub}. Finally, there {\em exists} an instance $I \in \cI$
such that
\begin{align*}
\left|\TJ \right| & \ge \frac{1}{2^{|V|}} \cdot \left( \uprod{\edge{v}{u} \in E} L_{\edge{v}{u}}^{x^*_{\edge{v}{u}}} \right) \numberthis \label{eq:lp-lb}.
\end{align*}
\end{thm}
The proof of~\eqref{eq:lp-lb} is in Appendix~\ref{sec:gen-lb}. In order to prove~\eqref{eq:lp-ub}, we proceed as follows -- we first invoke Theorem~\ref{thm:opt-decomp} on $G$ and then process its connected components one-by-one. If $G_i$ is a DAG, then we can invoke Theorem~\ref{thm:gen-main} directly. Otherwise (i.e., $G_i$ is cyclic), we prove an alternative version of Lemma~\ref{thm:acyclic-bound-lp}, where we successively upper bound $\B(\bard, G)$ by a sequence of three LPs (with the last one matching $\LP^{(+)}$). The first LP we consider is a natural relaxation of $\B(\bard, G)$ and the second LP is the dual of $\LP^{(+)}$ and the third/final LP is $\LP^{(+)}$. 
 Finally, to prove our upper bound and compute $\TJ$, we consider specific spanning acyclic subgraphs of each cyclic $G_i: i \in [t]$ (note that since $E_i \le |V_i|$, it can have at most one cycle).  We defer the proof to Appendix~\ref{sec:lp}.


%% file: panda_main.tex
\section{Results on $\ell_{1}$ and $\ell_{\infty}$ for all $G$} \label{sec:panda}
In this section, we design worst-case optimal join algorithms for the case when we are given $\ell_{1}$ and $\ell_{\infty}$ constraints. We make the following assumption.
\begin{assumption} \label{assump:deg}
For each $\edge{v}{u} \in E$, we are given a $\ell_{1}$ bound $L$ (i.e., $|R_{\edge{v}{u}}| \le L$) and  a $\ell_{\infty}$ bound $d$ on $u$ of the form $\norm{R_{\edge{v}{u}}}_\infty \le d$.
\end{assumption}
We split our results into two categories based on the value of $d$ --  (i) $d^2 \le L$, which we tackle in Section~\ref{sec:dlow} and (ii)  $d^2 > L$, which we tackle in Section~\ref{sec:dhigh}. Our results focus on the cases where all $L$ and $d$ values are the same. We can generalize our arguments for $d^2 > L$ and $d^2 \le L$ to handle some special cases of distinct $L$ and $d$ values; however, we cannot handle the most general case of potentially different $\ell_1$ and $\ell_{\infty}$ bounds.

\subsection{Small degree bound: $d^2 \le L$} \label{sec:dlow}
\mypar{Preliminaries} Recall from Assumption~\ref{assump:deg} that $G$ is a directed graph with each edge $\edge{v}{u}$ has $\ell_{\infty}$-norm constraints of the form $L_{\edgeinfty{v}{u}}$. Now, we decompose vertices in $V(G)$ into four buckets -- $(1)$ Set of non-trivial source Strongly Connected Components (SCCs) (i.e., SCCs with at least two vertices, one of which is a source) $C(G)$, $(2)$ The remaining sources (which are SCCs with one vertex) $S(G)$, $(3)$ Set of vertices $T(G)$, where each vertex is connected by at least one vertex in $S(G)$ (through an incoming edge) and $(4)$ The remaining set of vertices $\rho(G)$.

We consider the induced subgraph on vertices $S(G) \cup T(G)$. We further partition $S$ into $S_1(G)$ and $S_2(G)$ and $T$ into $T_1(G)$ and $T_2(G)$ as follows. We choose a subset $E(S_1(G), T_1(G)) \subset E(G)$ to be a (disjoint) set of stars\footnote{A star with $n$ vertices is where one vertex has degree $n-1$ (which we call the {\em center}) and the remaining vertices have degree $1$ (which we call {\em leaves}).} with each $s_1 \in S_1(G)$ as the center (in the undirected sense) and each $t_1 \in T_1(G)$ s.t. $\edge{s_1}{t_1} \in E(S_1(G), T_1(G))$ (for the fixed $s_1$) as a leaf. Similarly, we define $E(S_2, T_2) \subset E$ to be another (disjoint) set of disjoint stars such with each $t_2 \in T_2(G)$ is a center (in the undirected sense) and each $s_2 \in S_2(G)$ with $\edge{s_2}{t_2} \in E(S_2(G), T_2(G))$ (for the fixed $t_2$) as a leaf. 
We pick $S_i(G),T_i(G)$ and $E(S_i,T_i)$ for $i\in [2]$ that {\em minimizes} the size of this star cover, i.e. minimizes $\abs{E(S_1,T_1)}+\abs{E(S_2,T_2)}=\abs{T_1(G)}+\abs{S_2(G)}$.
We discuss in Appendix~\ref{sec:agm-edge-cover} how to compute an optimal star cover of this kind using the AGM LP and also argue why minimizing the star cover size also minimizes our bound in Theorem~\ref{thm:dsmall}.

We are now ready to state our main theorem.
\begin{thm} \label{thm:dsmall}
For any $G$, $L$ and $d$ with $d^2 \le L$ satisfying Assumption~\ref{assump:deg} and an optimal star cover $E(S_1(G), T_1(G))$ and $E(S_2(G), T_{2}(G))$,  Algorithm~\ref{algo:generic-ub} runs in time linear in
\begin{align*}
2^{2(|V| + |C(G)| + |S_1(G)| + |T_1(G)|)} \left( \uprod{C_i \in C(G)} L d^{|C_i| - 2}  \right) \cdot \left(  \left(\frac{L}{d}\right)^{|S_1(G)|} \cdot d^{|T_1(G)|} \right) \cdot L^{|S_2(G)|} \cdot d^{|\rho(G)|} \numberthis \label{eq:dlow-ub} 
\end{align*}
for instances $\cI = \{R_{\edge{v}{u}}: ||R_{\edge{v}{u}}||_{1} \le L, ||R_{\edge{v}{u}}||_{\infty} \le d, \edge{v}{u} \in E\}$. Further, $|\TJ|$ is at most~\eqref{eq:dlow-ub}. Finally, there exists an instance $I \in \cI$ such that
\begin{align*}
|\TJ| &\ge \frac{1}{2^{|V|}} \left( \uprod{C_i \in C(G)} L d^{|C_i| - 2}  \right)  \cdot \left( \left(\frac{L}{d}\right)^{|S_1(G)|} \cdot d^{|T_1(G)|} \right) \cdot L^{|S_2(G)|} \cdot d^{|\rho(G)|} \numberthis \label{eq:dlow-lb}.
\end{align*}	
\end{thm}
We prove~\eqref{eq:dlow-lb} using a hard instance with two types of embeddings -- one for relations in $C(G)$ and the other for relations in the remaining graph. For relations in $C(G)$, we use an embedding with a disjoint union of Cartesian product-based instances (where each component has degree $O(d)$) and for the remaining relations, we use a single Cartesian product-based instance (where incoming degree is at most $d$). Details are deferred to Appendix~\ref{app:dlow-lb}. 

For proving~\eqref{eq:dlow-ub}, we construct a specific spanning acyclic subgraph of $G$ that achieves~\eqref{eq:dlow-ub} as follows -- for each non-trivial source SCC, we fix an arbitrary edge $\edge{v}{u}$ (and drop all incoming edges to $v$ and $u$ (except from $v$). We proceed in the standard topological ordering (note that the induced subgraph on $(S(G), T(G))$ is acyclic) and if we hit a vertex in a non-trivial non-source SCC with back edges (i.e., edges from the SCC), we drop all the back edges and continue the process. Our upper bound for a fixed degree configuration $\bd$ then follows by upper bounding the effective domain sizes of the source vertices $s \in V(G)$, $|\Dom_{\bd}(s)|$ by $\umin{\edge{s}{w} \in E} \frac{L_{\edge{s}{w}}}{d_{\edge{s}{w}}}$ and the remaining vertices $u \in V \setminus \ucup{\text{sources } s \in V(G)} \{s\} $, $\Dom_{\bd}(u)$ by $\umin{\edge{v}{u} \in E} d_{\edge{v}{u}}$ and taking their product. Summing these bounds over all degree configurations, we get our upper bound as required. Details are deferred to Appendix~\ref{app:dlow-ub}. 

\subsection{Large Degree Bound: $d^2 > L$} \label{sec:dhigh}
We first state $\LP^{(+)}$ for this scenario (which is essentially $\LP$~\eqref{eq:gen-primal} from Section~\ref{sec:g-acyclic} for the case when $p = 1$), where variables $x_{\edge{v}{u}}$ correspond to the $\ell_1$-norm bounds and variables $z_{\edge{v}{u}}$ correspond to the $\ell_{\infty}$-norm bounds.
\begin{align*}
& \min  \usum{\edge{v}{u} \in E}\left( x_{\edge{v}{u}} \log(L) + z_{\edge{v}{u}} \log(d) \right) \tag{$\LP^{(+)}$} \\
& \text{ s.t. } \usum{\edge{v}{u} \in E} \left( x_{\edge{v}{u}} + z_{\edge{v}{u}} \right) + \usum{\edge{u}{w} \in E} x_{\edge{u}{w}} \ge 1 \quad \forall u \in V \\
& x_{\edge{v}{u}}, z_{\edge{v}{u}} \ge 0 \quad \forall \edge{v}{u} \in E.
\end{align*}
We are now ready to state our main theorem.
\begin{thm} \label{thm:dhigh-main}
For any $G$, $L$ and $d$ with $d^2 > L$ satisfying Assumption~\ref{assump:deg}, Algorithm~\ref{algo:generic-ub} computes $\TJ$ in time linear in
\begin{align*}
2^{(p + 1) |V|} \left(\uprod{\edge{v}{u} \in E} L^{x^*_{\edge{v}{u}}} \cdot d^{z^*_{\edge{v}{u}}} \right) \numberthis \label{eq:dhigh-ub}.
\end{align*}
for instances $\cI = \{R_{\edge{v}{u}}: ||R_{\edge{v}{u}}||_{1} \le L, ||R_{\edge{v}{u}}||_{\infty} \le d, \edge{v}{u} \in E\}$. Further, $|\TJ|$ is at most~\eqref{eq:dhigh-ub}.  Finally, there exists an instance $I \in \cI$ such that
\begin{align*}
|\TJ| & \ge \frac{1}{2^{|V|}} \cdot \left( \uprod{\edge{v}{u} \in E} L^{x^*_{\edge{v}{u}}} \cdot d^{z^*_{\edge{v}{u}}} \right) \numberthis \label{eq:dhigh-lb}.
\end{align*}
\end{thm}
The proof of~\eqref{eq:dhigh-lb} is the same as the proof of ~\eqref{eq:gen-lb} (which in turn is based on the AGM bound~\cite{AGM}) and is in Appendix~\ref{sec:gen-lb}. We prove~\eqref{eq:dhigh-ub} using a structural result similar to Theorem~\ref{thm:opt-decomp} on $\LP^{(+)}$ and argue that each resulting $G_i$ for every $i \in [t]$ is a DAG. Note that we can then invoke Theorem~\ref{thm:gen-main} on each of these $G_i$s to prove~\eqref{eq:dhigh-ub}, as required. The details are deferred to Appendix~\ref{sec:dhigh-ub}.

In conclusion, Theorems~\ref{thm:dsmall} and~\ref{thm:dhigh-main} together imply that we have a worst-case optimal algorithm for computing $\TJ$ for any $G$, $L$ and $d$, answering Question~\ref{ques:panda} in affirmative for Assumption~\ref{assump:deg}. Further, we prove our upper bounds~\eqref{eq:dlow-ub} and~\eqref{eq:dhigh-ub} by making $G$ acyclic, which answers Question~\ref{ques:acyclic} in the affirmative for Assumption~\ref{assump:deg} as well. 

%% file: related_work.tex
\section{Related Work} \label{sec:related}
Worst-case optimal join algorithms and their combinatorial counterpart, worst-case size bounds for conjunctive queries, have seen tremendous research activity in the last few years. The authors of~\cite{GLVV} came up with worst-case size bounds for join queries with functional dependencies (which in our setting is the $\ell_{\infty}$ bound being $1$), which were later extended combinatorially in~\cite{Szymon} and translated algorithmically in~\cite{FAQ-FD}. Computing join queries with degree bounds on the relations, which are a generalization of functional dependencies, have also been studied in~\cite{JR-paper} (in addition to papers by Abo Khamis et al. discussed in the introduction). Decomposing simple graphs (for the arity two case) into specific subgraphs in the context of worst-case optimal join algorithms has previously been studied in~\cite{NPRR}. Exploiting structure in input data (for e.g., data with bounded treewidth, in addition to structure of the input query) for efficient computaton of joins has been studied for at least two decades~\cite{Grohe}.

%% file: conclusion.tex
\section{Conclusions and Open Questions} \label{sec:concl}
In this paper, we have presented a worst-case optimal join algorithm for the case of $p\in (1,2]$ for any join query with relations having arity (at most) two. Our results work for any fixed $p$ as well (as long as the join query graph has large enough girth). Along this way, we have (partially) resolved in the affirmative two open questions (Questions~\ref{ques:panda} and ~\ref{ques:acyclic}) from~\cite{ngo-survey} regarding worst-case optimal join query processing with $\ell_1$ and $\ell_\infty$ bounds. We leave the following questions for future work.\footnote{We summarize the key one here and defer a detailed discussion to Appendix~\ref{sec:app-concl}.}
\begin{ques} \label{ques:gen-arity}
Can we extend our results on simple graphs for $p \in (1, 2]$ to more general hypergraphs?
\end{ques}
Our structural decomposition result for the arity two case limits us from extending our results to general hypergraphs. However, a natural question is if we can extend our results to acyclic hypergraphs. In Appendix~\ref{sec:hyper}, we show that we can recover the AGM bound using a suitable generalization of degrees and degree configuration, and leave the question of extending it to general $\ell_{p}$ (along $\ell_{\infty}$ bounds), as future work. We conjecture that if we replace Algorithm~\ref{algo:gen-ub} with this algorithm for hypergraphs then Algorithm~\ref{algo:generic-ub} would work for hypergraphs with simple degree constraints~\cite{panda-pods20}. We leave the challenging task of proving this conjecture as future work.

%% file: appendix_sec1.tex
\section{Missing Details in Section~\ref{sec:norm-intro}} \label{app:intro}

\subsection{Scale-free graphs} \label{app:power-law}
In this section, we consider {\em scale-free} graphs (or graphs that follow the power law)~\cite{power-law} and consider how various size based on $\ell_1,\ell_\infty$ and more generally $\ell_p$ bounds compare with each other. More specifically, we ask the question:
\begin{ques}
	Let $G$ be the query graph.
	If $H$ is a scale-free graph, how do bounds on number of copies of $G$ in $H$ compare based on various $\ell_p$ and $\ell_\infty$ bound.
\end{ques}
In the rest of the section, we will recall the formal definition of scale free graphs and then compare various bounds that we consider in this paper.
We would like to stress that we do not claim that these graphs are prevalent in practice (in fact, by now there is considerable doubt on whether such graphs strongly capture graphs that occur in practice~\cite{no-power-law}) but this section shows a {\em mathematically} natural class of graphs, for which using $\ell_p$ for $p\in (1,2)$ norm bound gives us a win over existing join output size bounds that use $\ell_1$ and $\ell_\infty$ bounds on the input relations.

To make our lives simpler, we will not explicitly talk about the direction of the  tuples in the relation defined by $H$ that follows the power law. Our argument below works for whichever direction we choose for all edges in the graph. In fact, we will use the `undirected' degrees when defining the degree sequence of the relation (which in turn defines the norm bounds). Wherever appropriate, we will point where direction matters and were it does not but by default the reader can assume all the edges in $G$ are directed. 

Recall that a (bipartite) graph\footnote{The two partitions correspond to the domains of attributes $A$ and $B$ and the tuples in $R(A,B)$ correspond to the edges in the graph.} $H=(\cV,\cE)$ that obeys power law with {\em exponent} $\alpha$ has fraction of vertices with degree $k$ proportional to $k^{-\alpha}$. \textbf{For this section we will assume $\alpha\in (2,3)$}. The main reason for doing so is that this is the most common scale exponent observed in practice~\cite{no-power-law}.

We will posit that the maximum degree of $H$ satisfies
\[d = |\cV|^{1/\alpha}.\]
Now note that the number of edges in the graph (which will be the same as $\ell_1$ bound is {\em proportional to}:
\[\sum_{k=1}^d |\cV|\cdot k^{-\alpha}\cdot k = |\cV|\cdot\sum_{k=1}^d \frac 1{k^{\alpha-1}} = \Theta\inparen{\abs{\cV}},\]
where the last equality follows since we have $\alpha-1>1$.
Similarly, we show that the the $\ell_p$ bound $L_p$ when $p\le \alpha-1$ in this case is {\em proportional to}:
\[\sqrt[p]{\sum_{k=1}^d |\cV|\cdot k^{-\alpha}\cdot k^p} = \sqrt[p]{|\cV|}\cdot\sqrt{\sum_{k=1}^d \frac 1{k^{\alpha-p}}} = \tilde{\Theta}\inparen{\sqrt[p]{\abs{\cV}}},\]
where the last equality follows since $\alpha\ge p$ and hence $\alpha-p\ge 1$ (and the $\tilde{\Theta}$ hides a log factor).

Thus, if we use our usual notation $N$ to denote the $\ell_1$ bound, then we have $|\cV|=\Theta(N)$, and hence we have for $p\in [1,\alpha-1]$:
\[d=\Theta\inparen{\sqrt[\alpha]{N}} \text{  and  } L_p=\tilde{\Theta}\inparen{\sqrt[p]{N}}.\]
With the above basic norm bounds in place, we undertake a more detailed comparison of the join sized bounds considered in this paper.

Before we proceed, we formally setup the join query whose size we will bound. By default we will assume that the query graph $G=(V,E)$ is directed (for each pair of vertices, there is a directed edge in at most one direction). For each edge $\edge{A}{B}$ in $G$, we will assume that the corresponding relation $R_{\edge{A}{B}}$ is exactly the set of edge in $H$ with $\Dom(A)=\Dom(B)=\cV$ (suitably directed from $\edge{A}{B}$). For the rest of the section, unless noted otherwise, fix an $\alpha\in (2,3)$.

\subsubsection{Size bounds based on $\ell_p$ bounds for $p\in (1,2)$}
By Theorem~\ref{thm:lp-main}, the size bound given an $\ell_p$ bound of $L_p$ on $H$, is given by the following LP, which we re-name as $\LP^{(+)}(G,p)$ to emphasize the dependence on $G$ and $p$):
\begin{align*}
& \min \quad \usum{\edge{v}{u} \in E} x_{\edge{v}{u}} \log(L_{p}) \tag{$\LP^{(+)}(G,p)$}\\
& \usum{\edge{v}{u} \in E} x_{\edge{v}{u}} + \usum{\edge{u}{w} \in E} \frac{x_{\edge{u}{w}}}{p} \ge 1 \quad \forall u \in V \numberthis \label{eq:lp-covering-app} \\
& x_{\edge{v}{u}} \ge 0 \quad \forall \edge{v}{u} \in E \numberthis \label{eq:lp-primal-app}.
\end{align*}

More specifically Theorem~\ref{thm:lp-main} states that the join size is bounded by $2^{\LP^{(+)}(G,p)}$, where we overload notation and use $\LP^{(+)}(G,p)$ to also denote the objective value of the above LP.

We first note that the bound improves as $p$ increases (provided $p\le \alpha-1$):
\begin{lemm}
	\label{lem:p-incr}
	Let $\alpha\in (2,3)$. Then for any $p,q\in [1,\alpha-1]$ such that $p\le q$, and for any $G$:
	\[\LP^{(+)}(G,p) \le \LP^{(+)}(G,q).\]
\end{lemm}

Note that the above implies that if the scale exponent is $\alpha\in (2,3)$, then the best bound is achieved with the $\ell_{\alpha-1}$ norm bound. We now prove the above lemma.

\begin{proof}[Proof of Lemma~\ref{lem:p-incr}]
	Let $\vx^{(p)}$ be an optimal solution to $\LP^{(+)}(G,p)$. Note that since $\log{L_p}=\frac{1}{p}\cdot\log{N}$,\footnote{Technically there should be an additive $O(\log\log{N})$ factor as well-- however, this is a lower order term that does not change the subsequent argument so we will ignore this additive term for clarity.} we have that
	\[\LP^{(+)}(G,p) =\frac{\log{N}}p\cdot \usum{\edge{v}{u} \in E} x_{\edge{v}{u}}^{(p)}.\]
	Now consider the related vector
	\[\tilde{\vx}=\frac qp\cdot \vx^{(p)}.\]
	We first claim that $\tilde{\vx}$ is a feasible solution for $\LP^{(+)}(G,q)$. Indeed,~\eqref{eq:lp-primal-app} is satisfied since all elements of $\vx^{(p)}$ are non-negative. Next, we show that $\tilde{\vx}$ satisfies~\ref{eq:lp-covering-app} for $\LP^{(+)}(G,q)$:
	\begin{align*}
	\usum{\edge{v}{u} \in E} \tilde{x}_{\edge{v}{u}} + \usum{\edge{u}{w} \in E} \frac{\tilde{x}_{\edge{u}{w}}}{q}
	&=\frac qp\cdot \usum{\edge{v}{u} \in E} x_{\edge{v}{u}}^{(p)} + \usum{\edge{u}{w} \in E} \frac{x_{\edge{u}{w}}^{(p)}}{p}\\
	&\ge \usum{\edge{v}{u} \in E} x_{\edge{v}{u}}^{(p)} + \usum{\edge{u}{w} \in E} \frac{x_{\edge{u}{w}}^{(p)}}{p}\\
	&\ge 1,
	\end{align*}
	where the equality follows from definition of $\tilde{\vx}$, the first inequality follows since $q\ge p$ and the final inequality follows since $\vx^p$ is a feasible solution to $\LP^{(+)}(G,p)$.
	
	Next, note that the objective value obtained by $\tilde{\vx}$ is given by (where the first equality follows from definition of $\tilde{\vx}$):
	\[\frac{\log{N}}q\cdot \usum{\edge{v}{u} \in E} \tilde{x}_{\edge{v}{u}} = \frac{\log{N}}p\cdot \usum{\edge{v}{u} \in E} x_{\edge{v}{u}}^{(p)} = \LP^{(+)}(G,p).\]
	The claim them follows from the fact that the LP has a minimizing objective value
\end{proof}

Next we argue that the inequality in Lemma~\ref{lem:p-incr} can be strict for $p<q$ for certain graphs:
\begin{lemm}
	\label{lem:cycle-scale-free-gap}
	Let $G$ be a cycle on $n$-nodes. Then for any $p\in [1,2]$,
	\[\LP^{(+)}(G,p) = \frac n{p+1}\cdot \log{N}.\]
\end{lemm}
\begin{proof}
	Consider the vector $\vx$ where for each $\edge{u}{n}\in E$, we set
	\[x_{\edge{u}{v}}=\frac p{p+1}.\]
	It is easy to check that since $G$ is a cycle the above is feasible solution. Further, it has an objective value of
	\[\frac{\log{N}}p\cdot\sum_{\edge{u}{n}\in E} x_{\edge{u}{v}} = \frac{\log{N}}p\cdot \frac{np}{p+1} = \frac n{p+1}\cdot \log{N}.\]
	This shows that $\LP^{(+)}(G,p)\le \frac n{p+1}\cdot \log{N}$.
	
	To prove that $\LP^{(+)}(G,p)\ge \frac n{p+1}\cdot \log{N}$, we claim that for any feasible solution $\vx$, we have
	\[\frac{p+1}p\cdot \sum_{\edge{u}{n}\in E} x_{\edge{u}{v}} \ge n,\]
	which prove the claimed lower bound above. The above inequality follows by summing up~\eqref{eq:lp-covering-app} overall vertices $u\in V$ and noting that since $G$ is a cycle each edge $x_{\edge{u}{v}}$ occurs exactly once with a coefficient $1$ in the constraint for $v$ and once with a coefficient of $\frac 1p$ in the constraint for $u$.
\end{proof}

Lemmas~\ref{lem:p-incr} and~\ref{lem:cycle-scale-free-gap} imply the following\footnote{While the results in this section have focused on the case of directed query graphs $G$ (with $H$ correspondingly directed), the result also holds for undirected by directing all edges in $G$ so that the resulting directed graph is a directed cycle and then directing the edges in $H$ correspondingly.}:
\begin{cor}
	\label{cor:cycle-scale-free-bound-p}
	Let $H$ have a scale exponent of $\alpha\in (2,3)$. Then our results imply that $H$ has at most $N^{\frac n\alpha}$ copies of $n$-cycles in it.
\end{cor}

\subsubsection{Join size bounds based on $\ell_1$ and $\ell_\infty$ norm bounds}

In this section, we focus on the case when $G$ and $H$ are both undirected.

It turns out that our $\ell_p$ based bounds for $p\in (1,2)$ are sometimes better than the $\ell_1+\ell_\infty$ bounds (but can also be worse). At a high level our bounds are better the more `cyclic' $G$ is. We simply focus on the case of $G$ being a cycle:

\begin{cor}
	\label{cor:cycle-scale-free-bound-infty}
	Let $H$ have a scale exponent of $\alpha\in (2,3)$. Then join sized bounds based on $\ell_1+\ell_\infty$ norm bounds imply that $H$ has at most $N^{\frac n\alpha+1-\frac 2\alpha}$ copies of $n$-cycles in it.
\end{cor}
\begin{proof}
	Since we have $\alpha>2$, we have that $d^2< N$ and hence Theorem~\ref{thm:dsmall} applies. Further, we need to orient the edges in $G$-- we again orient them so that the resulting directed graph is a directed cycle for which Theorem~\ref{thm:dsmall} implies a bound of
	\[N\cdot d^{n-2} = N^{1+\frac{n-2}\alpha}=N^{\frac n\alpha +1-\frac 2\alpha}.\]
	The claim follows from the fact that the bound of $N\cdot d^{n-2}$ is the smallest possible bound in~\eqref{eq:dlow-ub} over all possible orientations of a cycle.\footnote{The intuitive argument is as follows. All cyclic orientations are equivalent. So consider any acyclic orientation. If there is exactly one source, then it is easy to see that such an orientation also gives a bound of $N\cdot d^{n-2}$. For any other orientation with more than one source, we replace $d^2$ in the bound of $N\cdot d^{n-2}$ with an $N$, which is worse since $d^2<N$.}
\end{proof}

Note that by Corollary~\ref{cor:cycle-scale-free-bound-p} for the case of a cycle, $\ell_p$ based bounds (for $p=\alpha-1$) given better bounds the the $\ell_1+\ell_\infty$ bounds above. It is easy to see that this gap can be easily extended to the case when $G$ has a disjoint vertex cycle cover (a property that can be checked in polynomial time~\cite{tutte}).

However, in some cases the $\ell_1+\ell_\infty$ bounds can be (much better). Consider the case when $G$ is a star on $n$ vertices: i.e. with $n-1$ leaves. In this case, the best orientation for both the $\ell_1+\ell_\infty$ and the $\ell_p$ based bounds is to direct all the $n-1$ edges away from the center. In this case again for the $\ell_1+\ell_\infty$ norm based bound gives an size bound of $N\cdot d^{n-2}\le N^{\frac n\alpha +1-\frac 2\alpha}$. By contrast, the $\ell_p$ based bound in this case is $L_p^{n-1}=N^{\frac{n-1}p}$, which is minimized at $p=\alpha-1$, leading to a final bound of $N^{\frac{n-1}{\alpha-1}}$, which is clearly a worse bound.

\subsubsection{Lower bound instance for scale-free graphs}

Finally, we make a quick observation that our $\ell_p$ norm based bounds are right for the worst-case input with the given $\ell_p$ norm bound-- however, the hard instance are not scale-free graphs. In this subsection, we note that we can prove lower bounds for scale-free inputs, which are sort of close to our general $\ell_p$ based bound (and even match in some special setup).

We focus on the case of $G$ being an $n$-cycle and due to Corollary~\ref{cor:cycle-scale-free-bound-p}, we pick $p=\alpha-1$.
\begin{lemm}
	\label{lem:scale-free-lb}
	Fix $\alpha\in (2,3)$.
	For large enough $N$, there exists a graph $H$ with exponent $\alpha$ has $\Omega\left(N^{\frac{n+1}{\alpha+1}}\right)$ copies of $n$-cycle in it (in the join query setup of this section).
\end{lemm}
\begin{proof}
	Define $\Delta=\sqrt[\alpha+1]{N}$. Note that for every $k\le \Delta$, we have that the number of vertices with degree $k$ (which is proportional to $N/k^\alpha$) is at least $k$. Thus, in $H$, we can add $\left\lfloor \frac N{k^{1+\alpha}}\right\rfloor$ copies of the $[k]\times[k]$ sub-graph.\footnote{We will ignore vertices with degree larger than $\Delta$-- to satisfy the degree distribution for degrees, we need to only add $\Theta(N)$ dummy nodes that do not contribute to any $n$ cycle.} Thus, for any $k\le \Delta$, we get (ignoring the floors)
	\[\frac N{k^{1+\alpha}}\cdot k^n\]
	many $n$-cycles just from nodes with degree $k$. Thus, the overall number of $n$-cycles is at least (again ignoring constant factors):
	\[\sum_{k=1}^\Delta N\cdot k^{n-\alpha-1} \ge N\cdot \Delta^{n-\alpha}= N\cdot N^{\frac{n-\alpha}{\alpha+1}}=N^{\frac{n+1}{\alpha+1}},\]
	as desired.
\end{proof}

Note that for large enough $n$, there is a (polynomial) gap between the above lower bound and the upper bound of Corollary~\ref{cor:cycle-scale-free-bound-p}. However, for $\alpha=2$ the two bounds match for $n=3$, i.e. the triangle query. In this case, note that the tight bound is $\Theta(N)$ (which basically means that the triangles essentially comes from constant degree nodes, e.g. a matching instance gives such a lower bound).

\subsection{Join Queries in Columnar Databases} \label{app:columnar}
In this section, we discuss how join queries on single columns in columnar databases can alternatively be modeled as queries where relations have arity two.  Our goal here is to show that the space of arity two queries (a setting we consider in this paper) captures a non-trivial class of join queries in this setting.

Typically in columnar databases, each relation (irrespective of arity) is padded with an additional identifier column and when a join is computed on say two relations on the column to be joined, it can equivalently be modeled as a join query on relations with arity two, where the columns for each relation are the rowid and the column being joined. The remaining columns in the join output are then obtained using the matching rowids.

\subsection{Streaming Results for $\ell_{p}$-norm} \label{sec:streaming}
In this section, we survey $\ell_{p}$-norm related work in the streaming model. Our goal here is to capture the generality of our $\ell_{p}$-norm statistic (specifically for $p \in [1, 2]$).

$\ell_2$-norm has been well-studied in the streaming model~\cite{Alon}, where the tuples come one-by-one and the goal is to approximate the $\ell_2$-norm. It is known that 
approximating $p \in [1, 2]$ takes only logarithmic space~\cite{Alon} while approximating all other $p \in (2, \infty]$ takes polynomial space~\cite{CKS03}. From the joins point of view, streaming joins (in the context of update queries) with worst-case optimality guarantees have been studied for the triangle case~\cite{Updates-Kara} and we expect $\ell_2$-norms to find applications in this space.

\subsection{Going from Undirected to the Directed Setting} \label{app:orientation}
So far in the paper, we have assumed that the query graph and the input relations are directed i.e., each tuple in each $R_{e}$ for every $e = \edge{v}{u} \in E$ is oriented from $v$ to $u$. As we discussed in Section~\ref{sec:other-results} with the example of $\ell_{\infty}$-norm, if we decide the direction tuple-wise, the resulting $\ell_{p}$-norm could be significantly reduced. We show in Appendix~\ref{app:shi}\footnote{We thank Shi Li for giving us this result.} that the optimal orientation that minimizes the $\ell_{p}$-norm for all $p \in [1, \infty]$ can be achieved in polynomial time.

Next, we show how to handle this scenario from the join computation point of view since Algorithm~\ref{algo:generic-ub} expects the tuples to be directed consistently. It turns out there is 
is a simple way\footnote{We thank Szymon Toru\'{n}czyk for pointing this out.} to handle the more general orientation process mentioned above: we replace $R(A,B)$ by the union of two sub-relations $R\edge{A}{B}$ and $R\edge{B}{A}$ and as a result, the overall join problem reduces to the union of $2^{|E|}$ join problems in which all tuples in a relation are directed the same way. As we will show formally in Appendix~\ref{sec:max-orient}, this adds only a multiplicative factor of $2^{\abs{E}}$ to the runtime of our worst-case optimal algorithm. Note that this factor is only exponential in the query size and as a result, can be ignored in data complexity of computing the join.

\subsubsection{Orienting Relations in One Direction} \label{sec:max-orient}
We state the undirected setting formally here, where we fix $p \ge 1$ upfront and we have relations $R_{\edgeud{v}{u}}$ with $\norm{R_{\edgeud{v}{u}}}_{p} \le L_e$ for every $\edgeud{v}{u} \in E$ (where $E$ is undirected). For any instance $I = \{R_{\edgeud{v}{u}}: \norm{R_{\edgeud{v}{u}}}_{p} \le L_e \}$, we define $\TJ = \ujoin{\edgeud{v}{u} \in E} R_{\edgeud{v}{u}}$. As stated earlier, we decompose $R_{\edgeud{v}{u}}$ into $R_{\edge{v}{u}}$ and $R_{\edge{v}{u}}$ respectively.

Given this setup, our goal is to show that orienting each relation $R_{\edgeud{v}{u}}$ in one direction (i.e., either $\edge{v}{u}$ (or) $\edge{u}{v}$ but not both) changes our upper bound on $|\TJ|$ only by a factor of $2^{\abs{E}}$. For a fixed orientation of all relations, we overload $E$ to denote the orientation.
\begin{lemm} \label{thm:max-orient}
Let $L_{\edge{v}{u}} = L_{\edge{u}{v}} = L_e$ for every $e \in E$. For any undirected $G = (V, E)$, we have
\begin{align*}
|\TJ| &\le  2^{|E|} \umax{\text{all possible orientations of } G} \left|\ujoin{\edge{v}{u} \in E} R_{\edge{v}{u}} \right| \numberthis \label{eq:max-orient-ub}
\end{align*}
and there exists an instance $I$
\begin{align*}
\{R_{\edge{v}{u}}: ||R_{\edge{v}{u}}||_{p} \le L_{\edge{v}{u}}, \edge{v}{u} \in E\}
\end{align*}
such that
\begin{align*}
|\TJ| & \ge  \left|\ujoin{\edge{v}{u} \in E} R_{\edge{v}{u}} \right| \numberthis \label{eq:max-orient-lb},
\end{align*}
where the orientation in~\eqref{eq:max-orient-lb} is the one that achieves the maximum in~\eqref{eq:max-orient-ub}.
\end{lemm}
\begin{proof}
We start by recaling that $R_{\edge{v}{u}} \cup R_{\edge{u}{v}} = R_{\edgeud{v}{u}}$ for every $\edgeud{v}{u} \in E$. We claim the following based on definition of $|\TJ|$, which immediately gives~\eqref{eq:max-orient-ub} since there are $2^{|E|}$ orientations of $G$.
\begin{align*}
|\TJ| & \le \usum{\text{all possible orientations of } G}\left|\ujoin{\edge{v}{u} \in E} R_{\edge{v}{u}} \right|
\end{align*}
We prove the statement below and the above follows directly.
\begin{align*}
\TJ = \ucup{\text{all possible orientations of } G} \ujoin{\edge{v}{u} \in E} R_{\edge{v}{u}} \numberthis \label{eq:orient}.
\end{align*}
The proof is by contradiction. Let $\TJ \subset \ucup{\text{all possible orientations of } G} \ujoin{\edge{v}{u} \in E} R_{\edge{v}{u}}$. Then, there exists a tuple $\vt \in \ujoin{\edge{v}{u} \in E} R_{\edge{v}{u}}$ such that $\vt \not \in \TJ$. For each $\edge{v}{u} \in E$, we have $\pi_{\edge{v}{u}}(\vt) \in R_{\edge{v}{u}} \subseteq R_{\edgeud{v}{u}}$. Note that this implies $\vt \in \TJ$ as well, resulting in a contradiction. We argue the reverse direction as well i.e., let $\TJ \supset \ucup{\text{all possible orientations of } G} \ujoin{\edge{v}{u} \in E} R_{\edge{v}{u}}$. Then, there exists a tuple $\vt \in \TJ$ such that $\vt \not \in \ujoin{\edge{v}{u} \in E} R_{\edge{v}{u}}$ for every orientation of $G$. Note that this implies for each relation $R_{\edgeud{v}{u}}$, we have $\pi_{\edgeud{v}{u}}(\vt) \in R_{\edgeud{v}{u}}$, which in turn implies $\pi_{\edgeud{v}{u}}(\vt) = R_{\edge{v}{u}}$ (or) $\pi_{\edgeud{v}{u}}(\vt) = R_{\edge{u}{v}}$. As a result, $\vt$ will be present in the join out of some orientation of $J$, resulting in a contradiction. This proves~\eqref{eq:orient}.
	
To complete the proof, we argue~\eqref{eq:max-orient-lb}, which follows from the proof above as well i.e., for any orientation of $G$, we have 
\begin{align*}
\TJ & \supseteq \ujoin{\edge{v}{u} \in E} R_{\edge{v}{u}},
\end{align*}
which in turns implies~\eqref{eq:max-orient-lb} for the orientation that achieves the maximum in~\eqref{eq:max-orient-ub}.
\end{proof}

\subsubsection{Computing the Optimal Orientation} \label{app:shi}
Give an undirected graph $H=(V,E)$, let $\sigma$ denote an orientation of the edge set $E$ s.t. for every $\edgeu{u}{v}\in E$, we have $\sigma(\edgeu{u}{v})\in \inset{\edge{u}{v},\edge{v}{u}}$, i.e. we pick exactly one of the two possible orientations for the undirected edge. Given an orientation $\sigma$, define the corresponding `degree vector' $\vd_\sigma$ such that $\vd_\sigma(u)$ is the {\em outdegree} of $u\in V$ under this orientation.
We consider the following problem:
\begin{problem}
	\label{prob:orient}
	Given $p\in [1,\infty)$ and $H=(V,E)$, compute a $\sigma$ such that $\norm{\vd_\sigma}_p$ is minimized.
\end{problem}

Note that the problem for $p=1$ is trivial (since all orientations give the same norm value of $|E|$) while the problem for $p=\infty$ is the problem of computing degeneracy of $H$ (which has a well-known linear time algorithm).

In this appendix, we show that the above problem can be solved in polynomial time for any fixed $p\in [1,\infty)$. 
The results in this section are due to Shi Li. We thank him for allowing us to use his proof.

The high level idea is to come up with a convex programming relaxation of the above problem and then we argue that there indeed exists an integral solution (and that one can round a fractional optimal solution into an optimal integral one as well). 

Before we state the convex program, we define a piece-wise linear function that agrees with the $\ell_p$ norm for all integer values. In particular, define:
\[f_p(x) = \inparen{\floor{x}}^p\cdot \inparen{1-x+\floor{x}} + \inparen{\floor{x}+1}^p\cdot \inparen{x-\floor{x}}.\]
Indeed, note that if $x$ is an integer, we have $f(x)=x^p$, as desired. Also note that the function is convex (which in turn follows from the  facts that $x^p$ is convex  for $p\ge 1$  and that $f_p(x)$ is piece-wise linear in between two integral values).

Now consider the following convex program:
\begin{align}
\label{eq:cvx-obj}
& \min\sum_{u\in V} f_p\inparen{d_u}  \\
\label{eq:cvx-direct}
&\text{s.t. } x_{\edge{u}{v}}+x_{\edge{v}{u}}=1 \text{ for every } \edgeu{u}{v}\in E\\
\label{eq:cvx-deg}
& d_u =\sum_{\edgeu{u}{w}\in E} x_{\edge{u}{w}} \text{ for every } u\in V\\
& x_{\edge{u}{v}},x_{\edge{v}{u}}\ge 0 \text{ for every} \edgeu{u}{v}\in E\nonumber
\end{align}

Note that any integral solution must have $x_{\edge{u}{v}},x_{\edge{v}{u}}\in\inset{0,1}$ and corresponds to an orientation $\sigma$ ($x_{\edge{u}{v}}=1$ implies that $\sigma\inparen{\edgeu{u}{v}}=\edge{u}{v}$ and $x_{\edge{u}{v}}=0$ implies that $\sigma\inparen{\edgeu{u}{v}}=\edge{v}{u}$) and $d_u$ then corresponds to the outdegree of $u$ under $\sigma$. Further, by our earlier observation on $f_p(\cdot)$, the objective function corresponds to the objective function of our problem. Thus, if we can compute the optimal integral solution to the above convex program, then we will be done.

First, we recall the well known result that the convex program above can be solved in polynomial time (since the objective is convex and all constraints are linear). However, such a solution is not guaranteed to be integral. We argue next that we can convert an optimal fractional solution into an optimal integral solution in polynomial time (and hence also argue that the convex program also always has an optimal integral solution).

\begin{thm}
	\label{thm:shi}
	Let $\vx=\inparen{x_{\edge{u}{v}},x_{\edge{v}{u}}}_{\edgeu{u}{v}\in E}$ be an optimal solution to the convex program~\eqref{eq:cvx-obj}. Then in polynomial time, we can convert this into an integral optimal solution.
\end{thm}
\begin{proof}
	If $\vx$ is integral then we have nothing to argue so we assume there is at least one $\edgeu{u}{v}\in E$ such that $x_{\edge{u}{v}},x_{\edge{v}{u}}\in (0,1)$-- for notational simplicity we will state that the undirected edge $\edgeu{u}{v}$ is {\em fractional}. Let $d_u$ for every $u\in V$ be defined by~\eqref{eq:cvx-deg}.
	
	We make the following simple observation that will be useful later on. Let $u\in V$ and let $d_u$ be integral. Then if there is a fractional edge $\edgeu{u}{w}$, then there has to be at least one more fractional edge $\edgeu{u}{v}$ for $w\ne v$. (This follows because if exactly one edge incident to $u$ is fractional, then by~\eqref{eq:cvx-deg}, $d_u$ cannot be an integer.)
	
	In the rest of the proof, we  present a polynomial time procedure\footnote{We will not explicitly argue the runtime of the procedure below but its description immediately implies the claimed polynomial runtime.} that creates a new optimal solution $\vx'$ to~\eqref{eq:cvx-obj} such that $\vx'$ has strictly fewer number of fractional edges or fractional  degree values $d_u$ (than those for $\vx$). Note that this is enough to prove our claimed result (since the procedure below can be run at most $|V|+|E|$ to get the desired integral optimal solution).
	
	We first consider the case when all $d_u$ values are integers. Let $\edgeu{w}{v}$ be any fractional edge. Then by the observation above there exists an fractional edge $\edgeu{v}{y}$ for $y\ne w$. Thus, we can `move` from $w$ to $y$. We can continue this process by picking a fractional edge to get a new node  till we end up with a cycle $C$ (and such a cycle has to exists since there are $|V|$ vertices). Now define,
	\begin{equation}
	\label{eq:def-eps-shi}
	\eps=\min_{\edgeu{u}{v}\in C} \min\inset{x_{\edge{u}{v}},x_{\edge{v}{u}}}.
	\end{equation}
	Note that by construction of $C$, $\eps>0$. We will construct two related feasible solutions from $\vx$ that both are also optimal and at least one of them now has at least one less fractional edge. Let the vertices in $C$ in order be $u_0,u_1,\dots,u_{c-1},u_0$. Then we define two solution $\vx^+$ and $\vx^-$ as follows. Both these solutions agree with $\vx$ for all edges not in $C$. Otherwise for every $0\le i<c$, we have
	\begin{align*}
	x^+_{\edge{u_i}{u_{(i+1)\mod{c}}}}&=x_{\edge{u_i}{u_{(i+1)\mod{c}}}}+\eps\\
	x^+_{\edge{u_{(i+1)\mod{c}}}{u_i}}&=x_{\edge{u_{(i+1)\mod{c}}}{u_i}}-\eps\\
	x^-_{\edge{u_i}{u_{(i+1)\mod{c}}}}&=x_{\edge{u_i}{u_{(i+1)\mod{c}}}}-\eps\\
	x^-_{\edge{u_{(i+1)\mod{c}}}{u_i}}&=x_{\edge{u_{(i+1)\mod{c}}}{u_i}}+\eps\\
	\end{align*}
	i.e. in $\vx^+$ the `clockwise' directions get increased by $\eps$ and the `counter-clockwise' directions get decreased by $\eps$ (and this gets flipped in $\vx^-$). Let $d^+_u$ and $d^-_u$ be the corresponding degree values defined by~\eqref{eq:cvx-deg}. It can be verified that both $\vx^+$ and $\vx^-$ are still feasible solutions and $d^+_u=d^-_u=d_u$ (and hence, both $\vx^+$ and $\vx^-$ are still optimal). Finally, note that by definition of $\eps$ in~\eqref{eq:def-eps-shi}, in either $\vx^+$ or $\vx^-$ at least one edge in $C$ is no longer fractional.
	
	Finally, we consider the case when not all $d_u$ are integers. Since we have that $\sum_{u\in V} d_u=|E|$ (which in turn follows by summing up~\eqref{eq:cvx-direct} over all $\edgeu{u}{v}\in E$ and then noting each outgoing edge appears exactly once in this sum), which is integer. Hence, this implies there has to be at least two vertices with non-integral $d$ values. Let $s\in V$ such that $d_s$ is not an integer. This implies that there exists an incident edge $\edgeu{s}{u}$ that is not fractional (because if not, $d_s$ will be an integer). Now we keep on adding incident fractional edges as we did in the first case. Here we stop in case we get a cycle with all vertices $u$ in the cycle have an integer $d_u$ (in which case we just run the argument from the first case above) else we end up with another $t\ne s$ such that $d_t$ is not an integer. In other words, we end up with an $s-t$ path $\cP$. Now define
	\begin{equation}
	\label{eq:def-eps-shi-path}
	\eps=\min\inset{d_s-\floor{d_s}, d_t-\floor{d_t}, \ceil{d_s}-d_s,\ceil{d_t}-d_t,\min_{\edgeu{u}{v}\in \cP} \min\inset{x_{\edge{u}{v}},x_{\edge{v}{u}}}}.
	\end{equation}
	Construct $\vx^+$ by increasing the $x_{\edge{u}{v}}$ values (by $\eps$) that are in the directed $s\to t$ path and decreasing the $x_{\edge{v}{u}}$ values (by $\eps$) that are in the directed $t\to s$ path. Similarly we get $\vx^-$ by replacing $\eps$ by $-\eps$ in the definition of $\vx^+$. ($\vx^+,\vx^-$ and $\vx$ are exactly the same outside of edges in $\cP$.)   Let $d^+_u$ and $d-_u$ be the corresponding degree values defined by~\eqref{eq:cvx-deg}. It can be again verified that both $\vx^+$ and $\vx^-$ are still feasible solutions. Further it can be verified that for all vertices $u$ in $\cP$ other than $s$ and $t$ we have $d^+_u=d^-_u=d_u$. The only degree values that can change are those of $s$ and $t$ (and indeed we have $d^+_s=d_s+\eps,d^+_t=d_t-\eps,d^-_s=d_s-\eps,d^-_t=d_t+\eps$). Finally, by the choice of $\eps$ either in $\vx^+$ or $\vx^-$ at least one edge is  no more fractional {\em or} at least one vertex will no longer have a fractional  degree value.
	
	To complete the proof we argue that either $f_p(d^+_s)+f_p(d^+_t)-f_p(d_s)-f_p(d_t)$ or $f_p(d^-_s)+f_p(d^-_t)-f_p(d_s)-f_p(d_t)$ is at most zero (which note will complete the proof). To argue this, we inspect these differences. For notational convenience define $D_s=\floor{d_s}$ and $D_t=\floor{d_t}$. Then by definition of $\eps$ in~\eqref{eq:def-eps-shi-path}, we have $\floor{d^+_s}=\floor{d^-_s}=D_s$ and $\floor{d^+_t}=\floor{d^-_t}=D_t$. Further, note that by definition of $f_p(\cdot)$, we have
	\begin{align*}
	f_p(d^+_s)&= f_p(d_s+\eps)\\
	& = \inparen{D_s}^p\cdot (1-d_s-\eps+D_s) + \inparen{1+D_s}^p\cdot (d_s+\eps-D_s)\\
	& = \inparen{D_s}^p\cdot (1-d_s+D_s) + \inparen{1+D_s}^p\cdot (d_s-D_s) +\eps\cdot\inparen{\inparen{1+D_s}^p - \inparen{D_s}^p}\\
	& = f_p(d_s) +\eps\cdot\inparen{\inparen{1+D_s}^p - \inparen{D_s}^p}.
	\end{align*}
	Similarly, we get
	\begin{align*}
	f_p(d^-_s)&= f_p(d_s) -\eps\cdot\inparen{\inparen{1+D_s}^p - \inparen{D_s}^p}\\
	f_p(d^+_t)&= f_p(d_t) +\eps\cdot\inparen{\inparen{1+D_t}^p - \inparen{D_t}^p}\\
	f_p(d^-_t)&= f_p(d_t) -\eps\cdot\inparen{\inparen{1+D_t}^p - \inparen{D_t}^p}.
	\end{align*}
	Thus, we have
	\begin{align*}
	f_p(d^+_s) - f_p(d_s) + f_p(d^+_t) - f_p(d_t) = \eps\cdot\inparen{\inparen{1+D_s}^p - \inparen{D_s}^p + \inparen{1+D_t}^p - \inparen{D_t}^p}\\
	f_p(d^-_s) - f_p(d_s) + f_p(d^-_t) - f_p(d_t) = -\eps\cdot\inparen{\inparen{1+D_s}^p - \inparen{D_s}^p + \inparen{1+D_t}^p - \inparen{D_t}^p}.
	\end{align*}
	From the above is it easy to see that at least one of $f_p(d^+_s)+f_p(d^+_t)-f_p(d_s)-f_p(d_t)$ or $f_p(d^-_s)+f_p(d^-_t)-f_p(d_s)-f_p(d_t)$ is at most zero, as desired.
\end{proof}

%% file: appendix_sec2.tex
\section{Missing Details in Section~\ref{sec:tech}} \label{app:tech}

\subsection{$\panda$ as a Substitute for Algorithm~\ref{algo:gen-ub}} \label{sec:panda-algo}
In this section, our goal is to show that $\panda$ be used as a subsitute for Algorithm~\ref{algo:gen-ub}. We start by noting that $\panda$~\cite{panda} is traditionally designed for $\ell_1$ and $\ell_\infty$ norms. We note here that with a simple modification, it can be used in place of our Algorithm~\ref{algo:gen-ub} for general $\ell_{p}$-norm constraints as well, reflecting $\panda$'s generality. To see this, recall our runtime upper bound $\cD_{u}(\bard)$ from Theorem~\ref{claim:gen-acyclic}:
\begin{align*}
\min \inset{  \umin{\edge{v}{u} \in E}\inset{d_{\edge{v}{u}}},  \umin{\edge{u}{w} \in E} \inset{\frac{2^{p} \cdot L_{\edged{u}{w}}^{p}}{d_{\edge{u}{w}}^{p}}}   } \numberthis \label{eq:panda-min-expr}.
\end{align*}
While the terms $d_{\edge{v}{u}}$ can be modelled as standard $\ell_{\infty}$ bounds, the terms of the form 
$\frac{2^{p} \cdot L_{\edged{u}{w}}^{p}}{d_{\edge{u}{w}}^{p}}$ can be equivalently represented as a bound on the domain size $|\Dom(u)|$ of $u$ (which can in turn be modeled as an $\ell_{\infty}$ bound on a relation $\edge{\emptyset}{u}$). As a result, in addition to $d_{\edge{v}{u}}$, $\panda$ needs to know the values $L_{\edged{v}{u}}$ for every $\edge{v}{u} \in E$ to construct a proof sequence in this case whereas we do not this knowledge in our Algorithm~\ref{algo:generic-ub}. Assuming that using $\panda$ for each degree configuration $\bd$ achieves the polymatroid bound in this setup, note here that we would still end up losing a multiplicative factor of $O\inparen{\log{N}^{\inparen{\inparen{2^{|V|}}!}^2}}$ in its runtime, which we avoid in our algorithm. Using $\panda$ removes the need for Algorithm~\ref{algo:gen-ub} and Algorithm~\ref{algo:acyclic-ub}. Further, note that in Algorithm~\ref{algo:generic-ub}, we compute $\TJ$ as the union of all $\TJ(\bd)$s and as a result, by using $\panda$ for each $\bd$, we would lose another factor of $\log^{|E|} L$ (by upper bounding the algorithm's runtime with the $\bard$ achieving the maximum upper bound). However, we can ignore it since it is smaller than $O\inparen{\log{N}^{\inparen{\inparen{2^{|V|}}!}^2}}$.

To complete this discussion of using $\panda$ in our setup, we still need to show that the hard instance (with $\ell_1$ and $\ell_{\infty}$ bounds) for each $\bard$ (corresponding to the polymatroid bound) satisfies the $\ell_{p}$-norm bound. It turns out that these hard instances~\cite{panda} always have uniform $\ell_{\infty}$ bounds i.e., for relations $R_{\edge{v}{u}}$ and $R_{\edge{\infty}{u}}$ for every $\edge{v}{u} \in E$, the corresponding $\ell_{\infty}$ bounds are $\umin{\edge{v}{u} \in E}d_{\edge{v}{u}}$ and $\umin{\edge{v}{u} \in E} \left( \frac{L_{\edge{v}{u}}^{p}} {d_{\edge{v}{u}}^{p}} \right)$ (note that here we use $L_{\edge{v}{u}}$ in place of $L_{\edged{v}{u}}$ and this is ok to do since the latter is at most the former). Note that this implies the $\ell_{p}$-norm of each sub-relation $\edge{v}{u} \in E$ in the instance is upper bounded by $\sqrt[p]{ \frac{L_{\edge{v}{u}}^{p}}{d_{\edge{v}{u}}^{p}} d_{\edge{v}{u}}} \le L_{\edge{v}{u}}$, which is at most $L_{\edge{v}{u}}$. Since this instance is valid for any degree configuration $\bd$, it holds for the maximum among $\bd$s as well.

\subsection{Question~\ref{ques:acyclic} for General $\ell_{p}$-norm Constraints} \label{sec:app-acyclic-subgraph}
In this section, we revisit Question~\ref{ques:acyclic} for General $\ell_{p}$-norm constraints, which we considered earlier in Section~\ref{sec:sec2-acyclic}. We present an example below to demonstrate the answer to Question~\ref{ques:acyclic} is {\em no}.
\begin{Example} \label{ex:ell-2-acyclic-subgraph}
	Consider the triangle query $R\edge{A}{B}\Join S\edge{B}{C}\Join T\edge{C}{A}$ 
	where each relation has an $\ell_2$ bound of $L$. Now consider any acyclic-subgraph, e.g., the subquery $R\edge{A}{B}\Join S\edge{B}{C}$. It is easy to see that in the worst-case this subquery can have a join output size of at least $L^{5/2}$-- for e.g.,
	\begin{align*}
	R\edge{A}{B}=[L^2]\times[1] \text{ and } S\edge{B}{C}=[1]\times [\sqrt{L}].
	\end{align*}
	It is easy to verify that both relation instances have an $\ell_2$ bound of $L$ and the join output size for the above instance for $R\Join S$ is $L^{5/2}$. On the other hand, in Example~\ref{ex:triang-cyclic} we have already shown that for the triangle query, the tight bound on the join query size is $\Theta(L^2)$, while this example shows that the best acyclic subgraph join output size bound is $\Omega\inparen{L^{5/2}}$. \demo
\end{Example}

\subsection{Acyclic Triangle Upper Bound using Cauchy-Schwartz} \label{app:cauchy}
In this section, we present an upper bound argument for the case when $G$ is (what we call) an acyclic triangle (i.e., it is cyclic in the undirected sense but acyclic in the directed sense).
\begin{Example} \label{ex:cauchy}
$G$ has $V = \{ A, B, C\}$ and $E = \{ \edge{A}{B}, \edge{B}{C}, \edge{A}{C} \}$ with relations $R_{\edge{A}{B}}, S_{\edge{B}{C}}$ and $T_{\edge{A}{C}}$. Further, we assume the $\ell_2$-norm case and $\umax{\edge{v}{u} \in E}\norm{R_{\edge{v}{u}}}_{2} \le L$. 
	
Our goal is to prove an upper bound of $L^2$ on $R_{\edge{A}{B}} \Join S_{\edge{B}{C}} \Join T_{\edge{A}{C}}$. For any $a \in  \Dom(A)$, let $\deg_{R_{\edge{A}{B}}}(a)$ and $\deg_{T_{\edge{A}{C}}}(a)$ denote the outdegree of $a$ in $R_{\edge{A}{B}}$ and $T_{\edge{C}{A}}$. Then note that the output size of the join is: 
\begin{align*}
& \usum{a \in \Dom(A)}  \deg_{R_{\edge{A}{B}}}(a) \cdot  \deg_{T_{\edge{A}{C}}}(a) \\
& \le  \sqrt{ \usum{a \in \Dom(A)}   \deg_{R_{\edge{A}{B}}}(a)^2 } \cdot \sqrt{ \usum{a \in \Dom(A)}   \deg_{T_{\edge{A}{C}}}(a)^2}  \\
& = L \cdot  L = L^2,
\end{align*}
where the inequality follows from Cauchy-Schwartz. \demo
\end{Example}

%% file: appendix_sec3.tex
\section{Missing Details in Section~\ref{sec:norm-prelims}} \label{app:prelims}
The standard results below will be used in our analysis. 
\begin{lemm} \label{lemma:zero-sum}
Let $n$ be a positive integer, $a_{1}, \dots, a_n$ be non-negative real numbers and let $x_1, \dots, x_n \ge 0$ such that $\usum{i \in [n]}x_i \ge 1$. Then, we have
\begin{align*}
\umin{i \in [n]} a_i \le \uprod{ i \in [n]} a_i^{x_i}.
\end{align*}
\end{lemm}
\begin{lemm}[H\"{o}lder's inequality~\cite{ineq}] \label{lemma:Holders}
Let $m, n$ be positive integers and let $x_1, \dots, x_n \ge 0$ such that $\usum{i \in [n]} x_i \ge 1$. Let $a_{ij} \ge 0$ be non-negative real numbers for $i \in [m]$ and $j \in [n]$. We have
\begin{align*}
\usum{i \in [m]} \quad \uprod{j  \in [n]} a_{ij}^{x_{i}} \le \uprod{j \in [n]} \left( \usum{i  \in [m]} a_{ij} \right)^{x_i},
\end{align*}
assuming the convention $0^0 = 0$.
\end{lemm}

\subsection{Effective Domain Size Upper Bound based on Outdegree} \label{sec:node-bound}
For a fixed degree configuration $\bd = (d_{\edge{v}{u}})_{d_{\edge{v}{u}} \le L_{\edge{v}{u}}, \edge{v}{u} \in E}$, let $\Dom_{\bard}(v)$ denote the effective domain size of $v$ on $R^{d_{\edge{u}{v}}}_{\edge{u}{v}}$ for every $\edge{u}{v} \in E$ and $R^{d_{\edge{v}{w}}}_{\edge{v}{w}}$ for every $\edge{v}{w} \in E$, for every $v \in V$.
\begin{lemm} \label{lemma:node-bound}
	Fix $p: 1 \le p \le \infty$. We have for every $\edge{v}{u} \in E$:
	\begin{align*}
	\left|\Dom_{\bard}(v) \right| \le \left(\frac{2^{p} \cdot L_{\edged{v}{u}}^{p}}{d_{\edge{v}{u}}^{p}} \right).
	\end{align*} 
\end{lemm}
\begin{proof}
	We start by recalling that the subrelation $R^{d_{\edge{v}{u}}}_{\edge{v}{u}}$ satisifes the degree constraint $d_{\edge{v}{u}}$ i.e., each value of $v \in \Dom(v)$ has degree at least $\frac{d_{\edge{v}{u}}}{2} + 1$ and at most $d_{\edge{v}{u}}$.
	
	The proof is by contradiction. We assume
	\begin{align*}
	\left|\Dom_{\bard}(v) \right| > \left( \frac{2^{p} \cdot L_{\edged{v}{u}}^{p}}{d_{\edge{v}{u}}^{p}}\right) \numberthis \label{eq:dom-size}.
	\end{align*}
	Since each tuple in $R^{d_{\edge{v}{u}}}_{\edge{v}{u}}$ satisfies~\eqref{eq:subrelation-deg}, we have
	\begin{align*}
	\norm{R_{\edge{v}{u}}^{d_{\edge{v}{u}}}}^{p}_{p} & \ge  \left|\Dom_{\bd}(v) \right| \cdot \underset{x \in \Dom_{\bd}(u)}{\min} \left| D_{\edge{v}{u}}(x) \right|^{p} \\
	& > \left|\Dom_{\bd}(v) \right| \cdot \frac{d^{p}_{\edge{v}{u}}}{2^p}  \numberthis \label{eq:dom-lb-1} \\
	& > \left(\frac{2^{p} \cdot L_{\edged{v}{u}}^{p}}{d_{\edge{v}{u}}^{p}}\right) \cdot \frac{d_{\edge{v}{u}}^{p}}{2^p} \numberthis \label{eq:dom-lb-2} \\
	&  = L^{p}_{\edged{v}{u}}.
	\end{align*}
	Here,~\eqref{eq:dom-lb-1} follows from the fact stated above i.e., each value in $\Dom(v)$ has degree at least $d_{\edge{v}{u}}$ in $R^{d_{\edge{v}{u}}}_{\edge{v}{u}}$. Then,~\eqref{eq:dom-lb-2} follows from~\eqref{eq:dom-size}. Note that this contradicts $\norm{R_{\edge{v}{u}}^{d_{\edge{v}{u}}}}_{p} \le L_{\edged{v}{u}}$ and as a result, we have shown that $\left|\Dom_{\bd}(v)\right| \le  \left(\frac{2^{p} \cdot L_{\edged{v}{u}}^{p}}{d_{\edge{v}{u}}^{p}}\right)$, completing the proof.
\end{proof}

\subsection{More Details on Data Structures for our Algorithms} \label{app:data-struct}
In this section, we talk a bit more about our B-tree data structure. Recall that we store each relation $R_{\edge{v}{u}}$ as a B-tree-like index structure~\cite{Btree-ref}. For this B-tree, the first level is indexed by $(v)$ (i.e., all values in $\pi_{v}(R_{\edge{v}{u}})$) and the second level is indexed by $(u, \val_{v})$, where $\val_{v} \in \pi_{v}(R_{\edge{v}{u}})$. Note that this B-tree can be constructed in time $O(L_{\edge{v}{u}} \log(L_{\edge{v}{u}}))$. Since we construct such a B-tree for each $\edge{v}{u} \in E$, the total preprocessing time is $O( |E| \cdot L_{\edge{v}{u}} \log(L_{\edge{v}{u}}))$. Note that this time will be superseded by the final runtime of our join algorithm.

\subsection{Decomposing $R_{\edge{v}{u}}$ into buckets based on degree} \label{sec:log-decomp}
In this section, we state and prove a standard result on decomposing each relation $R_{\edge{v}{u}}$ into a logarithmic number of buckets.
\begin{lemm} \label{lemma:log-decomp}
Given any relation $R_{\edge{v}{u}}$, it can be decomposed into a union of $W$ subrelations (some of which can be empty), where each subrelation satisfies a degree constraint $d^{w}_{\edge{v}{u}} = 2^{w}$ for some $w \in[ W]$. Further, we have $W \le \log( L_{\edge{v}{u}}) + 1$.
\end{lemm}
\begin{proof}
We start by defining
\begin{align*}
D_{\edge{v}{u}}(x) = \left \{(x, y) \in R_{\edge{v}{u}}: y \in \Dom(u) \right \} \numberthis \label{eq:d-uv-x-app}.
\end{align*}
Next, we define for every $w \in [W]$:
\begin{align*}
R^{d_{\edge{v}{u}}}_{\edge{v}{u}} = \underset{x \in \Dom(v): 2^{ w -1} \le |D_{\edge{v}{u}}(x)| \le 2^{w} } {\cup} D_{\edge{v}{u}}(x) \numberthis \label{eq:sub-relation}, 
\end{align*}
where $D_{\edge{v}{u}}(x)$ is defined above. Note that each $R^{w}_{\edge{v}{u}}$ satisfies the degree constraint $d_{\edge{v}{u}} = 2^{w}$. Further, we have $\ucup{w \in [W]} R^{d_{\edge{v}{u}}}_{\edge{v}{u}} = R_{\edge{v}{u}}$ since for every $x \in \Dom(v)$, $D_{\edge{v}{u}}(x) \in R^{d_{\edge{v}{u}}}_{\edge{v}{u}}$, for some $w \in [W]$.
	
Since $\norm{R_{\edge{v}{u}}}_{p} \le L_{\edge{v}{u}}$, we have
\begin{align*}
2^{W - 1} &  \le \underset{x \in \Dom(v)}{\max} \left|D_{\edge{v}{u}}(x)\right| \\
& \le L_{\edge{v}{u}} = 2^{\log(L_{\edge{v}{u}})},
\end{align*} 
where the final inequality follows by the fact that $\norm{D_{\edge{v}{u}}}_{p} \ge \norm{D_{\edge{v}{u}}}_{\infty}$. This implies $W \le \log \left(L_{\edge{v}{u}} \right) + 1$. 
\end{proof}

\subsection{Computing $|\TJ|$ as the union of $|\TJ(\bd)|$s} \label{sec:join-decomp-size}
Recall that our algorithm (Algorithm~\ref{algo:generic-ub}) computes the final join output as the union of $\TJ(\bd)$ across all degree configurations $\bd$. We prove the following standard result.
 \begin{align*}
|\TJ| \le \usum{\bd = (d_{\edge{v}{u}})_{d_{\edge{v}{u}} \le L_{\edge{v}{u}}, \edge{v}{u} \in E}} |\TJ(\bd)| \numberthis \label{eq:join-decomp-size-app}.
\end{align*}
We argue the following result, which immediately implies the above.
\begin{align*}
\TJ = \ucup{\bd = (d_{\edge{v}{u}})_{d_{\edge{v}{u}} \le L_{\edge{v}{u}}, \edge{v}{u} \in E}} \TJ(\bd), \numberthis \label{eq:join-decomp}
\end{align*}
where recall that
\begin{align*}
\TJ(\bd) = \ujoin{\edge{v}{u} \in E} R^{d_{\edge{v}{u}}}_{\edge{v}{u}}.
\end{align*}
We start by claiming that based on our decomposition above, the following is true for every $\edge{v}{u} \in E$:
\begin{align*}
\ucup{d_{\edge{v}{u}} \le L_{\edge{v}{u}}} R^{d_{\edge{v}{u}}}_{\edge{v}{u}} = R_{\edge{v}{u}}. 
\end{align*}
In particular, this follows from the fact that each $x \in \Dom(v): (x, \cdot) \in R_{\edge{v}{u}}$ belongs to an unique $R^{d_{\edge{v}{u}}}_{\edge{v}{u}}$ (by definition of $R^{d_{\edge{v}{u}}}_{\edge{v}{u}}$). Note that this implies for every output tuple $\mathbf{t} \in \TJ$, we have $\pi_{\edge{v}{u}} (\mathbf{t})$ (which is $\mathbf{t}$ projected down to attributes $v$ and $u$) belongs to a unique $R^{d_{\edge{v}{u}}}_{\edge{v}{u}}$ and gives us~\eqref{eq:join-decomp}, as required.

%% file: appendix_structural_result.tex
\section{Proof of Theorem~\ref{thm:opt-decomp}} \label{sec:norm-structural}
We start by recalling $\LP^{(+)}$ on $G$:
\begin{align*}
& \min \quad \usum{\edge{v}{u} \in E} x_{\edge{v}{u}} \log(L_{\edge{v}{u}}) \tag{$\LP^{(+)}$}\\
& \usum{\edge{v}{u} \in E} x_{\edge{v}{u}} + \usum{\edge{u}{w} \in E} \frac{x_{\edge{u}{w}}}{p} \ge 1 \quad \forall u \in V  \numberthis \label{eq:struct-covering} \\
& x_{\edge{v}{u}} \ge 0 \quad \forall \edge{v}{u} \in E \numberthis \label{eq:struct-primal}.
\end{align*}
We also restate Theorem~\ref{thm:opt-decomp} here before proving it.
\begin{thm}
For any directed graph $G = (V, E)$, there exists an optimal solution $\mathbf{x}^* = \left( x^*_{\edge{v}{u}} \right)_{\edge{v}{u} \in E}$ to $\LP^{(+)}$ and a $t$ such that $G$ can be decomposed into a disjoint union of $t$ connected components (in the undirected sense) $G_i = (V_i, E_i)$ with 
\begin{align*}
|V_i| - 1 \le |E(G_i)| \le |V_i|, \text{ where }  x_{\edge{v}{u}} > 0 \quad \forall  \edge{v}{u} \in E(G_i)
\end{align*}
and $\mathbf{x}^*_{i} = \left( x^*_{\edge{v}{u}} \right)_{\edge{v}{u} \in E(G_i)}$ is an optimal basic feasible solution (see Definition~\ref{def:basic-feasible}) for $\LP^{(+)}$ on $G_i$ for every $i \in [t]$. Further, we have $\cup_{i=1}^{t}V(G_i) = V$ and $V(G_i) \cap V(G_j) = \emptyset$ \space $\forall i, j \in [t], i \neq j$. The following is true:
\begin{align*}
\TJ & = \times_{i \in [t]} \TJ(G_i).
\end{align*}
\end{thm}
We define some standard preliminaries on linear programs needed to prove this result.
\subsection{Preliminaries and Existing Results}
We start by recalling the definition of a basic feasible solution, specialized to $\LP^{(+)}$.
\begin{defn} [Basic Feasible Solution to LP~\eqref{eq:primal}~\cite{bf-book}] \label{def:basic-feasible}
A basic feasible solution $\mathbf{x} = \left(x_{\edge{v}{u}} \right)_{\edge{v}{u} \in E}$ to $\LP^{(+)}$ is one that satisfies all its $|V| + |E|$ constraints with at least $|E|$ of them satisfied with equality (ones we call {\em tight}). Let $C$ denote the  $(|V| + |E|)\times |E|$ constraint matrix, where the rows are indexed by constraints and columns are indexed by variables and let $S$ denote the set of tight constraints. Then, the matrix $C$ projected down to rows (i.e., constraints) in $S$ has rank exactly $|E|$. 
\end{defn}
We will invoke the following well-known theorem in our arguments.
\begin{thm} [From~\cite{bf-book}] \label{thm:bf-optimal}
There always exists an optimal solution to $\LP^{(+)}$ that is basic feasible. 
\end{thm} 

\subsection{Main Proof}
To prove Theorem~\ref{thm:opt-decomp}, we will use the following lemma.
\begin{lemma} \label{thm:basic-feasible}
For any directed graph $G = (V, E)$ and for {\em every} basic feasible solution $\mathbf{x} = \left(x_{\edge{v}{u}} \right)_{\edge{v}{u} \in E}$ to $\LP^{(+)}$ there exists a $t$ such that $G$ can be decomposed into a disjoint union of $t$ connected components (in the undirected sense) $G_i = (V_i, E_i)$ such that $|V_i| - 1 \le |E_i| \le |V_i|$ and $\mathbf{x}_{i} = \left( x_{\edge{v}{u}} \right)_{\edge{v}{u} \in E_i}$ is a basic feasible solution for LP~\eqref{eq:primal} on $G_i$ for every $i \in [t]$. Further, we have $\cup_{i=1}^{t}V(G_i) = V$ and $V(G_i) \cap V(G_j) = \emptyset$ \space $\forall i, j \in [t], i \neq j$. 
\end{lemma}
Assuming the above lemma is true, we prove Theorem~\ref{thm:opt-decomp}.
\begin{proof} [Proof of Theorem~\ref{thm:opt-decomp}]
Invoking Theorem~\ref{thm:bf-optimal}, there exists an optimal solution $\vx^*$ to $\LP^{(+)}$ that is basic feasible. Now, we can use Lemma~\ref{thm:basic-feasible} on $G$ to decompose $\vx^*$ into $t$ connected components (for some $t > 0$) such that $G_i = (V_i, E_i)$ such that $|V_i| - 1 \le |E_i| \le |V_i|$ and $\mathbf{x}^*_{i} = \left( x_{\edge{v}{u}} \right)_{\edge{v}{u} \in E_i}$ is a basic feasible solution for LP~\eqref{eq:primal} on $G_i$ for every $i \in [t]$. Our goal here is to argue that the solution $\vx^*_{i}$ (defined above) is optimal for every $i \in [t]$.

The proof is by contradiction. In particular, we show that if $\vx^*_{i}$ is not optimal for some $i \in [t]$, then we would contradict the fact that $\vx$ is optimal. Assume otherwise i.e., there exists an alternative optimal basic feasible solution\footnote{Note that there always exists such a solution by Theorem~\ref{thm:bf-optimal}} $\vx'_i$ such that 
\begin{align*}
\uprod{\edge{v}{u} \in E(G_i)} L_{\edge{v}{u}}^{x'_{\edge{v}{u}}} & < \uprod{\edge{v}{u} \in E(G_i)} L_{\edge{v}{u}}^{x^*_{\edge{v}{u}}} \numberthis \label{eq:opt-bf}.
\end{align*} 
We define $\vx'_{\edge{v}{u}} = \vx^*_{\edge{v}{u}}$ for every $\edge{v}{u} \in E \setminus E(G_i)$. The new solution $\vx'$ is (basic) feasible by our assumption for $\vx'_i$ and by construction for the remaining edges. Computing the objective value of $\vx'$, we have
\begin{align*}
\uprod{\edge{v}{u} \in E} L_{\edge{v}{u}}^{x'_{\edge{v}{u}}} & = \left( \uprod{\edge{v}{u} \in E \setminus E(G_i)}  L_{\edge{v}{u}}^{x'_{\edge{v}{u}}} \right) \cdot \uprod{\edge{v}{u} \in E(G_i)}  L_{\edge{v}{u}}^{x'_{\edge{v}{u}}}  \\
& < \left( \uprod{\edge{v}{u} \in E \setminus E(G_i)}  L_{\edge{v}{u}}^{x^*_{\edge{v}{u}}} \right) \cdot \uprod{\edge{v}{u} \in E(G_i)} L_{\edge{v}{u}}^{x^*_{\edge{v}{u}}},
\end{align*}
where the inequality follows from~\eqref{eq:opt-bf}. Since $\vx'$ has a smaller objective value than $\vx^*$, this contradicts the optimality of $\vx^*$ and completes the proof.
\end{proof}
To complete the proof of Theorem~\ref{thm:opt-decomp}, we prove Lemma~\ref{thm:basic-feasible}.
\begin{proof} [Proof of Lemma~\ref{thm:basic-feasible}]
We start by recalling the constraints of $\LP^{(+)}$ on $G$:
\begin{align*}
\usum{\edge{v}{u} \in E} x_{\edge{v}{u}} + \usum{\edge{u}{w} \in E} \frac{x_{\edge{u}{w}}}{p}   \ge 1 \qquad \forall u \in V\\ 
x_{\edge{v}{u}} \ge 0 \qquad \forall \edge{v}{u} \in E \numberthis \label{eq:primal-c}. 
\end{align*}
Then, the $(|V| + |E|) \times |E|$ constraint matrix $C$ is defined as follows (in the order of the constraints). The first $|V|$ rows are indexed by vertices in $V$ and the next $|E|$ rows are indexed by variables in $\LP^{(+)}$. The columns are indexed by variables of LP i.e., $(x_{\edge{v}{u}})_{\edge{v}{u} \in E}$. For every $u \in V$, note that $C \left[u, x_{\edge{v}{u}} \right] = 1$ and $C\left[u, x_{\edge{u}{w}} \right] = \frac{1}{p}$, where $v, w \in V$ and $\edge{v}{u}, \edge{u}{w} \in E$. For every edge $\edge{v}{u} \in E$, note that $C \left[\edge{v}{u} , x_{\edge{v}{u}} \right] = 1$. All the remaining entries in $C$ are $0$.
	
Consider any basic feasible solution $\mathbf{x} = \left(x_{\edge{v}{u}} \right)_{\edge{v}{u} \in E}$ to $\LP^{(+)}$. We remove every edge $\edge{v}{u} \in E$ from $G$ that satisfies $x_{\edge{v}{u}} = 0$. Note that this process does not remove any vertex from $G$ since 
\begin{align*}
\usum{\edge{v}{u} \in E} x_{\edge{v}{u}} + \usum{\edge{u}{w} \in E} \frac{x_{\edge{u}{w}}}{p}  \ge 1
\end{align*}
for every $u \in V(G)$ (by definition of $\mathbf{x}$). 
	
Let the resulting graph after this process be $G_{1}$ and for simplicity, we assume $G_1$ is a single connected component in the undirected sense (else we consider the components individually), which implies 
\begin{align*}
|E(G_1)| \ge |V| - 1 \numberthis \label{eq:bf-conn}
\end{align*}
and we have $x_{\edge{v}{u}} > 0$ for every edge $\edge{v}{u} \in E(G_1)$. Recall that each edge in $E \setminus E(G_1)$ has $x_{\edge{v}{u}} = 0$ and as a result, the solution $\mathbf{x}$ to $\LP^{(+)}$ has $|E \setminus E(G_1)|$ tight constraints among the last $|E|$ rows (denoted by $S(E \setminus E(G_1))$). Recall that $\mathbf{x}$ is a basic feasible solution and it has exactly $|E|$ tight constraints (by Definition~\ref{def:basic-feasible}). In particular, this implies there are $|E| - |E \setminus E(G_1)| = |E(G_1)|$ more tight constraints. Further, note that these constraints are of the form 
\begin{align*}
\usum{\edge{v}{u} \in E} x_{\edge{v}{u}} + \usum{\edge{u}{w} \in E} \frac{x_{\edge{u}{w}}}{p}  = 1
\end{align*}
for some subset of vertices in $V$. We note these tight constraints by $S(G_1)$ and it follows that $|S(G_1)| = |E(G_1)|$ (as discussed above). 
	
Let $S = S(G_1) \cup S(E \setminus E(G_1))$ denote the set of all such tight constraints. Consider the matrix $C_{S}$, which is matrix $C$ projected down to constraints (i.e., rows) in $S$. By definition of the basic feasible solution, $C_{S}$ has rank $|E|$ and we claim that $C_{S}$ can be written as 
\begin{equation*}
\delimitershortfall=0pt
\left(\begin{array}{c|c}
C_{S(G_1)} & \mathbf{0} \\
\underbrace{\mathbf{0}}_{E(G_1)}  &\underbrace{C_{S(E \setminus E(G_1))}}_{E \setminus E(G_1)}
\end{array}
\right)
\end{equation*}
Here, $C_{S(G_1)}$ and $C_{S(E \setminus E(G_1))}$ denote the submatrices corresponding to tight constraints in $S(G_1)$ and $S(E \setminus E(G_1))$ respectively. The block structure above follows from the fact that the edges in $E \setminus E(G_1)$ do not belong in $E(G_1)$ (by construction) and as a result, the constraints in $S(G_1)$ do not contain edges in $E \setminus E(G_1)$. Further, by definition of constraints in $S(E \setminus E(G_1))$, $C_{S(E \setminus E(G_1))}$ forms an identity matrix, resulting in
\begin{align*}
C_{S} = 
\begin{bmatrix}
C_{S(G_1)} & \mathbf{0} \\
\mathbf{0} & I_{S(E \setminus E(G_1))} 
\end{bmatrix}.
\end{align*}
Thus, we can write 
\begin{align*}
\rank(C_{S}) & = \rank \left(C_{S(G_1)} \right) + \rank(I_{S(E \setminus E(G_1))})   \\
& = \rank \left(C_{S(G_1)} \right) + |E \setminus E(G_1)|.
\end{align*}
Since $\rank(C_{S}) = |E|$, we have
\begin{align*}
\rank \left(C_{S(G_1)} \right) = |E| - |E \setminus E(G_1)| = |E(G_1)|. 
\end{align*}
Further, $C_{S(G_1)}$ is a $|S(G_1)| \times |E(G_1)|$ matrix, which implies
\begin{align*}
|E(G_1)| = \rank \left( C_{S(G_1)} \right) \le \min(|S(G_1)|, |E(G_1)|) \le |S(G_1)| \le |V|. 
\end{align*}
Combining this with~\eqref{eq:bf-conn}, we get
\begin{align*}
|V| - 1 \le |E(G_1)| \le |V|,
\end{align*}
as required. Since $\vx_1 = \left( x_{\edge{v}{u}} \right)_{\edge{v}{u} \in E(G_1)}$ is a feasible solution to $\LP^{(+)}$ on $G_1$, where $|E(G_1)|$ of them (of the following kind) are tight:
\begin{align*}
\left( \sum_{\edge{v}{u} \in E(G_1)} x_{\edge{v}{u}} \right) + \left( \sum_{\edge{u}{w} \in E(G_1)} \frac{x_{\edge{u}{w}}}{p} \right) > 1,
\end{align*}
 which follows from the fact that $\vx$ was basic feasible to start with. Further, the matrix $C_{S(G_1)}$ has rank $E(G_1)$, as shown earlier. Thus, we have that $\vx_1$ is basic feasible.

When there are multiple components, we consider each (undirected) connected component individually after edges satisfying $x_{\edge{v}{u}} = 0$ are removed. Since the (undirected) connected components are pairwise vertex disjoint (if not, we can collapse them into a single connected component, again in the undirected sense), the connected components are maximal. Applying a similar argument as above, we would have that the matrices $C_{S(G_i)}$ (along with $I_{S(E \setminus \ucup{i \in [t]}E(G_i))}$) are all block diagonal and we have
\begin{align*}
\rank(C_{S}) & = \usum{i  \in [t]} \left( \rank ((C_{S(G_i)}) \right) + \rank \left(I_{S(E \setminus \ucup{i \in [t]}E(G_i))} \right),
\end{align*}
which in turn implies (using $\rank(C_{S}) = |E|$)
\begin{align*}
\usum{i  \in [t]} \left( \rank (C_{S(G_i)}) \right) & = |E| - | E \setminus \ucup{i \in [t]} E(G_i)| \\
& = \usum{i \in [t]}E(G_i).
\end{align*}
Since $\rank(C_{S(G_i)}) \le \min(|E(G_i)|, |V(G_i)|)$ (by definition), the above implies $\rank(C_{S(G_i)}) = |E(G_i)|$ for every $i \in [t]$. To complete the proof, we argue that the solution $\mathbf{x}_{i} = \left( x_{\edge{v}{u}} \right)_{\edge{v}{u} \in E(G_i)}$ is basic feasible for $\LP^{(+)}$ on $G_i$ for every $i \in [t]$. In particular, there are $|E(G_i)|$ variables and $|V(G_i)|$ constraints in it, of which $|E(G_i)|$ are tight. The remaining $|V(G_i)| - |E(G_i)|$ constraints are of the form
\begin{align*}
\left( \sum_{\edge{v}{u} \in E(G_i)} x_{\edge{v}{u}} \right) + \left( \sum_{\edge{u}{w} \in E(G_i)} \frac{x_{\edge{u}{w}}}{p} \right) > 1,
\end{align*}
which follows from (basic) feasibility of $\vx$. Thus, we have shown that $\mathbf{x}_{i}$ is a feasible solution to $\LP^{(+)}$ on $G_i$ and it is basic feasible since $\rank(C_{S(G_i)}) = |E(G_i)|$. This completes the proof.
\end{proof}

Next, we extend this result to the case when we are given both $\ell_{p}$ (for a fixed $p$) and $\ell_{\infty}$ constraints for the same relation as well. 

\subsection{Proof of Corollary~\ref{cor:dhigh-struct}} \label{sec:dhigh-struct}
We restate $\LP^{(+)}$ for this case and Corollary~\ref{cor:dhigh-struct} before proving the latter.
\begin{align*}
& \min \left( \usum{(v, u) \in E} x_{\edgeud{v}{u}} \log(L) +  \usum{\edge{v}{u} \in E} z_{\edge{v}{u}} \log(d) \right) \tag{$\LP^{(+)}$}\\
& \text{ s.t. } \left( \usum{e = (v, u) \ni u} x_{v, u}  \right) + \left( \usum{\edge{v}{u} \in E} z_{\edge{v}{u}} \right) \ge 1 \quad \forall u \in V \numberthis \label{eq:struct-gen-covering} \\
& x_{\edge{v}{u}}, z_{\edge{v}{u}} \ge 0 \quad \forall \edge{v}{u} \in E \numberthis \label{eq:struct-gen-primal}.
\end{align*}
\begin{corollary}
For any $G = (V, E)$, there exists an optimal solution $(\vx^*, \vz^*) = (x^*_{\edge{v}{u}}, z^*_{\edge{v}{u}})_{\edge{v}{u} \in E}$ to $\LP^{(+)}$ on $G$ that can be decomposed into a disjoint union of $t$ (some $t > 0$) connected components (in the undirected sense) $G_i = (V_i, E_i)$ with 
\begin{align*}
|V_i| - 1 \le |Q(E(G_i))| \le |V_i|, \text{ where } 
\end{align*}
\begin{align*}
Q(E(G_i))  = \{x_{\edgeud{v}{u}}: x_{\edgeud{v}{u}} \neq 0 , \edgeud{v}{u} \in E(G_i)\} \cup \{z_{\edge{v}{u}}: z_{\edge{v}{u}} \neq 0, \edge{v}{u} \in E(G_i) \}.
\end{align*}
and $(\mathbf{x}^*_{i}, \mathbf{z}^*_{i}) = \left( x_{\edgeud{v}{u}}, z_{\edge{v}{u}} \right)_{\edgeud{v}{u} \in E(G_i)}$ is an optimal basic feasible solution for $\LP^{(+)}$ on $G_i$ for every $i \in [t]$. Further, we have $\cup_{i=1}^{t}V(G_i) = V$ and $V(G_i) \cap V(G_j) = \emptyset$ \space $\forall i, j \in [t], i \neq j$. The following is true:
\begin{align*}
\TJ & = \times_{i \in [t]} \TJ(G_i).
\end{align*}
\end{corollary}
Before proving this corollary, we first state an extremal property on optimal basic feasible solutions for each cyclic $G_i: i \in [t]$ (from Corollary~\ref{cor:dhigh-struct}). 
\begin{property} \label{assump:dhigh}
Every optimal basic feasible solution $$(\vx^*_i, \vz^*_i) = (x^*_{\edgeud{v}{u}}, z^*_{\edge{v}{u}})_{\edgeud{v}{u} \in E(G_i)}$$ to $\LP^{(+)}$ on $G_i: i \in [t]$ is such that for every cyclic $G_i$, exactly one of $x^*_{\edgeud{v}{u}}$ (or) $z^*_{\edge{v}{u}}$ is non-zero for every (undirected) edge $\edge{v}{u} \in E(G_i)$.
\end{property}
\begin{proof}[Proof of Property~\ref{assump:dhigh}]
Among all optimal basic feasible solutions (OBFS) $(\vx^*, \vz^*)$ to $\LP^{(+)}$ on $G$, we are ruling out every OBFS, where 
\begin{align*}
\exists \edgeud{v}{u}, \edge{v}{u} \in E(G_i) \text{ s.t. } x^*_{\edgeud{v}{u}} = 0 \text { and } z^*_{\edge{v}{u}} = 0.
\end{align*}
For the remaining OBFS, we have
\begin{align*}
\forall \edgeud{v}{u}, \edge{v}{u} \in E(G_i), x^*_{\edgeud{v}{u}} > 0 \text { or } z^*_{\edge{v}{u}} > 0.
\end{align*}
Note that there always exists a OBFS of this kind using Corollary~\ref{cor:dhigh-struct} and as a result, the property can always be satisfied.
\end{proof}

We start by stating the corresponding version of Lemma~\ref{thm:basic-feasible} for this case.
\begin{lemma} \label{thm:basic-feasible-gen}
For any $G = (V, E)$ and for {\em every} basic feasible solution $(\mathbf{x}, \mathbf{z}) = \left(x_{\edgeud{v}{u}}, z_{\edge{v}{u}} \right)_{\edge{v}{u} \in E}$ to $LP^{(+)}$ on $G$, there exists a $t$ such that $G$ can be decomposed into a disjoint union of $t$ connected components (in the undirected sense) $G_i = (V_i, E_i)$ such that $|V_i| - 1 \le |E_i| \le |V_i|$ and $(\vx_{i}, \vz_{i}) = \left( x_{\edgeud{v}{u}}, z_{\edge{v}{u}}, z_{\edge{u}{v}} \right)_{\edge{v}{u} \in E_i}$ is a basic feasible solution to $LP^{(+)}$ on $G_i$ for every $i \in [t]$. Further, we have $\cup_{i=1}^{t}V(G_i) = V$ and $V(G_i) \cap V(G_j) = \emptyset$ \space $\forall i, j \in [t], i \neq j$. 
\end{lemma}
Assuming the above lemma is true, we can do the same proof as the Proof of Theorem~\ref{thm:opt-decomp} to complete the proof (and is omitted). We would like to mention here that the proof of Lemma~\ref{thm:basic-feasible-gen} is very similar to the proof of Lemma~\ref{thm:basic-feasible}.
\begin{proof} [Proof of Lemma~\ref{thm:basic-feasible-gen}]
We start by recalling the constraints of $\LP^{(+)}$ on $G$:
\begin{align*}
& \left( \usum{e = (v, u) \ni u} x_{v, u}  \right) + \left( \usum{\edge{v}{u} \in E} z_{\edge{v}{u}} \right) \ge 1 \quad \forall u \in V  \\
& x_{\edge{v}{u}}, z_{\edge{v}{u}}, z_{\edge{u}{v}} \ge 0 \quad \forall \edge{v}{u} \in E 
\end{align*}
We define the $(|V| + 2 \cdot |E|) \times (2 \cdot |E|)$ constraint matrix $C$ as follows (in the order of the constraints). The first $|V|$ rows are indexed by vertices in $V$ and the next $2 \cdot |E|$ rows are indexed by variables in lexicographic order in the LP. The columns are lexicographically indexed by variables of the LP i.e., $(x_{\edgeud{v}{u}}, z_{\edge{v}{u}})_{\edgeud{v}{u} \in E}$, consistent with the order of $2 \cdot |E|$ rows discussed above. For every $u \in V$, note that $C \left[u, x_{\edgeud{v}{u}} \right] = C \left[u, z_{\edge{v}{u}}\right] = 1$ and $C\left[u, x_{\edgeud{u}{w}} \right] = 1$, where $v, w \in V$ and $\edgeud{v}{u}, \edgeud{u}{w} \in E$ \footnote{Note here that for a given $u$, we have a non-zero entry in only $C[u, z_{\edge{v}{u}}] = 1$ (and not $C[v, z_{\edge{v}{u}}]$).}. For every edge $\edgeud{v}{u} \in E$, note that $C \left[\edgeud{v}{u} , x_{\edgeud{v}{u}} \right] = 1, C \left[\edge{v}{u} , z_{\edge{v}{u}} \right] = 1$. All the remaining entries in $C$ are $0$.

Consider any basic feasible solution $(\mathbf{x}, \vz) = \left(x_{\edgeud{v}{u}}, z_{\edge{v}{u}} \right)_{\edgeud{v}{u} \in E}$ to $\LP^{(+)}$. We remove every edge $\edgeud{v}{u} \in E$ from $G$ such that $x_{\edgeud{v}{u}} = z_{\edge{v}{u}} = 0$. Note that this process does not remove any vertex from $G$ since 
\begin{align*}
\left( \usum{e = (v, u) \ni u} x_{v, u}  \right) + \left( \usum{\edge{v}{u} \in E} z_{\edge{v}{u}} \right) \ge 1 
\end{align*}
for every $u \in V$ (by definition of $\mathbf{x}$ and $\mathbf{z}$). 

Let the resulting graph after this process be $G_{1}$ and for simplicity, we assume $G_1$ is a single connected component in the undirected sense (else we consider the components individually), which implies 
\begin{align*}
|E(G_1)| \ge |V| - 1 \numberthis \label{eq:bf-conn-2}
\end{align*}
and we have either $x_{\edgeud{v}{u}} > 0$ or $z_{\edge{v}{u}} > 0$ for every edge $\edgeud{v}{u} \in E(G_1)$. Recall that each edge in $E \setminus E(G_1)$ has $x_{\edgeud{v}{u}} = z_{\edge{v}{u}} = 0$ and as a result, the solution $(\mathbf{x}, \vz)$ to $\LP^{(+)}$ has $2 \cdot |E \setminus E(G_1)|$ tight constraints for each edge not in $E(G_1)$ and we denote this set of tight constraints by $S_{\vx, \vz}(E \setminus E(G_1))$. Since $(\mathbf{x}, \vz)$ is a basic feasible solution and it has exactly $2 \cdot |E|$ tight constraints. In particular, this implies there are exactly
\begin{align*}
2 \cdot |E| - 2 \cdot |E \setminus E(G_1)| = 2 \cdot |E(G_1)|
\end{align*}
more tight constraints. We argue that $\rank \left( C_{S_{\vx, \vz}(G_1)}\right)  \le 2 \cdot |E(G_1)|$, where the submatrix $C_{S_{\vx, \vz}(G_1)}$ corresponds to the tight constraints in $E(G_1)$

We note that the tight constraints in $C_{S_{\vx, \vz}(G_1)}$ can be of two types: the first one is of the form
\begin{align*}
\left( \usum{e = (v, u) \ni u} x_{v, u}  \right) + \left( \usum{\edge{v}{u} \in E} z_{\edge{v}{u}} \right) = 1
\end{align*}
for some subset of vertices in $V$. We denote these tight constraints by $S^{1}_{\vx, \vz}(G_1)$. The second type is of the form $x_{\edgeud{v}{u}} = 0$ or $z_{\edge{v}{u}} = 0$ (both cannot be true simultaneously since we would have removed the edge when that is the case), which we denote by $S^{2}_{\vx, \vz}(G_1)$. Note that by definition, for each edge $\edgeud{v}{u} \in E(G_1)$, we have either $x_{\edgeud{v}{u}} > 0$ (or) $z_{\edge{v}{u}} > 0$. It follows that $|S^{1}_{\vx, \vz}(G_1)|+ |S^{2}_{\vx, \vz}(G_1)| = 2 \cdot |E(G_1)|$ (as discussed above).

Consider the matrix $C_{S_{\vx, \vz}(G_1)}$, which is $C$ projected down to constraints in $S^{1}_{\vx, \vz}(G_1) \cup S^{2}_{\vx, \vz}(G_1)$. By definition of the basic feasible solution, $C_{S_{\vx, \vz}(G_1)}$ can be written as 
\begin{equation*}
\left(\begin{array}{cc}
C_{S^{1}_{\vx, \vz}(G_1)} \\
C_{S^{2}_{\vx, \vz}(G_1)}
\end{array}
\right).
\end{equation*}
Note that this implies
\begin{align*}
\rank \left(C_{S_{\vx, \vz}(G_1)} \right) \le \rank \left(C_{S^{1}_{\vx, \vz}(G_1)} \right) + 
\rank  \left(C_{S^{2}_{\vx, \vz}(G_1)} \right) \numberthis \label{eq:rank-s1-s2} \\
\end{align*}
Recall that $S^{2}_{\vx, \vz}(G_1)$ contains only constraints of the form $x_{\edgeud{v}{u}} = 0$ (or) $z_{\edge{v}{u}} = 0$ for a subset of edges $\edgeud{v}{u} \in E$ and as a result, $C_{S^{2}_{\vx, \vz}(G_1)}$ forms an identity matrix. 
Since $\rank \left(C_{S_{\vx, \vz}(G_1)} \right) = 2 \cdot |E(G_1)|$ (since $(\vx,\vz)$ is a basic feasible solution) and 
\begin{align*}
\rank  \left(C_{S^{2}_{\vx, \vz}(G_1)} \right) = |S^{2}_{\vx, \vz}(G_1)|  = 2 \cdot |E(G_1)| - |Q(E(G_1))|,
\end{align*}
which follows from our definition of $|Q(E(G_1))|$. More specifically there are $2 \cdot |E(G_1)|$ variables $x_{\edgeud{v}{u}}$  $z_{\edge{v}{u}}$ over all $\edgeud{u}{v}\in E(G_1)$ and the set of non-zero ones (which cannot result in tight constraints in $S^{2}_{\vx, \vz}(G_1)$) are exactly captured by $Q(E(G_1))$. Thus,~\eqref{eq:rank-s1-s2} implies that
\begin{align*}
\rank \left(C_{S^{1}_{\vx, \vz}(G_1)} \right) & \ge 2 \cdot |E(G_1)| - \left(2 \cdot |E(G_1)| - |Q(E(G_1))|\right) \\
& = |Q(E(G_1))|.
\end{align*}
Since $C_{S^{1}_{\vx, \vz}(G_1)}$ is at most a $|V| \times |S^{1}_{\vx, \vz}(G_1)|$ matrix, we have
\begin{align*}
|Q(E(G_1))| \le \rank \left(C_{S^{1}_{\vx, \vz}(G_1)} \right) \le \min \left(|V|,  |S^{1}_{\vx, \vz}(G_1)| \right) \le |V| \numberthis \label{eq:rank-s1}.
\end{align*}
Finally, we have
\begin{align*}
|Q(E(G_1))| \ge |E(G_1)| \ge |V| - 1 \numberthis \label{eq:rank-s1-lb},
\end{align*}
where the first inequality follows from the definition of $Q(E(G_1))$ as a union of sets $\{x_{\edgeud{v}{u}}: x_{\edgeud{v}{u}} \neq 0 , \edgeud{v}{u} \in E(G_i)\}$ and $\{z_{\edge{v}{u}}: z_{\edge{v}{u}} \neq 0, \edge{v}{u} \in E(G_i) \}$, which in turn implies that every edge $\edgeud{u}{v}\in E(G_1)$ gets counted at least once in $Q(E(G_1))$ (either as $x_{\edgeud{u}{v}}\ne 0$ or $z_{\edge{u}{v}}\ne 0$).
In particular, we can now combine~\eqref{eq:rank-s1-lb} with~\eqref{eq:rank-s1} to show that
\begin{align*}
|V| - 1 \le |Q(E(G_1))| \le |V|.
\end{align*}
Note that~\eqref{eq:rank-s1-lb} implies there exists at most one edge $\edge{v}{u} \in E_1$ such that {\em both} $x_{\edgeud{u}{v}}\ne 0$ and $z_{\edge{u}{v}}\ne 0$. 

When there are multiple components $G_i: i \in [t]$, we consider each (undirected) connected component individually after edges $\edge{v}{u} \in E$ of the form $x_{\edgeud{v}{u}} = z_{\edge{v}{u}} = 0$ are removed from $G$. Note that the (undirected) connected components are pairwise vertex-disjoint (if not, we can collapse them into one connected component in the undirected sense) and as a result, the connected components we obtained are maximal. Applying a similar argument as above, we would have that the matrices $C_{S^{1}_{\vx, \vz}(G_i)}$ and $C_{S^{2}_{\vx, \vz}(G_i)}$ for every $i \in [t]$ (along with $I_{S(E \setminus \ucup{i \in [t]}E(G_i))}$) are all block diagonal and we have
\begin{align*}
\rank(C_{S}) & = \usum{i  \in [t]} \left( \rank (C_{S^{1}_{\vx, \vz}(G_i)}) + \rank (C_{S^{2}_{\vx, \vz}(G_i)}) \right) + \rank \left(I_{S(E \setminus \ucup{i \in [t]}E(G_i))} \right),
\end{align*}
which in turn implies (using $\rank(C_{S}) = 2|E|$ and $\rank (C_{S^{2}_{\vx, \vz}(G_i)}) = 2 |E(G_i)| - |Q(E(G_i))|$)
\begin{align*}
\usum{i  \in [t]} \left(\rank (C_{S^{1}_{\vx, \vz}(G_i)}) \right) & = 2|E| - | E \setminus \ucup{i \in [t]} E(G_i)| - \usum{i \in [t]} \left(  2 |E(G_i)| - |Q(E(G_i))|  \right) \\
& = \usum{i \in [t]} |Q(E(G_i))|.
\end{align*}
Since $\rank (C_{S^{1}_{\vx, \vz}(G_i)}) \le \min(|Q(E(G_i))|, |V(G_i)|)$ (by definition), the above implies $\rank(C_{S^{1}_{\vx, \vz}(G_i)}) = |Q(E(G_i))|$ for every $i \in [t]$.

To complete the proof, we argue that the solution $(\mathbf{x}_{i} \vz_{i}) = \left( x_{\edgeud{v}{u}}, z_{\edge{v}{u}} \right)_{\edge{v}{u} \in E(G_i)}$ is basic feasible for $\LP^{(+)}$ on $G_i$ for every $i \in [t]$.  Note that there are $2 \cdot |E(G_i)|$ variables and $|V(G_i)| + 2 \cdot |E(G_i)|$ constraints in it of which $|Q(E(G_i))|$ are tight i.e., we have either
\begin{align*}
\left( \usum{e \ni u: e = (v, u) \in E} x_{v, u}  \right) + \left( \usum{\edge{v}{u} \in E(G_i)} z_{\edge{v}{u}} \right) > 1 \text{ for some } u \in V \\
\text { (or) }x_{\edgeud{v}{u}} > 0 \\
\text { (or) } z_{\edge{v}{u}} > 0,
\end{align*}
which follows from basic feasibility of $(\vx, \vz)$. Thus, we have shown that $(\mathbf{x}_{i}, \vz_{i})$ is a feasible solution to $\LP^{(+)}$ on $G_i$ and it is basic feasible since $C_{S_{\vx, \vz}(G_i)}$ has rank exactly $2 \cdot |E(G_1)|$. This completes the proof.
\end{proof}

%% file: appendix_sec4.tex
\section{Proof of Theorem~\ref{claim:gen-acyclic}} \label{sec:gen-acyclic}

\begin{proof}[Proof of Theorem~\ref{claim:gen-acyclic}]
We first prove the following result using Algorithm~\ref{algo:gen-ub} (without assuming anything about its runtime):
\begin{equation} \label{eq:J-d-bound}
\abs{\TJ(\bd)} = |J_{n}|\le \B(\bd, G).
	\end{equation}
	In particular, we will argue that for each $i\in [n]$ and $\vt\in J_{i-1}$, we have at the end of the iteration $i$:
	\begin{align*}
	\abs{P_i(\vt)}\le \cD_{u_i}(\bd).
	\end{align*}
	Note that the above combined with the defnition of $P_i(\vt)$ implies
	\begin{align*}
	|J_{i}| \le \uprod{(u_1, \dots u_i)} \cD_{u_i}(\bd).
	\end{align*}
	When $i = n$, we have
	\begin{align*}
	|J_{n}| \le \uprod{u \in V} \cD_{u}(\bard).
	\end{align*}
	Recall that the RHS is exactly $\B(\bd, G$).
	
	To prove~\eqref{eq:J-d-bound}, it suffices (by definition of $\cD_{u_i}(\bd)$) to prove for every $i\in [n]$ and $\vt\in \TJ_{i-1}$:
	\begin{equation} \label{eq:L-in-bound}
	\abs{P_{\text{in}}(i,\vt)} \le \umin{\edge{v}{u_i} \in E} d_{\edge{v}{u_i}}
	\end{equation}
	and
	\begin{equation} \label{eq:L-out-bound}
	\abs{P_{\text{out}}(i)} \le \umin{\edge{u_i}{w} \in E} \frac{2^{p} \cdot L_{{\edged{u_i}{w}}}^{p}}{d_{\edge{u_i}{w}}^{p}}.
	\end{equation}
	Indeed,~\eqref{eq:L-in-bound} just follows from the definition of the degree configuration $\bd$. Finally,~\eqref{eq:L-out-bound} follows by applying Lemma~\ref{lemma:node-bound} on each edge $\edge{u_i}{w} \in E$.
	
	Next, we argue the correctness of our algorithm by showing that $\TJ(\bd) = J_{n} = \ujoin{\edge{v}{u} \in E} R^{d_{\edge{v}{u}}}_{\edge{v}{u}}$. Let $J_{n} \subset\ujoin{\edge{v}{u} \in E} R^{d_{\edge{v}{u}}}_{\edge{v}{u}}$. Note that this implies there exists a tuple $\vt \in \ujoin{\edge{v}{u} \in E} R^{d_{\edge{v}{u}}}_{\edge{v}{u}}$ such that $\vt \notin J_{n}$. Based on Algorithm~\ref{algo:gen-ub}, this would imply there exists at least one attribute $u_i: i \in [n]$ such that $\pi_{u_i}(\vt) \notin P_{i}(\vt_{u_1, \dots, u_{i -1}})$, which implies $\vt \notin \ujoin{\edge{v}{u} \in E} R^{d_{\edge{v}{u}}}_{\edge{v}{u}}$, contradicting our earlier assumption. To complete the argument, we argue the reverse as well i.e., $J_{n} \supset \ujoin{\edge{v}{u} \in E} R^{d_{\edge{v}{u}}}_{\edge{v}{u}}$. Assume there exists a tuple $\vt \in J_{n}$ such that $\vt \notin \ujoin{\edge{v}{u} \in E} R^{d_{\edge{v}{u}}}_{\edge{v}{u}}$. Note that there exists at least one attribute $u_i: i \in [n]$ and a relation $R_{\edge{u_i}{\cdot}}^{d_{\edge{u_i}{\cdot}}}/R_{\edge{\cdot}{u_i}}^{d_{\edge{\cdot}{u_i}}}$ for some $\edge{u_i}{\cdot}, \edge{\cdot}{u_i} \in E$ such that $\pi_{u_i}(\vt) \not \in \pi_{u_i} \left(R_{\edge{u_i}{\cdot}}^{d_{\edge{u_i}{\cdot}}} \right), \pi_{u_i} \left(R_{\edge{\cdot}{u_i}}^{d_{\edge{\cdot}{u_i}}} \right)$. This results in a contradiction since $P_{i}(\vt'_{u_1, \dots, u_{i - 1}})$ (by definition) for any tuple $\vt' \in J_{i - 1}$ has only values from $\pi_{u_i} \left(R_{\edge{u_i}{\cdot}}^{d_{\edge{u_i}{\cdot}}} \right)$ and $\pi_{u_i} \left(R_{\edge{\cdot}{u_i}}^{d_{\edge{\cdot}{u_i}}} \right)$. Thus, we have $\TJ(\bd) = J_{n} = \ujoin{\edge{v}{u} \in E} R^{d_{\edge{v}{u}}}_{\edge{v}{u}}$, as required.
	
	Finally, we argue that Algorithm~\ref{algo:gen-ub} runs in time $O(\B(\bd, G))$. In order to do this, we argue that for every $i \in [n]$ and $\vt \in J_{i - 1}$, $\abs{P_i(\vt)}$ (defined above) can computed in time $O(|E| \cdot \abs{P_i(\vt)})$ and the final runtime follows as a result. Recall that each relation is stored in a two level B-tree-like index structure consistent with the topological ordering $u_1, \dots, u_n$. In particular, for any relation $R_{\edge{u_i}{u_j}}$ (with $j > i$), the first level of its B-tree is indexed by $(u_i)$ (i.e., all values ) and the second level is indexed by $(u_j, \val_{u_j})$, where $\val_{u_j} \in \Dom(u_j)$. As a result, for a fixed $i \in [n]$ and $\vt \in J_{i - 1}$, note that we can access the set
	\begin{align*}
	S_{(\vt[v] , u_i)} = \inset{y: (\vt[v],y)\in R_{\edge{v}{u_i}}^{d_{\edge{v}{u_i}}}}
	\end{align*}
	directly using the B-tree index on $(v, \vt[v])$ and similarly, we can access the set
	\begin{align*}
	S_{\edge{u_i}{w}} = \pi_{u_i}\left(R_{\edge{u_i}{w}}^{d_{\edge{u_i}{w}}}\right)
	\end{align*}
	directly using the B-tree index on $(u_i)$. In particular, we are now computing a set intersection on sorted sublists 
	\begin{align*}
	(S_{(\vt[v] , u_i)}, S_{\edge{u_i}{w}}),
	\end{align*}
	whose size
	\begin{align*}
	\left | \left(\cap_{(\vt[v] , u_i): \edge{v}{u_ i} \in E} S_{(\vt[v] , u_i)} \right) \cap \left( \cap_{\edge{u_i}{w} \in E }  S_{(u_i, w)}\right)  \right|  \numberthis \label{eq:set-intersection}
	\end{align*}
	is exactly $|P_{i}(\vt)|$. Since~\eqref{eq:set-intersection} can be computed in time $|E| \cdot |P_{i}(\vt)|$ (where recall that we are working in the RAM model), this completes the proof.
\end{proof}

%% file: appendix_sec4_1.tex
\section{Missing Details in Section~\ref{sec:g-acyclic}} \label{app:g-acyclic}

\subsection{Proof of Lemma~\ref{thm:acyclic-bound-lp}} \label{app:acyclic-bound-lp}
\begin{proof} [Proof of Lemma~\ref{thm:acyclic-bound-lp}]
Invoking Theorem~\ref{claim:gen-acyclic}, we have
\begin{align*}
\B(\bard, G) & = \uprod{u \in V} \min \left(  (d_{\edge{v}{u}})_{\edge{v}{u} \in E},  \left( \frac{2^{p} \cdot L_{{\edged{u}{w}}}^{p}}{d_{\edge{u}{w}}^{p}} \right)_{\edge{u}{w} \in E}   \right) \numberthis \label{eq:gen-holder-1} \\ 
& \le 2^{p |V|} \uprod{u \in V} \min \left(  (d_{\edge{v}{u}})_{\edge{v}{u} \in E},  \left( \frac{ L_{{\edged{u}{w}}}^{p}}{d_{\edge{u}{w}}^{p}} \right)_{\edge{u}{w} \in E}   \right) \\
& \le 2^{p |V|} \cdot \uprod{u \in V} \left( \left( \uprod{\edge{v}{u} \in E} d_{\edge{v}{u}}^{x_{\edge{v}{u}} + z_{\edge{v}{u}}} \right) \cdot \left(\uprod{\edge{u}{w} \in E} \left( \frac{L_{{\edged{u}{w}}}^{p}}{d_{\edge{u}{w}}^{p}} \right)^{\frac{x_{\edge{u}{w}}}{p}}  \right) \right) \numberthis \label{eq:gen-holder-2} \\ 
& =  2^{p |V|} \cdot \uprod{u \in V} \left( \uprod{\edge{v}{u} \in E}  d_{\edge{v}{u}}^{x_{\edge{v}{u}}} \cdot d_{\edge{v}{u}}^{z_{\edge{v}{u}}} \cdot \uprod{\edge{u}{w} \in E} \frac{L_{{\edged{u}{w}}}^{x_{\edge{u}{w}}}}{d_{\edge{u}{w}}^{x_{\edge{u}{w}}}} \right) \\
& = 2^{p |V|} \cdot \uprod{\edge{v}{u} \in E}   \left( d_{\edge{v}{u}}^{z_{\edge{v}{u}}} \cdot L_{\edged{v}{u}}^{x_{\edge{v}{u}}} \right) \numberthis \label{eq:gen-holder-last}.
\end{align*}
In the above,~\eqref{eq:gen-holder-1} follows by definition of $\B(\bard, G)$. We argue~\eqref{eq:gen-holder-2} next. Note that
\begin{align*}
& \min \left(  (d_{\edge{v}{u}})_{\edge{v}{u} \in E},  \left( \frac{L_{{\edged{u}{w}}}^{p}}{d_{\edge{u}{w}}^{p}}  \right)_{\edge{u}{w} \in E}   \right) \\
& \le  \left( \uprod{\edge{v}{u} \in E} d_{\edge{v}{u}}^{x_{\edge{v}{u}} + z_{\edge{v}{u}}} \right) \cdot \left(\uprod{\edge{u}{w} \in E} \left( \frac{L_{{\edged{u}{w}}}^p}{d_{\edge{u}{w}}^{p}} \right)^{\frac{x_{\edge{u}{w}}}{p}} \right) \quad \forall u \in V
\numberthis \label{eq:zero-sum-1}
\end{align*}
for any $\left( (x_{\edge{v}{u}}, z_{\edge{v}{u}})_{\edge{v}{u} \in E}, (x_{\edge{u}{w}})_{\edge{u}{w} \in E}\right)$ such that 
\begin{itemize}
	\item{$x_{\edge{v}{u}}, z_{\edge{v}{u}} \ge 0$ for every $\edge{v}{u} \in E$ follows from~\eqref{eq:gen-primal}.}
	\item{
		\begin{align*}
		\usum{\edge{v}{u} \in E} \left( x_{\edge{v}{u}} + z_{\edge{v}{u}} \right) +\usum{\edge{u}{w} \in E} \frac{x_{\edge{u}{w}}}{p} \ge 1. \numberthis \label{eq:gen-holder-cons}
		\end{align*}
		Note that this is the same as~\eqref{eq:gen-covering}.}
\end{itemize}
We can now invoke Lemma~\ref{lemma:zero-sum} to get~\eqref{eq:zero-sum-1}. Further,~\eqref{eq:gen-holder-last} follows by noting that since $G$ is a DAG, each edge $\edge{u}{w} \in E$ occurs exactly twice -- once as a bound of $w$ and the other as a bound for $u$ (both following from Theorem~\ref{claim:gen-acyclic}). As a result, we can cancel the $d_{\edge{u}{w}}$ term to get~\eqref{eq:gen-holder-last}. 
\end{proof}

\subsection{Proof of ~\eqref{eq:gen-ub} for $p \in (|V| - 1, \infty]$} \label{sec:genp-ub}
In this section, we prove the following upper bound for $|\TJ|$ for any $p$ in $(|V| - 1, \infty]$:
\begin{align*}
|\TJ| & \le 2^{p |V|} \cdot \inparen{(p |V|)^2}^{(p |V|)^2} \cdot c^{|E|} \cdot \uprod{\edge{v}{u} \in E} \left(L_{\edge{v}{u}}^{x^*_{\edge{v}{u}}} \cdot L_{\edgeinfty{v}{u}}^{z^*_{\edge{v}{u}}}\right) \numberthis \label{eq:genp-ub},
\end{align*}
where $c$ is a small constant independent of $G$. 

Invoking Lemma~\ref{thm:acyclic-bound-lp}, and summing up $\B(\bard, G)$ over all possible degree configurations $\bd = (d_{\edge{v}{u}})_{d_{\edge{v}{u}} \le \min( L_{\edge{v}{u}}, L_{\edgeinfty{v}{u}}), \edge{v}{u} \in E}$, we get (for some small constant $c$ that we will pick later)
\begin{align*}
& \usum{\bard = (d_{\edge{v}{u}})_{d_{\edge{v}{u}} \le L_{\edgeinfty{v}{u}}, \edge{v}{u} \in E}} \B(\bard, G) \\
& \le \usum{\bard = (d_{\edge{v}{u}})_{d_{\edge{v}{u}} \le L_{\edgeinfty{v}{u}}, \edge{v}{u} \in E}} 2^{p |V|} \quad \uprod{\edge{v}{u} \in E}  d_{\edge{v}{u}}^{z^*_{\edge{v}{u}}} \cdot L_{\edged{v}{u}}^{x^*_{\edge{v}{u}}} \numberthis \label{eq:genp-holder-5} \\
& \le 2^{p |V|} \left( \usum{\bard = (d_{\edge{v}{u}})_{d_{\edge{v}{u}} \le L_{\edgeinfty{v}{u}}, \edge{v}{u} \in E}} \quad \uprod{\edge{v}{u} \in E}  d_{\edge{v}{u}}^{z^*_{\edge{v}{u}}}  \right) \\
& \quad \cdot \left(  \usum{\bard = (d_{\edge{v}{u}})_{d_{\edge{v}{u}} \le L_{\edge{v}{u}}, \edge{v}{u} \in E}} \quad \uprod{\edge{v}{u} \in E}  L_{\edged{v}{u}}^{x^*_{\edge{v}{u}}} \right) \\
& = 2^{p |V|} \left( \usum{\bard = (d_{\edge{v}{u}})_{d_{\edge{v}{u}} \le L_{\edgeinfty{v}{u}}, \edge{v}{u} \in E}} \quad \cdot \uprod{\edge{v}{u} \in E}  d_{\edge{v}{u}}^{z^*_{\edge{v}{u}}}  \right) \\
& \quad \cdot \left(  \usum{\bard = (d_{\edge{v}{u}})_{d_{\edge{v}{u}} \le L_{\edge{v}{u}}, \edge{v}{u} \in E}} \quad\cdot  \uprod{\edge{v}{u} \in E}  \left(L_{\edged{v}{u}}^{p}\right)^{\frac{x^*_{\edge{v}{u}}}{p}} \right) \\
& \le 2^{p |V|} \left( \usum{\bard = (d_{\edge{v}{u}})_{d_{\edge{v}{u}} \le L_{\edgeinfty{v}{u}}, \edge{v}{u} \in E}} \quad \uprod{\edge{v}{u} \in E}  d_{\edge{v}{u}}^{z^*_{\edge{v}{u}}}  \right) \\
& \quad \cdot \uprod{\edge{v}{u} \in E} \quad \left( \usum{\bard = (d_{\edge{v}{u}})_{d_{\edge{v}{u}} \le L_{\edge{v}{u}}, \edge{v}{u} \in E}} L_{\edged{v}{u}}^{p}  \right)^{\frac{x^*_{\edge{v}{u}}}{p}} \numberthis \label{eq:genp-holder-6} \\
& =  2^{p |V|} \left(\uprod{\edge{v}{u} \in E} \quad \usum{d_{\edge{v}{u}}: d_{\edge{v}{u}} \le L_{\edgeinfty{v}{u}}, \edge{v}{u} \in E} d_{\edge{v}{u}}^{z^*_{\edge{v}{u}}} \right) \cdot  \left( \uprod{\edge{v}{u} \in E}  L_{\edge{v}{u}}^{x^*_{\edge{v}{u}}} \right) \numberthis \label{eq:genp-holder-7} \\
&  \le 2^{p |V|} \cdot 3^{|E|} \uprod{\edge{v}{u} \in E} \frac{1}{z_{\edge{v}{u}}} \quad \uprod{\edge{v}{u} \in E} L_{\edgeinfty{v}{u}}^{z^*_{\edge{v}{u}}} \cdot L_{\edge{v}{u}}^{\frac{x^*_{\edge{v}{u}}}{p}} \numberthis \label{eq:genp-holder-8} \\
& \le 2^{p |V|} \cdot \inparen{(p |V|)^2}^{(p|V|)^2} \cdot  3^{|E|} \uprod{\edge{v}{u} \in E} L_{\edgeinfty{v}{u}}^{z^*_{\edge{v}{u}}} \cdot L_{\edge{v}{u}}^{x^*_{\edge{v}{u}}} \numberthis \label{eq:genp-holder-9}.
\end{align*}
Here, ~\eqref{eq:genp-holder-6} follows by a direct application of H\"{o}lder's inequality (assuming the following is true):
\begin{align*}
\usum{\edge{v}{u} \in E} \frac{x^*_{\edge{v}{u}}}{p} \ge 1. \numberthis \label{eq:genp-sum}
\end{align*}
We prove~\eqref{eq:genp-sum} here (assuming $s$ is the source vertex in $G$):
\begin{align*}
\usum{\edge{v}{u} \in E} \frac{x^*_{\edge{v}{u}}}{p} = \left(\usum{\edge{s}{w} \in E}\frac{x_{\edge{s}{w}}}{p} \right) + \left( \usum{\edge{v}{u} \in E \setminus \{\edge{s}{w} \in E\}} \frac{x^*_{\edge{v}{u}}}{p} \right) \ge 1,
\end{align*}
where the inequality follows from~\eqref{eq:gen-covering} for $s$ and the fact that $x_{\edge{v}{u}} \ge 0$ for every $\edge{v}{u} \in E$.x

Further,~\eqref{eq:genp-holder-7} follows by
\begin{align*}
\usum{1 \le d_{\edge{u}{w}} \le L_{{\edge{u}{w}}}} L_{{\edged{u}{w}}}^{p} = L_{\edge{u}{w}}^{p}
\end{align*}
and pushing the sum inside on $L_{\edged{v}{u}}$ for each $d_{\edge{v}{u}}$. We prove~\eqref{eq:genp-holder-8} as follows:
\begin{align*} 
\usum{d_{\edge{v}{u}} : d_{\edge{v}{u}} \le L_{\edgeinfty{v}{u}}, \edge{v}{u} \in E} \quad d_{\edge{v}{u}}^{z^*_{\edge{v}{u}}} & \le \frac{1}{z_{\edge{v}{u}}} L_{\edgeinfty{v}{u}}^{z^*_{\edge{v}{u}}}.
\end{align*}
Recall our assumption that the $d_{\edge{v}{u}}$ values are powers of two and $z_{\edge{v}{u}} > 0$ and as a result, we have (assuming $f = \log(L_{\edgeinfty{v}{u}})$)
\begin{align*}
1^{z^*_{\edge{v}{u}}} + 2^{z^*_{\edge{v}{u}}} + 2^{2 z^*_{\edge{v}{u}}} + \dots + 2^{z^*_{\edge{v}{u}} \cdot f} 
&  = \frac{2^{z^*_{\edge{v}{u}} (f + 1)} - 1}{2^{z^*_{\edge{v}{u}}} - 1} \\
& \le 2^{z^*_{\edge{v}{u}}} \frac{2^{z^*_{\edge{v}{u}} f}}{2^{z^*_{\edge{v}{u}}} - 1} \\
& = \frac{2^{z^*_{\edge{v}{u}} f}}{1 - 2^{-z^*_{\edge{v}{u}}}} \\
& \le 3 \cdot \frac{2^{z^*_{\edge{v}{u}} f}}{z^*_{\edge{v}{u}}} \\ 
& = \left( \frac{1}{z^*_{\edge{v}{u}}} \right) \cdot 3 \cdot L_{\edgeinfty{v}{u}}^{z^*_{\edge{v}{u}}}.
\end{align*}
In the above, the first equation follows by treating the left hand side expression as a geometric progression with first term $1$ and common difference $2^{z^*_{\edge{v}{u}}}$. The final inequality follows by noting that for small enough $z^*_{\edge{v}{u}} > 0$, $1 - 2^{-z^*_{\edge{v}{u}}}$ is at least $\frac{z^*_{\edge{v}{u}}}{c}$ for any $c\ge 3$.

Finally, we prove~\eqref{eq:genp-holder-9} by invoking a standard result in linear programming, which we state in the language of $\LP^{(+)}$ below. Recall that $\LP^{(+)}$ has rational coefficients for every constraint and each entry in the constraint matrix has only values in $[0, 1]$. Consider an optimal basic feasible solution $(\vx^*, \vz^*) = \left( x^*_{\edge{v}{u}}, z^*_{\edge{v}{u}} \right)$ to $\LP^{(+)}$ on $G$ such that for every $\edge{v}{u} \in E$ with $z^*_{\edge{v}{u}} > 0$ and $x^*_{\edge{v}{u}} > 0$, we have applying Cramer's rule, we get $z^*_{\edge{v}{u}} \ge \frac{1}{(p (2|E| + |V|))!} \ge \frac{1}{(p |V|^2)!}  \ge \frac{1}{\inparen{(p |V|)^2}^{(p |V|)^2}}$.\footnote{We would like to note here that these bounds hold for any basic feasible solution as well. For optimal basic feasible solutions, we can achieve a better bound than this one but we stick to this since it is sufficient for our arguments.}

In other words, this implies that the non-zero values of an optimal basic feasible solution of any linear program 
with rational coefficients are polynomially bounded in the size of its input. The above theorem immediately gives us~\eqref{eq:genp-holder-9} since
\begin{align*}
\uprod{\edge{v}{u} \in E} \frac{1}{z^*_{\edge{v}{u}}} & \le  \inparen{(p |V|)^2}^{(p |V|)^2},
\end{align*}
as required. Note that this combinatorial result implies the runtime of Algorithm~\ref{algo:gen-ub} as well.
This completes the proof.

\subsection{Proof of~\eqref{eq:gen-ub} for $p \le |V| - 1$} \label{sec:gen-deg-main}
Invoking Lemma~\ref{thm:acyclic-bound-lp} and summing up $\B(\bard, G)$ over all possible degree configurations$$\bd = (d_{\edge{v}{u}})_{d_{\edge{v}{u}} \le \min\inset{L_{\edge{v}{u}},L_{\edgeinfty{v}{u}}}, \edge{v}{u} \in E},$$ we get
\begin{align*}
\usum{\bard} \B(\bard, G) & \le \usum{\bard} 2^{p |V|} \uprod{\edge{v}{u} \in E}  \left( d_{\edge{v}{u}}^{z_{\edge{v}{u}}} \cdot L_{\edged{v}{u}}^{x_{\edge{v}{u}}} \right) \numberthis \label{eq:gen-holder-5} \\
& =  2^{p |V|} \usum{\bard} \uprod{\edge{v}{u} \in E} \left(  d_{\edge{v}{u}}^{z_{\edge{v}{u}}} \cdot \left( L_{\edged{v}{u}}^{p} \right)^{\frac{x_{\edge{v}{u}}}{p}} \right) \\
& \le 2^{p |V|} \uprod{\edge{v}{u} \in E} \left( \usum{\bard} d_{\edge{v}{u}} \right)^{z_{\edge{v}{u}}} \left( \usum{\bard} L_{\edged{v}{u}}^{p} \right)^{\frac{x_{\edge{v}{u}}}{p}} \numberthis \label{eq:gen-holder-6}  \\
& = 2^{p |V|} \uprod{\edge{v}{u} \in E} \left(  \left( \usum{\bard} d_{\edge{v}{u}}\right)^{z_{\edge{v}{u}}}  L_{\edge{v}{u}}^{x_{\edge{v}{u}}} \right)  \numberthis \label{eq:gen-holder-7} \\
& \le 2^{p |V|} \uprod{\edge{v}{u} \in E} \left(  \left(2 \cdot L_{\edgeinfty{v}{u}} \right)^{z_{\edge{v}{u}}} \cdot L_{\edge{v}{u}}^{x_{\edge{v}{u}}} \right) \numberthis \label{eq:gen-holder-8} \\
& \le 2^{(p + 1) |V|} \uprod{\edge{v}{v} \in E} \left( L_{\edgeinfty{v}{u}}^{z_{\edge{v}{u}}} \cdot L_{\edge{v}{u}}^{x_{\edge{v}{u}}} \right) \numberthis \label{eq:gen-holder-9}.
\end{align*}
Here,~\eqref{eq:gen-holder-6} follows by a direct application of H\"{o}lder's inequality assuming the following:
\footnote{The proof for the case when $p > |V| - 1$ actually diverges at this point and we push the sums differently instead of doing it through~\eqref{eq:gen-sum}.} 
\begin{align*}
\usum{\edge{v}{u} \in E} \left( z_{\edge{v}{u}} + \frac{x_{\edge{v}{u}}}{p} \right) \ge 1. \numberthis \label{eq:gen-sum}
\end{align*}
We prove~\eqref{eq:gen-sum} below:
\begin{align*}
\usum{\edge{v}{u} \in E} \left( z_{\edge{v}{u}} + \frac{x_{\edge{v}{u}}}{p} \right) & \ge \usum{\edge{v}{u}} \left( \frac{z_{\edge{v}{u}}}{p + 1} +  \frac{x_{\edge{v}{u}}}{p} \right) \\
& \ge \frac{|V|}{p + 1} \\
& \ge 1,
\end{align*}
where the first inequality follows from~\eqref{eq:gen-primal} and the fact that $p + 1 > 0$. The second inequality follows by summing up~\eqref{eq:gen-covering} for every $u \in V$ and the final inequality follows from our assumption that $p \le |V| - 1$. To complete the proof, we prove~\eqref{eq:gen-holder-7} and~\eqref{eq:gen-holder-8}, where the former follows from 
\begin{align*}
\usum{1 \le d_{\edge{u}{w}} \le L_{{\edge{u}{w}}}} L_{{\edged{u}{w}}}^{p} = L_{\edge{u}{w}}^{p}
\end{align*}
and the latter follows by definition of $d_{\edge{v}{u}}$ values being powers of two. This completes the proof. 

To complete the proof of Theorem~\ref{thm:gen-main}, we need to prove~\eqref{eq:gen-lb} as well, which we do using the dual of $\LP^{(+)}$ in Appendix~\ref{sec:gen-lb}.
\subsection{Proof of~\eqref{eq:gen-lb},~\eqref{eq:lp-lb} and~\eqref{eq:dhigh-lb}} \label{sec:gen-lb}
In this section, we prove a lower bound for $|\TJ|$ by constructing an instance $I = \{R_{\edge{v}{u}}: ||R_{\edge{v}{u}}||_{p} \le L_{\edge{v}{u}}, ||R_{\edge{v}{u}}||_{\infty} \le L_{\edgeinfty{v}{u}}, \edge{v}{u} \in E \}$ using the dual of $\LP^{(+)}$. We would like to note here that our lower bound holds for any $G$, any $p \in [1, \infty)$ and also for the case when there are no $\ell_{\infty}$ constraints.

We start by stating the dual of $\LP^{(+)}$.
\begin{align*}
& \max \quad \usum{u \in V} y_{u}  \\
& \frac{y_v}{p} + y_u \le \log(L_{\edge{v}{u}}) \quad \forall \edge{v}{u} \in E \numberthis \label{eq:gen-dual-ub} \\
& y_{u} \le \log(L_{\edgeinfty{v}{u}}) \quad \forall u \in V \numberthis \label{eq:gen-dual-ub-2} \\
& y_u \ge 0 \quad \forall u \in V \numberthis \label{eq:gen-dual}.
\end{align*}
We will now construct a join instance $\{R_{\edge{v}{u}}: \edge{v}{u} \in E\}$ based on an optimal dual solution $\vy^* = (y^*_u)_{u \in V}$ such that
\begin{align*}
||R_{\edge{v}{u}}||_{p} \le L_{\edge{v}{u}},  ||R_{\edge{v}{u}}||_{\infty} \le L_{\edgeinfty{v}{u}}, \quad \forall \edge{v}{u} \in E
\end{align*}
and 
\begin{align*}
|\TJ| \ge \frac{1}{2^{|V|}} \uprod{\edge{v}{u} \in E} \left( L_{\edge{v}{u}}^{x^*_{\edge{v}{u}}} L_{\edgeinfty{v}{u}}^{z^*_{\edge{v}{u}}} \right),
\end{align*}
where recall that $(\vx^*, \vz^*) = (x^*_{\edge{v}{u}}, z^*_{\edge{v}{u}})_{\edge{v}{u} \in E}$ denotes an optimal solution to $\LP^{(+)}$.

Given $\vy^*$, we define $\Dom(u) = \left[ \floor{2^{y^*_u}} \right]$ for every $u \in V$ and $R_{\edge{v}{u}} = \Dom(v) \times \Dom(u)$ for every $\edge{v}{u} \in E$. For each $\edge{v}{u} \in E$, we have
\begin{align*}
||R_{\edge{v}{u}}||_p & = \sqrt[p]{ \usum{\val_{v} \in \Dom(v)}| \Dom(u)|^p } \\
& = \sqrt[p]{ |\Dom(v)| \cdot |\Dom(u)|^p} \\
& \le \sqrt[p]{ 2^{y^*_v} \cdot 2^{p \cdot y^*_u}} \\
& = 2^{\frac{y^*_v}{p}} 2^{y^*_u} \\
& \le 2^{\log(L_{\edge{v}{u}})} \quad \text{ follows from~\eqref{eq:gen-dual-ub}}  \\
& = L_{\edge{v}{u}}
\end{align*}
and
\begin{align*}
||R_{\edge{v}{u}}||_{\infty} & = \umax{\val_{v} \in \Dom(v)} |\Dom(u)| \\
& = \floor{2^{y^*_u}} \\
&  \le L_{\edgeinfty{v}{u}} \quad \text{ follows from~\eqref{eq:gen-dual-ub-2}}. 
\end{align*}
Based on this instance, we obtain a lower bound of
\begin{align*}
|\TJ| & \ge \uprod{u \in V} |\Dom(u)| \\
& = \uprod{u \in V}  \floor{2^{y^*_u}} \\
& \ge \frac{1}{2^{|V|}} \uprod{u \in V} 2^{y^*_u} \\
& = \frac{1}{2^{|V|}} \uprod{\edge{v}{u} \in E} \left( L_{\edge{v}{u}}^{x^*_{\edge{v}{u}}} \cdot L_{\edgeinfty{v}{u}}^{z^*_{\edge{v}{u}}} \right),
\end{align*}
where the first inequality follows by our definition of $|\Dom(u)|$ and the final equality follows from strong duality. In particular, at optimality, $\LP^{(+)}$ and its dual have the same objective value and since $(\vx^*, \vz^*)$ and $\vy^*$ are optimal solutions to $\LP^{(+)}$ and its dual respectively, this completes the proof.

A similar construction as above holds for any orientation of $G$ even when only $\ell_{p}$-norm constraints are given (and no $\ell_{\infty}$ constraints are given), leading to the following corollary.
\begin{corollary} \label{cor:gen-lb}
For any $G$, any $p \in [1, \infty)$ and an optimal solution $\vx^* = (x^*_{\edge{v}{u}})_{\edge{v}{u} \in E}$ to $\LP^{(+)}$ on $G$, there exists an instance $I = \{R_{\edge{v}{u}}: ||R_{\edge{v}{u}}||_{p} \le L_{\edge{v}{u}}, \edge{v}{u} \in E\}$ such that
\begin{align*}
|\TJ| & \ge \frac{1}{2^{|V|}} \cdot \uprod{\edge{v}{u} \in E} \left( L_{\edge{v}{u}}^{x^*_{\edge{v}{u}}}  \right).
\end{align*}
\end{corollary}

%% file: appendix_sec4_2.tex
\section{Missing Details in Section~\ref{sec:lp-main}} \label{sec:lp}
In this section, our goal is to prove~\eqref{eq:lp-ub}. As discussed in Section~\ref{sec:lp-main}, we invoke Theorem~\ref{thm:opt-decomp} on $G$ and consider the $G_i$s one-by-one -- if $G_i$ is a DAG, we can invoke Theorem~\ref{thm:gen-main} without any $\ell_{\infty}$ constraints to prove~\eqref{eq:lp-ub}. We fall back to the case when $G_i$ contains at least one cycle $C$ and that will be the main focus of this section. Since $|E| = |V|$, there are $|C|$ edge disjoint (directed) trees each one rooted at a unique $A_i \in C$. Note that this is a standard graph theory result and we prove it in Appendix~\ref{sec:cycle-trees}. 

Consider the topological ordering for each of the $|C|$ edge (directed) disjoint trees denoted by $\cT_{A_i}$ for every $A_i \in C$.  Let $\cT_{A_i^{\inc}}$ and $\cT_{A_i^{\out}}$ denote an ordering {\em before} and {\em after} $A_i$ in $\cT_{A_i}$. The following result is true.
\begin{corollary} \label{cor:cycle-trees}
\begin{align*}
\mathcal{T} = \left( \left(\cT_{A_i^{\inc}} \right)_{A_i \in C},  C, \left(\cT_{A_i^{\out}} \right)_{A_i \in C}  \right) \numberthis \label{eq:top-cycle}
\end{align*}
is a valid topological ordering for $G_i'$, which is essentially $G_i$ and we treat the cycle $C$ as a single vertex.
\end{corollary}
For simplicity of notation, we will assume $G = G_i$ for the rest of this argument. Since the $G_i$s are a disjoint union by Theorem~\ref{thm:opt-decomp}, we can apply the same argument on each $G_i$. We define some notation. Let $\cA_i^{\inc}, \cA_i^{\out}$ denote the set of vertices with incoming/outgoing edges to $A_i$ for every $A_i \in C$ and $L_{{\edged{v}{u}}}$ denotes the $\ell_{p}$-norm constraint corresponding to $d_{\edge{v}{u}}$. We follow the same proof structure as in Section~\ref{sec:overview} and in particular, Section~\ref{sec:g-acyclic}. 

We first state the corresponding version of Theorem~\ref{claim:gen-acyclic}.
\begin{cor} \label{claim:cyclic}
For any $G$ with $|E| = |V|$ and every $\bd = (d_{\edge{v}{u}})_{d_{\edge{v}{u}} \le L_{\edge{v}{u}}, \edge{v}{u} \in E}$, we have
\begin{align*}
|\TJ(\bd)| & \le \B'(\bard, G),
\end{align*}
where	
\begin{align*}
\B'(\bard, G) = \uprod{u \in \cT_{A_i^{\inc}}: A_i \in C} \cD_{u}\left(\bard \right) \cdot \underset{A_i \in C}{\min} \left \{ \cD'_{i} \left(\bd \right) \cdot \prod_{A_j \in C, j \neq i}  \cD_{j} \left(\bard \right) \right \} \cdot \uprod{u \in \cT_{A_i^{\out}}, A_i \in C} \cD_{u} \left(\bard \right)  \numberthis \label{eq:cyc-bdg-ub}
\end{align*}
with
\begin{align*}
\cD'_{i} \left(\bd \right) & =  \min \left \{ \left \{ d_{\edge{A_i^{\inc}}{A_i}} \right \}_{A_i^{\inc} \in \cA_i^{\inc}}, \left \{ \frac{L_{\edged{A_i}{A_i^{\out}}}^{p}}{d_{\edge{A_i}{A_i^{\out}}}^{p}} \right\}_{A_i^{\out} \in \cA_i^{\out}}, \frac{L_{\edged{A_i}{A_{i + 1}}}^{p}}{d_{\edge{A_i}{A_{i + 1}}}^{p}}  \right \} \quad A_i \in C \numberthis \label{eq:cyc-cd-prime} \\
& \cD_{i}\left(\bd \right) = \min \left \{ \cD'_{i}(\bard) , d_{\edge{A_{i - 1}}{A_i}}  \right \} \quad A_i \in C \numberthis \label{eq:cyc-cd} \\
& \cD_{u}\left(\bd \right) = \min \left \{  \{ d_{\edge{v}{u}} \}_{\edge{v}{u} \in E},  \left \{ \frac{L_{\edged{u}{w}}^{p}}{d_{\edge{u}{w}}^p}  \right \}_{\edge{u}{w} \in E}   \right \} \quad u \in V \setminus V(C). \numberthis \label{eq:cyc-cd-all}
\end{align*}
Further, $|\TJ(\bd)|$ can be computed in time $O(\B'(\bard, G))$.
\end{cor}
We would like to mention here that $\B'(\bard, G)$ in the claim above is different from the upper bound $\B(\bard, G)$ for the acyclic case. To prove Corollary~\eqref{claim:cyclic}, we consider all possible acyclic subgraphs (each one obtained by deleting a back edge from the cycle) and invoke Theorem~\ref{claim:gen-acyclic} on each of these acyclic subgraphs. The upper bound follows as a result. From the algorithmic aspect, this translates to running Algorithm~\ref{algo:acyclic-ub} over all these acyclic subgraphs and we can in fact, with the knowledge of $p$ and the corresponding $\ell_{p}$-norm bound compute the ordering achieving the minimum as well.

Next, we restate our primal $\LP^{(+)}$ for this scenario.
 \begin{align*}
 & \min \quad \usum{\edge{v}{u} \in E} x_{\edge{v}{u}} \log(L_{\edge{v}{u}}) \tag{$\LP^{(+)}$}\\
 & \usum{\edge{v}{u} \in E} x_{\edge{v}{u}} + \usum{\edge{u}{w} \in E} \frac{x_{\edge{u}{w}}}{p} \ge 1 \quad \forall u \in V \numberthis \label{eq:lp-covering} \\
 & x_{\edge{v}{u}} \ge 0 \quad \forall \edge{v}{u} \in E \numberthis \label{eq:lp-primal}.
 \end{align*}
 We are now ready to state the corresponding version of Lemma~\ref{thm:acyclic-bound-lp}.
\begin{lemm} \label{thm:cyclic-bound-lp}
For any $G$ satisfying $|E| = |V|$ with girth at least $p + 1$, any feasible solution $\vx = (x_{\edge{v}{u}})_{\edge{v}{u} \in E}$ to $\LP^{(+)}$ on $G$ and degree configuration $\bd$, we have 
\begin{align*}
\B'(\bard, G)  \le \uprod{\edge{v}{u} \in E} L_{\edged{v}{u}}^{x_{\edge{v}{u}}} \numberthis \label{eq:imp-eq-in-app},
\end{align*}
where $\B'(\bard, G)$ is defined in~\eqref{eq:cyc-bdg-ub}.
\end{lemm} 
We note here that our assumption of $G$ having girth at least $p+1$ in Theorem~\ref{thm:lp-main} stems from this result.
For computing $|\TJ(\bd)|$, we first identify the $u \in V(C)$ that achieves the $\min$ in $\B'(\bard, C)$. Since $D'_{u}$ does not contain the incoming degree constraint $d_{\edge{v}{u}}$ (where $\edge{v}{u} \in E(C)$), we can run Algorithm~\ref{algo:gen-ub} on a topological ordering $(u, \dots, v)$ of $C$ and it would run in time $O(\B'(\bard, G))$.

Finally, we can use these two results to prove~\eqref{eq:lp-ub} assuming Lemma~\ref{thm:cyclic-bound-lp} is true.
\begin{proof} [Proof of~\eqref{eq:lp-ub}]
Invoking Lemma~\ref{thm:cyclic-bound-lp} and summing $\B(\bd, G)$ over all possible degree configurations and using the upper bound above, we have
\begin{align*}
	& \usum{\bd = \left(d_{\edge{v}{u}} \right)_{d_{\edge{v}{u}} \le L_{\edge{v}{u}}, \edge{v}{u} \in E}} \B'(\bd, G) \\
	& \le \usum{\bd = \left(d_{\edge{v}{u}} \right)_{d_{\edge{v}{u}} \le L_{\edge{v}{u}}, \edge{v}{u} \in E}} \quad \uprod{\edge{v}{u} \in E} \left(L_{\edged{v}{u}}^{p}\right)^{\frac{x_{\edge{v}{u}}}{p}} \\
	& \le \uprod{\edge{v}{u} \in E} \left( \usum{\bd = \left(d_{\edge{v}{u}} \right)_{d_{\edge{v}{u}} \le L_{\edge{v}{u}}, \edge{v}{u} \in E}} L_{\edged{v}{u}}^{p} \right)^{\frac{x_{\edge{v}{u}}}{p}} \\
	& =  \uprod{\edge{v}{u} \in E} \left( L_{\edge{v}{u}}^{p} \right)^{\frac{x_{\edge{v}{u}}}{p}} = \uprod{\edge{v}{u} \in E} L_{\edge{v}{u}}^{x_{\edge{v}{u}}}. 
	\end{align*}
	Here, the second inequality follows by applying H\"{o}lder's inequality assuming $\usum{\edge{v}{u} \in E} \frac{x_{\edge{v}{u}}}{p} \ge 1$ (which we argue below). The second equation follows from the fact that $$\usum{d_{\edge{v}{u}} \le L_{\edge{v}{u}}} L_{\edged{v}{u}}^p = L_{\edge{v}{u}}^p$$
	(by definition of $L_{\edged{v}{u}}$).
	
	We argue $\usum{\edge{v}{u} \in E} \frac{x_{\edge{v}{u}}}{p} \ge 1$ as follows:
	\begin{align*}
	& \usum{\edge{v}{u} \in E} \frac{x_{\edge{v}{u}}}{p} \\
	& = \frac{1}{p + 1} \usum{u \in V} \left(  \usum{\edge{v}{u} \in E} x_{\edge{v}{u}} + \usum{\edge{u}{w} \in E} \frac{x_{\edge{u}{w}}}{p} \right) \\
	& \ge \frac{|V|}{p + 1} \ge 1.
	\end{align*}
	Here, the first inequality follows from~\eqref{eq:lp-covering} for every $u \in V$ and the second inequality follows from $p \le |V| - 1$. Note that from the algorithm point of view (li.e., Algorithm~\ref{algo:generic-ub}), to compute $J$, Algorithm~\ref{algo:acyclic-ub} needs to go over all possible acyclic subgraphs in the cycle in $G$. The remaining runtime arguments follow through with additional linear factors in $|V|$ and $|E|$.	

	This completes the proof.
\end{proof}

Our remaining work in this section is to prove Lemma~\ref{thm:cyclic-bound-lp}, which we do in two broad steps. We define 
some notation:
\begin{align*}
E_{\inc} = \left \{ \edge{A_i^{\inc}}{A_i}: A_i^{\inc} \in \cA_i^{\inc}, A_i \in C \right \} \\  
E_{\out} = \left \{ \edge{A_i}{A_i^{\out}}: A_i^{\out} \in \cA_i^{\out}, A_i \in C \right \} \\ 
E_{R} = E \setminus \left( E(C) \cup E_{\inc} \cup E_{\out} \right).
\end{align*}
Through the rest of this section, when we use $(i + 1)$ and $(i - 1)$, we mean $(i + 1) \bmod k$ and $(i - 1) \bmod k$, where $k$ is the length of the cycle in $G$.
\subsection{Proof of Lemma~\ref{thm:cyclic-bound-lp}} \label{sec:cycle-w-trees}
We start by modeling the computation of the logarithmic version of $\B'(\bard, G)$ as a linear program relaxation (which we call $\LP^{(*)}$) and defined below:
\begin{align*}
\max \quad \left(\usum{u \in V \setminus V(C)} y_{u}  \right) + z_{C}  \tag{$\LP^{(*)}$} \\ 
\text{ s.t. } z_{C} \le z_{\edge{A_i}{A_{i + 1}}} + \sum_{A_j \in C \setminus A_i} y_{A_j}'   \quad \forall A_i \in C \numberthis \label{eq:ct27} \\
\frac{z_{\edge{A_i}{A_{i + 1}}}}{p} + y_{A_{i + 1}}' \le \log(L_{{\edged{A_i}{A_{i + 1}}}})  \quad \forall A_i \in C \numberthis \label{eq:ct21}\\
\frac{y_{A_i^{\inc}}}{p} + z_{\edge{A_i}{A_{i + 1}}}  \le \log(L_{{\edged{A_i^{\inc}}{A_i}}})  \quad \forall \edge{A_i^{\inc}}{A_i} \in E_{\inc} \numberthis \label{eq:ct23} \\
\frac{z_{\edge{A_i}{A_{i + 1}}}}{p} + y_{A_i^{\out}} \le \log(L_{{\edged{A_i}{A_i^{\out}}}}) \quad \forall \edge{A_i}{A_i^{\out}} \in E_{\out} \numberthis \label{eq:ct22} \\ 
\frac{y_v}{p} + y_{u} \le \log(L_{{\edged{v}{u}}}) \quad \forall \edge{v}{u} \in E_R \numberthis \label{eq:extra21} \\
y_{A_i}' \le z_{\edge{A_i}{A_{i + 1}}} \quad \forall A_i \in C \numberthis \label{eq:ct24} \\
z_{\edge{A_i}{A_{i + 1}}} \ge 0 \quad \forall A_i \in C \numberthis \label{eq:ct25} \\
z_{C}, y_{A_i}' \ge 0 \quad \forall A_i \in C \numberthis \label{eq:ct26} \\
y_{u} \ge 0 \quad \forall u \in V \setminus V(C). \numberthis \label{eq:extra22}
\end{align*}
We prove the following two results, which when combined together prove Lemma~\ref{thm:cyclic-bound-lp}.
\begin{lemm} \label{claim:gen-lp}
	\begin{align*}
	\B'(\bard, G) & \le 2^{\LP^{(*)}}.
	\end{align*}
\end{lemm}
\begin{lemm} \label{lem:gen-2}
	\begin{align*}
	2^{\LP^{(*)}} & \le \uprod{\edge{v}{u} \in E} L_{\edged{v}{u}}^{x_{\edge{v}{u}}},
	\end{align*}
	where $\vx = (x_{\edge{v}{u}})_{\edge{v}{u} \in E}$ is a feasible solution to $\LP^{(+)}$ on $G$ with values $(L_{\edged{v}{u}})_{\edge{v}{u} \in E}$.
\end{lemm}
The proof of Lemma~\ref{claim:gen-lp} follows from standard linear programming techniques to convert $\min$ bounds into linear programming relaxations and as a result, we defer it to Appendix~\ref{sec:gen-lp}. We prove Lemma~\ref{lem:gen-2} here using the following claim (the proof is Section~\ref{sec:opt-lp-gen}).
\begin{claim} \label{claim:opt-lp-gen}
	There exists an optimal solution to $\LP^{(*)}$ such that
	\begin{align*} 
	\underset{A_i \in C}{\min} \left(z_{\edge{A_i}{A_{i + 1}}} - y'_{A_i}\right) = 0 \numberthis \label{eq:zero-opt-gen}.
	\end{align*}
\end{claim}
\subsubsection{Proof of Lemma~\ref{lem:gen-2}} \label{sec:gen-2-app}
Our proof will use the dual of $\LP^{(+)}$ (which we call $\LP^{(**)}$) and weak duality. We start by restating the dual here.
\begin{align*}
& \max \quad \usum{u \in V} y_u  \\
& \frac{y_v}{p} + y_u \le \log(L_{\edge{v}{u}} \quad \forall \edge{u}{v} \in E \numberthis \label{eq:dual-restate-g1}\\
& y_u \ge 0 \quad \forall u \in V \numberthis \label{eq:dual-restate-g3}
\end{align*}
We will now restate the dual in a language closer to the notation in this section.
\begin{align*}
& \max \quad \left ( \usum{u \in \cT_{A_i^{\inc}}: A_i \in C} y_{u} \right) + \left(\usum{u \in C} y_{u} \right) + \left( \usum{u \in \cT_{A_i^{\out}}: A_i \in C} y_{u} \right)  \tag{$\LP^{(**)}$}\\
& \frac{y_{A_i}}{p} + y_{A_{i + 1}} \le \log(L_{{\edged{A_i}{A_{i + 1}}}})  \quad \forall A_i \in C \numberthis \label{eq:dual-g1}\\
& \frac{y_{A_i^{\inc}}}{p} + y_{A_i}  \le \log(L_{{\edged{A_i^{\inc}}{A_i}}})  \quad \forall \edge{A_i^{\inc}}{A_i} \in E_{\inc} \numberthis \label{eq:dual-g2} \\
& \frac{y_{A_i}}{p} + y_{A_i^{\out}} \le \log(L_{{\edged{A_i}{A_i^{\out}}}}) \quad \forall \edge{A_i}{A_i^{\out}} \in E_{\out} \numberthis \label{eq:dual-g3} \\ 
& \frac{y_v}{p} + y_u \le \log(L_{\edged{v}{u}}) \quad \forall \edge{v}{u} \in E_{R} \numberthis \label{eq:dual-g4} \\
& y_u \ge 0 \quad \forall u \in V \numberthis \label{eq:dual-g5}.
\end{align*}
Here, we decompose the objective value based on Corollary~\ref{cor:cycle-trees} and we decompose the constraints based on definitions of $E_{\inc}, E_{\out}$ and $E_{R}$ respectively.

\begin{proof}[Proof of Lemma~\ref{lem:gen-2}]
	We first argue that there exists an optimal solution to $\LP^{(*)}$ on $G$ that can be converted to a feasible solution to $\LP^{(**)}$ on $G$. Since the objective values of both $\LP^{(*)}$ and $\LP^{(**)}$ are maximizing, this completes the proof.
	
	Consider a feasible solution to $\LP^{(*)}$:
	\begin{align*} 
	\vy = \left( \left(y_{u}\right)_{u \in \cT_{A_i^{\inc}, A_i}, A_i \in C},  \left( z_{\edge{A_i}{A_{i + 1}}}  \right)_{A_i \in C}, \left( y'_{A_i}  \right)_{A_i \in C},  z_{C} , \left(y_{u}\right)_{u \in \cT_{A_i, A_i^{\out}}, A_i \in C}\right)
	\end{align*}
	We construct the following solution to $\LP^{(**)}$, where we have (with a slight abuse of notation)
	\begin{align*}
	\vy' = \left( \left(y_{u} = y_{u}\right)_{u \in \cT_{A_i^{\inc}, A_i}, A_i \in C},  \left( y_{A_i} = y'_{A_i}  \right)_{A_i \in C},\left(y_{u} = y_{u}\right)_{u \in \cT_{A_i, A_i^{\out}}, A_i \in C}\right) 
	\end{align*}
	With this assignment, note that the only difference between the linear programs $\LP^{(*)}$ and $\LP^{(**)}$ is that the set of variables $(z_{\edge{A_i}{A_{i + 1}}})_{A_i \in C}$ and $z_{C}$ in $\LP^{(*)}$ are not present in $\LP^{(**)}$. 
	
	Next, we argue that the solution $\vy'$ we constructed is feasible for $\LP^{(**)}$ (i.e., satisfies constraints~\eqref{eq:dual-g1}--\eqref{eq:dual-g5}). Recall from definition of $\LP^{(**)}$, we have 
	\begin{align*}
	E_{A_i^{\inc}} = \left \{ \edge{A_i^{\inc}}{A_i}: A_i^{\inc} \in \cA_i^{\inc}, A_i \in C \right \} \\ 
	E_{A_i^{\out}} = \left \{ \edge{A_i}{A_i^{\out}}: A_i^{\out} \in \cA_i^{\out}, A_i \in C \right \} \\
	E_{R} = E \setminus \left( E(C) \cup E_{A_i^{\inc}} \cup E_{A_i^{\out}} \right).
	\end{align*} 
	We start with~\eqref{eq:dual-g1}.
	\begin{align*}
	\frac{y_{A_i}}{p} + y_{A_{i + 1}} \le \log(L_{{\edged{A_i}{A_{i + 1}}}})  \quad \forall A_i \in C \numberthis \label{eq:y-z-g}.
	\end{align*}
	In particular, we have (by our construction) for every $A_i \in C$:
	\begin{align*}
	\frac{y_{A_i}}{p} + y_{A_{i + 1}} & = \frac{y_{A_i}'}{p} + y_{A_{i + 1}}'  \\
	& \le \frac{z_{\edge{A_i}{A_{i + 1}}}}{p} + y_{A_{i + 1}}' \\
	& \le  \log(L_{{\edged{A_i}{A_{i + 1}}}}),
	\end{align*}
	where the equality follows by definition, first inequality follows from~\eqref{eq:ct24} and the final inequality follows from~\eqref{eq:ct21}. We can make similar arguments for~\eqref{eq:dual-g2}. For edges in $E_{A_i^{\inc}}$ and $E_{A_i^{\out}}$, we can do a similar argument  as above using~\eqref{eq:ct22},~\eqref{eq:ct23} and~\eqref{eq:ct24}. Finally, for edges in $E_R$, we have (from~\eqref{eq:extra21}):
	\begin{align*}
	\frac{y_v}{p} + y_u\le \log(L_{{\edged{v}{u}}}).
	\end{align*} 
	The remaining constraints in $\LP^{(**)}$ are satisfied as they directly follow from~\eqref{eq:ct26} and~\eqref{eq:extra22} (of $\LP^{(*)}$). We restate them here for the sake of completeness. 
	\begin{align*}
	y_{A_i} = y_{A}' \ge 0 \quad \forall A_i \in C \\
	\left(y_{u} = y_{u}\right)_{u \in \cT_{A_i^{\inc}, A_i}, A_i \in C} \\
	\left(y_{u} = y_{u}\right)_{u \in \cT_{A_i, A_i^{\out}}, A_i \in C}.
	\end{align*}
	
	To complete the proof, we argue that the optimal objective value of $\LP^{(*)}$ is equal to the 
	objective value of $\vy'$. We first claim that the following is true at optimally for $\LP^{(*)}$:
	\begin{align*}
	z_{C} = \underset{A_i \in C}{\min} \left(z_{\edge{A_i}{A_{i + 1}}} + \sum_{A_j \in C \setminus A_i} y_{A_j}' \right) \numberthis \label{eq:zc-opt}.
	\end{align*}
	Assuming this is true, we can take an optimal solution $(\vy_u, z_c)$ to $\LP^{(*)}$ that satisfies Claim ~\ref{claim:opt-lp-gen} and rewrite the optimal objective value of $\LP^{(*)}$ as follows:
	\begin{align*}
	& \left(\usum{u \in \cT_{A_i^{\inc}, A_i}, A_i \in C} y_{u} \right) + z_{C} + \left(\usum{u \in \cT_{A_i, A_i^{\out}}, A_i \in C} y_{u}\right) \\
	& = \left(\usum{u \in \cT_{A_i^{\inc}, A_i}, A_i \in C} y_{u} \right)  + \umin{A_i \in C} \left(z_{\edge{A_i}{A_{i + 1}}} + \sum_{A_j \in C \setminus A_i} y_{A_j}' \right) +  \left(\usum{u \in \cT_{A_i, A_i^{\out}}, A_i \in C} y_{u}\right)  \\
	& = \left(\usum{u \in \cT_{A_i^{\inc}, A_i}, A_i \in C} y_{u} \right) + \left( \underset{A_i \in C}{\min} \left(z_{\edge{A_i}{A_{i + 1}}} - y_{A_i}' \right) + \sum_{A_i \in C} y_{A_i}' \right) +  \left(\usum{u \in \cT_{A_i, A_i^{\out}}, A_i \in C} y_{u}\right)  \\
	& =  \left(\usum{u \in \cT_{A_i^{\inc}, A_i}, A_i \in C} y_{u} \right) +  \left( \sum_{A_i \in C} y_{A_i}'  \right) + \left(\usum{u \in \cT_{A_i, A_i^{\out}}, A_i \in C} y_{u}\right),
	\end{align*}
	where the final equation follows from Claim~\ref{claim:opt-lp-gen}. Note that this shows the optimal objective value of $\LP^{(*)}$ is equal to the objective value of $\vy'$. To complete this argument, we prove~\eqref{eq:zc-opt}. Recall from~\eqref{eq:ct27} that $z_{C}$ is upper bounded by $z_{\edge{A_i}{A_{i + 1}}} + \sum_{A_j \in C \setminus A_i} y_{A_j}'$ for every $A_i \in C$. In particular, this implies
	\begin{align*}
	z_{C} \le \underset{A_i \in C}{\min} \left(z_{\edge{A_i}{A_{i + 1}}} + \sum_{A_j \in C \setminus A_i} y_{A_j}' \right).
	\end{align*}
	Recall that the objective value of $\LP^{(*)}$ is maximizing and as a result, if $\underset{A_i \in C}{\min} \left(z_{\edge{A_i}{A_{i + 1}}} + \sum_{A_j \in C \setminus A_i} y_{A_j}' \right)$ is greater than $z_{C}$, we can always increase $z_{C}$ to make it equal to $\underset{A_i \in C}{\min} \left(z_{\edge{A_i}{A_{i + 1}}} + \sum_{A_j \in C \setminus A_i} y_{A_j}' \right)$. Note that this can only increase the objective value of $\LP^{(*)}$ resulting solution is still feasible since $z_{C}$ is not involved in any other constraint except $z_{C} \ge 0$.

	This shows $\LP^{(*)} \le \LP^{(**)}$ and using strong duality, we have that the optimal objective values of both $\LP^{(**)}$ and $\LP^{(+)}$ are equal. This completes the proof, showing that
	\begin{align*}
	2^{\LP^{(*)}} & \le 2^{\LP^{(**)}} \\
	& = 2^{\LP^{(+)}} \\
	& \le \uprod{\edge{v}{u} \in E} L_{\edged{v}{u}}^{x_{\edge{v}{u}}},
	\end{align*}
	as required.
\end{proof}

\subsubsection{Proof of Claim~\ref{claim:opt-lp-gen}} \label{sec:opt-lp-gen}
At a high level, we show how to convert any optimal solution to $\LP^{(*)}$ to the form stated in this claim. We start by assuming that there exists an optimal solution to $\LP^{(*)}$
\begin{align*}
\vy^* = \left( \left(y_{u}\right)_{u \in \cT_{A_i^{\inc}, A_i}, A_i \in C},  \left( z_{\edge{A_i}{A_{i + 1}}}  \right)_{A_i \in C}, \left( y'_{A_i}  \right)_{A_i \in C},  z_{C} , \left(y_{u}\right)_{u \in \cT_{A_i, A_i^{\out}}, A_i \in C}\right)
\end{align*}
such that
\begin{align*}
\underset{A_i \in C}{\min} \left(z_{\edge{A_i}{A_{i + 1}}} - y_{A_i}' \right) = \epsilon > 0 \numberthis \label{eq:opt-eps-gen}.
\end{align*}
Then, we construct a related solution 
\begin{align*}
\mathbf{z}^{*} = \left( \left(y_{u}\right)_{u \in \cT_{A_i^{\inc}, A_i}, A_i \in C},  \left( z^*_{\edge{A_i}{A_{i + 1}}}  \right)_{A_i \in C}, \left( y^*_{A_i}  \right)_{A_i \in C},  z_{C} , \left(y_{u}\right)_{u \in \cT_{A_i, A_i^{\out}}, A_i \in C}\right)
\end{align*}
to $\LP^{(*)}$ using $\mathbf{y}^{*}$, where we set (assuming $k = |C|$)
\begin{align*}
z^*_{\edge{A_i}{A_{i + 1}}} = z_{\edge{A_i}{A_{i + 1}}} - \frac{p \cdot \epsilon}{k}, \quad y^*_{A_i} = y_{A_i}' + \frac{\epsilon}{k} \quad \forall A_i \in C \numberthis \label{eq:zy-opt-gen}.
\end{align*}

Note that the remaining values $\left(\left(y_{u}\right)_{u \in \cT_{A_i^{\inc}, A_i}, A_i \in C}, \left(y_{u}\right)_{u \in \cT_{A_i, A_i^{\out}}, A_i \in C}\right)$ remain the same. Note if we argue $\mathbf{z}^{*}$ is an optimal solution to $\LP^{(*)}$, then that completes the proof of Claim~\ref{claim:opt-lp-gen}. We argue this in two steps.
\begin{lemm} \label{lemma:z-opt-gen}
	The objective value of \space $\mathbf{z}^{*}$ is at least the optimal objective value of $\LP^{(*)}$.
\end{lemm}
\begin{lemm} \label{lemma:z-feasible-gen}
	$\mathbf{z}^{*}$ is a feasible solution to $\LP^{(*)}$.
\end{lemm}
The latter proof follows from standard techniques and we defer it to Appendix~\ref{sec:z-feasible-gen}. We prove Lemma~\ref{lemma:z-opt-gen} here.
\begin{proof}[Proof of Lemma~\ref{lemma:z-opt-gen}]
	Recall that $\mathbf{y}^*$ is an optimal solution to $\LP^{(*)}$ with objective value 
	\begin{align*}	\left(\usum{u \in \cT_{A_i^{\inc}, A_i}, A_i \in C} y_{u} \right) + z_{C} + \left(\usum{u \in \cT_{A_i, A_i^{\out}}, A_i \in C} y_{u}\right).
	\end{align*}
	and $\mathbf{z}^*$ has objective value
	\begin{align*}
	\left(\usum{u \in \cT_{A_i^{\inc}, A_i}, A_i \in C} y_{u} \right) + \underset{A_i \in C}{\min} \left(z^*_{\edge{A_i}{A_{i + 1}}} - y^*_{A_i} \right) + \left(\sum_{A_i \in C} y^*_{A_i} \right) +  \left(\usum{u \in \cT_{A_i, A_i^{\out}}, A_i \in C} y_{u}\right).
	\end{align*}
	Our goal here is to prove that the objective value of $\vz^*$ is at least the objective value of $\vy^*$. More formally, we prove
	\begin{align*}
	& \left(\usum{u \in \cT_{A_i^{\inc}, A_i}, A_i \in C} y_{u} \right) + \underset{A_i \in C}{\min} \left(z^*_{\edge{A_i}{A_{i + 1}}} - y^*_{A_i} \right) + \left(\sum_{A_i \in C} y^*_{A_i} \right) +  \left(\usum{u \in \cT_{A_i, A_i^{\out}}, A_i \in C} y_{u}\right) \\
	&  \ge \left(\usum{u \in \cT_{A_i^{\inc}, A_i}, A_i \in C} y_{u} \right) + z_{C} + \left(\usum{u \in \cT_{A_i, A_i^{\out}}, A_i \in C} y_{u}\right) \\
	&  = \left(\usum{u \in \cT_{A_i^{\inc}, A_i}, A_i \in C} y_{u} \right) + \umin{A_i \in C}  \left(z_{\edge{A_i}{A_{i + 1}}} + \sum_{A_j \in C \setminus A_i} y_{A_j}' \right) + 
	\left(\usum{u \in \cT_{A_i, A_i^{\out}}, A_i \in C} y_{u}\right) \\
	& = \left(\usum{u \in \cT_{A_i^{\inc}, A_i}, A_i \in C} y_{u} \right) +  \umin{A_i \in C}  \left( z_{\edge{A_i}{A_{i + 1}}}  - y'_{A_i} \right) + \left(\sum_{A_i \in C} y'_{A_i}  \right) + \left(\usum{u \in \cT_{A_i, A_i^{\out}}, A_i \in C} y_{u}\right) \\
	& \ge \left(\usum{u \in \cT_{A_i^{\inc}, A_i}, A_i \in C} y_{u} \right) +  \mathbf{\epsilon} + \left(\sum_{A_i \in C} y'_{A_i}  \right) + \left(\usum{u \in \cT_{A_i, A_i^{\out}}, A_i \in C} y_{u}\right) \numberthis \label{eq:obj-opt-gen},
	\end{align*}
	where the first inequality follows from definition of $\vy^*$ and the second inequality follows from ~\eqref{eq:opt-eps-gen}. Since the terms $\left(\usum{u \in \cT_{A_i^{\inc}, A_i}, A_i \in C} y_{u} \right)$ and $\left(\usum{u \in \cT_{A_i, A_i^{\out}}, A_i \in C} y_{u}\right)$ are common to both sides of the above expression, we can ignore them for the rest of the argument. We start by considering the left-hand side of~\eqref{eq:obj-opt-gen}:
	\begin{align*}
	& \underset{A_i \in C}{\min} \left(z^*_{\edge{A_i}{A_{i + 1}}} - y^*_{A_i} \right) + \left(\sum_{A_i \in C} y^*_{A_i} \right) \\
	&= \underset{A_i \in C}{\min} \left(z_{\edge{A_i}{A_{i + 1}}} - \frac{p \cdot \epsilon}{k} -  y'_{A_i}  - \frac{\epsilon}{k} \right) +  \sum_{A_i \in C} \left(y_{A_i}' + \frac{\epsilon}{k} \right) \\
	& =  \underset{A_i \in C}{\min} \left(z_{\edge{A_i}{A_{i + 1}}} - y_{A_i}  - \frac{(p + 1)\epsilon}{k} \right) +  \sum_{A_i \in C} \left(y'_{A_i} + \frac{\epsilon}{k} \right) \\
	& =  \underset{A_i \in C}{\min} \left(z_{\edge{A_i}{A_{i + 1}}} - y_{A_i} \right) - \frac{p + 1}{k} \cdot \epsilon  + \epsilon + \sum_{A_i \in C} y'_{A_i} \\
	& = \frac{k - 1 - p}{k} \cdot \epsilon + \underset{A_i \in C}{\min} \left(z_{\edge{A_i}{A_{i + 1}}} - y_{A_i} \right) + \left(\sum_{A_i \in C} y'_{A_i}  \right) \\
	& \ge  \frac{k - 1 - p}{k} \cdot \epsilon  + \epsilon + \left(\sum_{A_i \in C} y'_{A_i}  \right) \\
	& =  \frac{2 \cdot k - 1 - p}{k} \cdot \epsilon + \left(\sum_{A_i \in C} y'_{A_i}  \right) \\
	& \ge \epsilon + \left(\sum_{A_i \in C} y'_{A_i}  \right), \text{ since } p \le k - 1.
	\end{align*}
	Here, the first equality follows by substituting the values of $z^*_{\edge{A_i}{A_{i + 1}}}$ and $y^*_{A_i}$ (from~\eqref{eq:zy-opt-gen}) for every $A_i \in C$ and the first inequality follows by substituting $\underset{A_i \in C}{\min} \left(z_{\edge{A_i}{A_{i + 1}}} - y_{A_i} \right)  \ge \epsilon$ (from~\eqref{eq:opt-eps-gen}). This proves~\eqref{eq:obj-opt-gen}, as desired.
\end{proof}

\subsection{Proof of Lemma~\ref{lemma:z-feasible-gen}} \label{sec:z-feasible-gen}
\begin{proof}[Proof of Lemma~\ref{lemma:z-feasible-gen}]
	To argue the feasibility of $\mathbf{z}^{*}$, we go over the constraints of $\LP^{(*)}$ one-by-one and show how each one is satisfied. Throughout the proof, we will use the values $z^*_{\edge{A_i}{A_{i + 1}}}$ and $y^*_{A_i}$ for every $A_i \in C$ from~\eqref{eq:zy-opt-gen}.
	
	We start by proving
	\begin{align*}
	z_{C} \le z^*_{\edge{A_i}{A_{i + 1}}} +  \sum_{A_j \in C \setminus A_i} y^*_{A_j}
	\end{align*}
	below. We have
	\begin{align*}
	z^*_{\edge{A_i}{A_{i + 1}}} +  \sum_{A_j \in C \setminus A_i} y^*_{A_j} & =  z_{\edge{A_i}{A_{i + 1}}}  - \frac{p \epsilon}{k} + \sum_{A_j \in C \setminus A_i} \left(y_{A_j}' + \frac{\epsilon}{k}  \right) \quad \forall A_i \in C \\
	& = \left( \frac{(k - 1) \epsilon}{k} - \frac{p \epsilon}{k}  \right) + z_{\edge{A_i}{A_{i + 1}}}  + \sum_{A_j \in C \setminus A_i}  y_{A_j}' \quad \forall A_i \in C \\
	& \ge z_{\edge{A_i}{A_{i + 1}}}  + \sum_{A_j \in C \setminus A_i}  y_{A_j}' \quad \forall A_i \in C \\
	& \ge z_{C}.
	\end{align*}
	Here, the first equation follows by direct substitution. Further, the first inequality follows from our assumption that $p \le k - 1$ and the second inequality follows from~\eqref{eq:ct27}. This shows that $\vz^*$ satisfies~\eqref{eq:ct27}. 
	
	Next, we argue 
	\begin{align*}
	\frac{z^*_{\edge{A_i}{A_{i + 1}}}}{p} + y^*_{A_{i + 1}} \le \log(L_{\edge{A_i}{A_{i + 1}}})  \quad \forall A_i \in C \numberthis \label{eq:cs-lp-2-gen}. 
	\end{align*}
	As in the earlier case, we substitute these values from~\eqref{eq:zy-opt-gen}
	\begin{align*}
	\frac{z^*_{\edge{A_i}{A_{i + 1}}}}{p} + y^*_{A_{i + 1}}  & = \frac{z_{\edge{A_i}{A_{i + 1}}} - \frac{p \cdot \epsilon}{k}}{p} +  y^*_{A_{i + 1}} \quad \forall A_i \in C\\
	& = \frac{z_{\edge{A_i}{A_{i + 1}}}}{p} - \frac{\epsilon}{k} +  y'_{A_{i + 1}} + \frac{\epsilon}{k} \quad \forall A_i \in C\\
	& = \frac{z_{\edge{A_i}{A_{i + 1}}}}{p} + y'_{A_{i + 1}} \quad \forall A_i \in C\\
	& \le \log(L_{{\edged{A_i}{A_{i + 1}}}}) \quad \forall A_i \in C,
	\end{align*}
	where the final inequality follows from~\eqref{eq:ct21} and proves~\eqref{eq:cs-lp-2-gen}, as required. Recall from definition of $\LP^{(**)}$, we have 
	\begin{align*}
	E_{A_i^{\inc}} = \left \{ \edge{A_i^{\inc}}{A_i}: A_i^{\inc} \in \cA_i^{\inc}, A_i \in C \right \} \\ 
	E_{A_i^{\out}} = \left \{ \edge{A_i}{A_i^{\out}}: A_i^{\out} \in \cA_i^{\out}, A_i \in C \right \} \\
	E_{R} = E \setminus \left( E(C) \cup E_{A_i^{\inc}} \cup E_{A_i^{\out}} \right).
	\end{align*} 
	We can do a similar argument for arguing the following two inequalities as well:
	\begin{align*}
	\frac{z^*_{\edge{A_i}{A_{i + 1}}}}{p} + y_{A_i^{\out}} \quad \le \quad \log \left(L_{{\edged{A_i}{A_i^{\out}}}} \right) \quad \forall \edge{A_i}{A_i^{\out}} \in E_{A_i^{\out}} \\
	\frac{y_{A_i^{\inc}}}{p}  + z^*_{\edge{A_i}{A_{i + 1}}}  \quad \le \quad \log \left(L_{{\edged{A_i^{\inc}}{A_i}}} \right)  \quad \forall \edge{A_i^{\inc}}{A_i} \in E_{A_i^{\inc}},
	\end{align*}
	which we argue below:
	\begin{align*}
	\frac{z^*_{\edge{A_i}{A_{i + 1}}}}{p} + y_{A_i^{\out}} & = \frac{z_{\edge{A_i}{A_{i + 1}}}}{p} - \frac{\epsilon}{k} + y_{A_i^{\out}} \\ 
	& \le \quad \frac{z_{\edge{A_i}{A_{i + 1}}}}{p} + y_{A_i^{\out}} \\
	& \le \quad \log \left(L_{{\edged{A_i}{A_i^{\out}}}} \right) \quad \forall \edge{A_i}{A_i^{\out}} \in E_{A_i^{\out}} \\ 
	\frac{y_{A_i^{\inc}}}{p}  + z^*_{\edge{A_i}{A_{i + 1}}} & = \frac{y_{A_i^{\inc}}}{p} + z_{\edge{A_i}{A_{i + 1}}}  - \frac{p \cdot \epsilon}{k} \\
	& \le  \quad \frac{y_{A_i^{\inc}}}{p} + z_{\edge{A_i}{A_{i + 1}}} \\
	& \le \quad \log \left(L_{{\edged{A_i^{\inc}}{A_i}}} \right)  \quad \forall \edge{A_i^{\inc}}{A_i} \in E_{A_i^{\inc}}.
	\end{align*}
	Note that the final inequality in both constraints come from~\eqref{eq:ct22} and~\eqref{eq:ct23} respectively. For all the remaining edges (i.e., in $E_R$), we have (directly by definition)
	\begin{align*} 
	\frac{y_{v}}{p} + y_u \le \log \left(L_{{\edged{v}{u}}} \right) \quad \forall \edge{v}{u} \in E_R.
	\end{align*}
	
	Next, we show
	\begin{align*}
	y^*_{A_i} \le z^*_{\edge{A_i}{A_{i + 1}}} \quad \forall A_i \in C,
	\end{align*}
	which we rewrite by substituting $z^*_{\edge{A_i}{A_{i + 1}}}$ and $y^*_{A_i} $ for every $A_i \in C$ as follows:
	\begin{align*}
	z^*_{\edge{A_i}{A_{i + 1}}} - y^*_{A_i}  & = z_{\edge{A_i}{A_{i + 1}}} - \frac{p \cdot \epsilon}{k} - \left(y'_{A_i} + \frac{\epsilon}{k} \right) \\
	& = z_{\edge{A_i}{A_{i + 1}}} - y'_{A_i} - \frac{(p + 1) \epsilon}{k} \\
	& \ge \epsilon - \frac{(p + 1) \cdot \epsilon}{k} \\
	& = \frac{\epsilon (k - 1 - p)}{k} \\
	& \ge 0.
	\end{align*}
	
	Finally, we have 
	\begin{align*}
	z^*_{\edge{A_i}{A_{i + 1}}} & \ge  y^*_{A_i} \quad \forall A_i \in C \\
	& \ge y_{A_i}' + \frac{\epsilon}{k}  \quad \forall A_i \in C \\
	& \ge 0,
	\end{align*}
	where the final inequality follows from the fact that $y_{A_i}', \epsilon \ge 0$. Note that this implies $y^*_{A_i} \ge 0$ as well for every $A_i \in C$. Since $z_{\edge{A_i}{A_{i + 1}}}, y'_{A_i} \ge 0$ for every $A_i \in C$, this implies $z_{C} \ge 0$. Further, we already have $y_{u} \ge 0 \quad \forall u \in \cT_{A_i^{\inc}, A_i} \cup \cT_{A_i, A_i^{\out}}, \forall A_i \in C$ (by definition). Thus, $\mathbf{z}^*$ satisfies all constraints of $\LP^{(*)}$ and is feasible, completing the proof.
\end{proof}

\subsection{Proof of Lemma~\ref{claim:gen-lp}} \label{sec:gen-lp}
We start by restating $\LP^{(*)}$ below.
\begin{align*}
\max \quad \left(\usum{u \in V \setminus V(C)} y_{u}  \right) + z_{C}  \tag{$\LP^{(*)}$} \\ 
\text{ s.t. } z_{C} \le z_{\edge{A_i}{A_{i + 1}}} + \sum_{A_j \in C \setminus A_i} y_{A_j}'   \quad \forall A_i \in C \numberthis \label{eq:ct27-app} \\
\frac{z_{\edge{A_i}{A_{i + 1}}}}{p} + y_{A_{i + 1}}' \le \log(L_{{\edged{A_i}{A_{i + 1}}}})  \quad \forall A_i \in C \numberthis \label{eq:ct21-app}\\
\frac{y_{A_i^{\inc}}}{p} + z_{\edge{A_i}{A_{i + 1}}}  \le \log(L_{{\edged{A_i^{\inc}}{A_i}}})  \quad \forall \edge{A_i^{\inc}}{A_i} \in E_{\inc} \numberthis \label{eq:ct23-app} \\
\frac{z_{\edge{A_i}{A_{i + 1}}}}{p} + y_{A_i^{\out}} \le \log(L_{{\edged{A_i}{A_i^{\out}}}}) \quad \forall \edge{A_i}{A_i^{\out}} \in E_{\out} \numberthis \label{eq:ct22-app} \\ 
\frac{y_v}{p} + y_{u} \le \log(L_{{\edged{v}{u}}}) \quad \forall \edge{v}{u} \in E_R \numberthis \label{eq:extra21-app} \\
y_{A_i}' \le z_{\edge{A_i}{A_{i + 1}}} \quad \forall A_i \in C \numberthis \label{eq:ct24-app} \\
z_{\edge{A_i}{A_{i + 1}}} \ge 0 \quad \forall A_i \in C \numberthis \label{eq:ct25-app} \\
z_{C}, y_{A_i}' \ge 0 \quad \forall A_i \in C \numberthis \label{eq:ct26-app} \\
y_{u} \ge 0 \quad \forall u \in V \setminus V(C). \numberthis \label{eq:extra22-app}
\end{align*}
\begin{proof} [Proof of Lemma~\ref{claim:gen-lp}]
	We start by showing that for every degree configuration $$\bard = \left(d_{\edge{v}{u}} \right)_{d_{\edge{v}{u}} \le L_{{\edge{v}{u}}}, \edge{v}{u} \in E},$$ there exists a feasible solution to $\LP^{(*)}$ with objective value $\B(\bard, G)$. Since $\B(G)$ is the maximum among all degree configurations and $\LP^{(*)}$ has a maximizing objective value, the proof immediately follows.
	
	Recall from~\eqref{eq:cyc-bdg-ub} that
	\begin{align*}
	\B'(\bard, G) & =  \left( \uprod{u \in \cT_{A_i^{\inc}, A_i}: A_i \in C} \cD_{u}\left(\bard \right) \right) \cdot \underset{A_i \in C}{\min} \left \{ \cD'_{i} \left(\bd \right) \cdot \prod_{A_j \in C, j \neq i}  \cD_{j} \left(\bard \right) \right \} \cdot \left( \uprod{u \in \cT_{A_i, A_i^{\out}}, A_i \in C} \cD_{u}\left(\bard \right) \right),
	\end{align*}
	where $\cD_{i} \left(\bard \right)$, $\cD'_{i} \left(\bard \right)$ and $\cD_{u} \left(\bard \right)$ are defined in~\eqref{eq:cyc-cd},~\eqref{eq:cyc-cd-prime} for every $A_i \in C$ and~\eqref{eq:cyc-cd-all} for every $u \in V \setminus C$.
	
	We now construct a feasible solution $\tilde{\vy}$ for $\LP^{(*)}$ with objective value $\B(\bard, G)$, starting with some notation. We define $\tilde{z}_{\edge{A_i}{A_{i + 1}}}(\bd)$ as
	\begin{align*}
	\min \left \{ \left( \log \left(d_{\edge{A_i^{\inc}}{A_i}} \right) \right)_{A_i^{\inc} \in \cA_i^{\inc}}, \left(\log \left(\frac{L_{{\edged{A_i}{A_i^{\out}}}}^{p}}{d_{\edge{A_i}{A_i^{\out}}}^{p}} \right) \right)_{A_i^{\out} \in \cA_i^{\out}}, \log \left(\frac{L_{{\edged{A_i}{A_{i + 1}}}}^{p}}{d_{\edge{A_i}{A_{i + 1}}}^{p}}\right)  \right \} \quad \forall A_i \in C, \numberthis \label{eq:gen-z} 
	\end{align*}
 followed by $\tilde{y}'_{A_i}(\bd)$ (for every $A_i \in C$)
\begin{align*}
	\min \left \{ \left( \log \left \{d_{\edge{A_i^{\inc}}{A_i}} \right \} \right)_{A_i^{\inc} \in \cA_i^{\inc}}, \left \{ \log \left(\frac{L_{{\edged{A_i}{A_i^{\out}}}}^{p}}{d_{\edge{A_i}{A_i^{\out}}}^{p}} \right) \right \}_{A_i^{\out} \in \cA_i^{\out}}, \log \left( d_{\edge{A_{i - 1}}{A_i}} \right), \log \left(\frac{L_{{\edged{A_i}{A_{i + 1}}}}^{p}}{d_{\edge{A_i}{A_{i + 1}}}^{p}} \right) \right \}  \numberthis \label{eq:gen-yprime} \\
	\end{align*}
and
\begin{align*}
	\tilde{z}_{C}(\bd) & = \underset{A_i \in C}{\min} \left \{ \tilde{z}_{\edge{A_i}{A_{i + 1}}}(\bd) + \sum_{A_j \in C \setminus A_i} \tilde{y}'_{A_j}(\bd) \right \} \numberthis \label{eq:gen-zc} \\ 
	\tilde{y}_{u}(\bd) & = \min\left \{\left \{ \log \left(d_{\edge{v}{u}} \right) \right \}_{\edge{v}{u} \in E}, \left \{ \frac{L_{{\edged{u}{w}}}^{p}}{d_{\edge{u}{w}}^p}  \right \}_{\edge{u}{w} \in E}   \right \} \quad \forall u \in V \setminus C \numberthis \label{eq:gen-yu}.
	\end{align*}
	We are now ready to define $\tilde{\vy}$:
	\begin{align*}
	\tilde{\vy} = \left( \left(\tilde{y}_{u}(\bd)\right)_{u \in \cT_{A_i^{\inc}, A_i}, A_i \in C},  \left( \tilde{z}_{\edge{A_i}{A_{i + 1}}} (\bd) \right)_{A_i \in C}, \left( \tilde{y}'_{A_i} (\bd) \right)_{A_i \in C},  \tilde{z}_{C} (\bd), \left(\tilde{y}_{u}(\bd)\right)_{u \in \cT_{A_i, A_i^{\out}}, A_i \in C} \right)
	\end{align*}
	
	Next, we show that $\bar{y}$ satisfies all constraints in $\LP^{(*)}$.
	\begin{itemize}
		\item {We start with~\eqref{eq:ct27-app}:
			\begin{align*}
			\tilde{z}_{C}(\bd) & = \underset{A_i \in C}{\min} \left \{ \tilde{z}_{\edge{A_i}{A_{i + 1}}} (\bd)+ \sum_{A_j \in C \setminus A_i} \tilde{y}'_{A_j}(\bd) \right \} \\
			& \le \tilde{z}_{\edge{A_i}{A_{i + 1}}}(\bd) + \sum_{A_j \in C \setminus A_i} \tilde{y}_{A_j}'(\bd) \quad \forall A_i \in C,
			\end{align*}
			where the equation follows from~\eqref{eq:gen-zc} and the inequality follows from the definition of $\min$.}
		\item {Constraint~\eqref{eq:ct21-app}:
			\begin{align*}
			& \frac{\tilde{z}_{\edge{A_i}{A_{i + 1}}}(\bd)}{p} +  \tilde{y}'_{A_i}(\bd)  \\
			& = \frac{\min \left \{ \left \{\log \left(d_{\edge{A_i^{\inc}}{A_i}} \right) \right \}_{A_i^{\inc} \in \cA_i^{\inc}}, \left \{\log \left(\frac{L_{{\edged{A_i}{A_i^{\out}}}}^{p}}{d_{\edge{A_i}{A_i^{\out}}}^{p}} \right) \right \}_{A_i^{\out} \in \cA_i^{\out}}, \log \left(\frac{L_{{\edged{A_i}{A_{i + 1}}}}^{p}}{d_{\edge{A_i}{A_{i + 1}}}^{p}}\right)  \right \}}{p} \\
			& + \min \left \{\left \{ \log \left (d_{\edge{A_i^{\inc}}{A_i}} \right) \right \}_{A_i^{\inc} \in \cA_i^{\inc}}, \left \{\log \left(\frac{L_{{\edged{A_i}{A_i^{\out}}}}^{p}}{d_{\edge{A_i}{A_i^{\out}}}^{p}} \right) \right \}_{A_i^{\out} \in \cA_i^{\out}}, \log \left( d_{\edge{A_{i - 1}}{A_i}} \right), \log \left(\frac{L_{{\edged{A_i}{A_{i + 1}}}}^{p}}{d_{\edge{A_i}{A_{i + 1}}}^{p}} \right) \right \} \quad \forall A_i \in C \\
			&  \le  \frac{\log \left(\frac{L_{{\edged{A_i}{A_{i + 1}}}}^{p}}{d_{\edge{A_i}{A_{i + 1}}}^{p}}\right)} {p} +  \log \left(d_{\edge{A_i}{A_{i + 1}}}\right)   \quad \forall A_i \in C \\
			& = \log \left( L_{\edged{A_i}{A_{i + 1}}} \right) \quad \forall A_i \in C,
			\end{align*}
			where the first equation follows by direct substitution and the inequality follows from the definition of $\min$.}
		\item {Constraint~\eqref{eq:ct22-app}: For every $\edge{A_i^{\inc}}{A_i} \in E_{\inc}$, we have
			\begin{align*}
			& \frac{\tilde{y}_{A_{i}^{\inc}}(\bd)}{p} + \tilde{z}_{\edge{A_i}{A_{i + 1}}}(\bd) \\
			& = \frac{\min \left \{  (d_{\edge{v}{A_i^{\inc}}})_{\edge{v}{A_i^{\inc}} \in E},  \left( \frac{L_{{\edged{A_i^{\inc}}{w}}}^{p}}{d_{\edge{A_i^{\inc}}{w}}^p}  \right)_{\edge{A_i^{\inc}}{w} \in E}   \right \}}{p} \\
			& +  \min \left \{ \left( \log \left(d_{\edge{A_i^{\inc}}{A_i}} \right) \right)_{A_i^{\inc} \in \cA_i^{\inc}}, \left(\log \left(\frac{L_{{\edged{A_i}{A_i^{\out}}}}^{p}}{d_{\edge{A_i}{A_i^{\out}}}^{p}} \right) \right)_{A_i^{\out} \in \cA_i^{\out}}, \log \left(\frac{L_{{\edged{A_i}{A_{i + 1}}}}^{p}}{d_{\edge{A_i}{A_{i + 1}}}^{p}}\right)  \right \}  \\
			& \le  \frac{\log \left(\frac{L_{{\edged{A_i^{\inc}}{A_i}}}^{p}}{d_{\edge{A_i^{\inc}}{A_i}}^p} \right)}{p}  +  \log \left(d_{\edge{A_i^{\inc}}{A_i}} \right)  \quad \forall \edge{A_i^{\inc}}{A_i} \in E_{\inc} \\
			& =  \log \left(L_{{\edged{A_{i}^{\inc}}{A_i}}} \right) \quad \forall \edge{A_i^{\inc}}{A_i} \in E_{\inc},
			\end{align*}
			where the first equation follows by direct substitution and the inequality follows from the definition of $\min$.}
		\item {Constraint~\eqref{eq:ct23-app}:
			\begin{align*}
		& \frac{\tilde{z}_{\edge{A_i}{A_{i + 1}}}(\bd)}{p} + \tilde{y}_{A_{i}^{\out}}(\bd) \\
			& = \frac{\min \left \{ \left( \log \left(d_{\edge{A_i^{\inc}}{A_i}} \right) \right)_{A_i^{\inc} \in \cA_i^{\inc}}, \left(\log \left(\frac{L_{{\edged{A_i}{A_i^{\out}}}}^{p}}{d_{\edge{A_i}{A_i^{\out}}}^{p}} \right) \right)_{A_i^{\out} \in \cA_i^{\out}}, \log \left(\frac{L_{{\edged{A_i}{A_{i + 1}}}}^{p}}{d_{\edge{A_i}{A_{i + 1}}}^{p}}\right)  \right \}}{p} \\
			& +  \min \left \{ \left( \log \left( d_{\edge{v}{A_{i}^{\out}}} \right) \right)_{\edge{v}{A_i^{\out} \in E} }, \left( \frac{L_{{\edged{A_i^{\out}}{w}}}^{p}}{d_{\edge{A_i^{\out}}{w}}^p}  \right)_{\edge{A_i^{\out}}{w} \in E} \right \} \quad \forall \edge{A_i}{A_i^{\out}} \in E_{\out} \\
			& \le  \frac{\log \left(\frac{L_{{\edged{A_i}{A_{i}^{\out}}}}^{p}}{d_{\edge{A_i}{A_{i}^{\out}}}^{p}}\right)} {p} + \log \left(d_{\edge{A_i}{A_i^{\out}}} \right) \quad \forall \edge{A_i}{A_i^{\out}} \in E_{\out} \\ 
			& = \log \left(L_{{\edged{A_i}{A_{i}^{\out}}}} \right) \quad \forall \edge{A_i}{A_i^{\out}} \in E_{\out},
			\end{align*}
			where the first equation follows by direct substitution and the inequality follows from the definition of $\min$.}
		\item {Constraint~\eqref{eq:extra21-app}:
			\begin{align*}
			& \frac{\ty_v(\bd)}{p} + \ty_u(\bd) \\
			& = \frac{\min\left \{\left \{ \log \left(d_{\edge{v'}{v}} \right) \right \}_{\edge{v'}{v} \in E}, \left \{ \frac{L_{{\edged{v}{w}}}^{p}}{d_{\edge{v}{w}}^p}  \right \}_{\edge{v}{w} \in E}   \right \}}{p}  \\ 
			& + \min\left \{\left \{ \log \left(d_{\edge{v}{u}} \right) \right \}_{\edge{v}{u} \in E}, \left \{ \frac{L_{{\edged{u}{w}}}^{p}}{d_{\edge{u}{w}}^p}   \right \}_{\edge{u}{w} \in E}   \right \} \quad \forall \edge{v}{u} \in E_R \\
			& \le \frac{ \log \left( \frac{L_{{\edged{v}{u}}}^{p}}{d_{\edge{v}{u}}^{p}}\right)}{p} + \log \left( d_{\edge{v}{u}}\right) \quad \forall \edge{v}{u} \in E_R \\
			& = \log \left( L_{\edge{v}{u}} \right) \quad \forall \edge{v}{u} \in E_R,
			\end{align*}
			where the first equation follows by direct substitution and the inequality follows from the definition of $\min$.
		}
		\item {Constraint~\eqref{eq:ct24-app}: $\tilde{y}_{A_i}'(\bd)$ (for every $A_i \in C$) equals
			\begin{align*}
			 \min \left \{ \left \{ \log \left(d_{\edge{A_i^{\inc}}{A_i}} \right) \right \}_{A_i^{\inc} \in \cA_i^{\inc}}, \left(\log \left(\frac{L_{{\edged{A_i}{A_i^{\out}}}}^{p}}{d_{\edge{A_i}{A_i^{\out}}}^{p}} \right) \right)_{A_i^{\out} \in \cA_i^{\out}}, \log \left( d_{\edge{A_{i - 1}}{A_i}} \right), \log \left(\frac{L_{{\edged{A_i}{A_{i + 1}}}}^{p}}{d_{\edge{A_i}{A_{i + 1}}}^{p}} \right) \right \},
			 \end{align*}
		 which in turn is upper bounded by
		 \begin{align*}
			& \min \left \{\left( \log \left(d_{\edge{A_i^{\inc}}{A_i}} \right) \right)_{A_i^{\inc} \in \cA_i^{\inc}}, \left(\log \left(\frac{L_{\edge{A_i}{A_i^{\out}}}^{p}}{d_{\edge{A_i}{A_i^{\out}}}^{p}} \right) \right)_{A_i^{\out} \in \cA_i^{\out}}, \log \left(\frac{L_{\edge{A_i}{A_{i + 1}}}^{p}}{d_{\edge{A_i}{A_{i + 1}}}^{p}}\right)  \right \} \quad \forall A_i \in C \\
			& =  \tilde{z}_{\edge{A_i}{A_{i + 1}}}(\bd) \quad \forall A_i \in C,
			\end{align*}
			where the first inequality follows by direct substitution and the inequality follows from the definition of $\min$.}
		\item {Finally, the constraints~\eqref{eq:ct25-app}~\eqref{eq:ct26-app} and~\eqref{eq:extra22-app} are satisfied directly by definition. We state them here for completeness:
			\begin{align*}
			\tilde{z}_{\edge{A_i}{A_{i + 1}}}(\bd), \tilde{y}_{A_i}'(\bd), \tilde{z}_{C} \ge 0 \quad \forall A_i \in C  \\
			\tilde{y}_{u}(\bd) \ge 0 \quad \forall u \in V \setminus V(C). 
			\end{align*}} 
	\end{itemize}
	This completes the proof.
\end{proof}

\subsection{Proof of Corollary~\ref{cor:cycle-trees}} \label{sec:cycle-trees}
The proof of Corollary~\ref{cor:cycle-trees} follows directly from the proof of the claim below.
\begin{claim} \label{claim:cycle-trees}
	Let $G$ be a connected graph (in the undirected sense) with a directed cycle $C$ such that $|E| = |V|$. Consider the graph without edges in the cycle $C$ i.e., $G' = (V(G), E(G) \setminus E(C))$ and for every vertex $A_i \in C$, let $\cA_i^{\inc}$ and $\cA_i^{\out}$ denote the set of incoming and outgoing edges from $A_i$ that are not in $E(C)$. Then, $G'$ has $|C|$ edge disjoint trees with each containing a unique node $A_i$ in $C$.
\end{claim}
\begin{proof} [Proof of Claim~\ref{claim:cycle-trees}]
	We prove this claim in two steps -- first, we argue that the graph $G' = (V(G), E(G) \setminus E(C))$ is a forest with $|C|$-edge disjoint trees. Then, we argue that all these trees contain an unique node in $C$. 
	
	The proof is by contradiction. Assume that the $|C|$ trees are not edge disjoint i.e., there exists at least one edge $\edge{v}{u}$ between these trees. We now argue that this would imply $|E(G)| > |V(G)|$ and start by counting the number of edges in $G'$. In particular, we have
	\begin{align*}
	|E(G')| & = 1 + \usum{A_i \in C} \left(|\cT_{A_i}| - 1\right) \\
	& = 1 - |C| + \usum{A_i \in C} |\cT_{A_i}| \\
	& = 1 - |C| + |V|,
	\end{align*}
	where the last equation follows from the fact that $\ucup{A_i \in C} |\cT_{A_i}| = |V|$. Note that this implies $|E| =  |E(G')| + |E(C)| = |V| - |C| + 1 + |C| = |V| + 1$, which contradicts our earlier assumption that $|E| = |V|$. 
	
	To complete the proof, we argue that each of these trees contains an unique node in $C$. Note that if otherwise i.e., the tree contains no nodes from $C$, then that implies the original graph $G$ was not connected in the undirected sense.
\end{proof}

%% file: appendix_sec5.tex
\section{Missing Details in Section~\ref{sec:panda}} \label{app:panda}

\subsection{Missing Details in Section~\ref{sec:dlow}} \label{app:dlow}
We first restate Theorem~\ref{thm:dsmall} from the main paper.
\begin{thm}
For any $G$, $L$ and $d$ with $d^2 \le L$ satisfying Assumption~\ref{assump:deg} and an optimal star cover $E(S_1(G), T_1(G))$ and $E(S_2(G), T_{2}(G))$,  Algorithm~\ref{algo:generic-ub} runs in time linear in
\begin{align*}
2^{2(|V| + |C(G)| + |S_1(G)| + |T_1(G)|)} \left( \uprod{C_i \in C(G)} L d^{|C_i| - 2}  \right) \cdot \left(  \left(\frac{L}{d}\right)^{|S_1(G)|} \cdot d^{|T_1(G)|} \right) \cdot L^{|S_2(G)|} \cdot d^{|\rho(G)|} \numberthis \label{eq:dlow-ub-app} 
\end{align*}
for instances $\cI = \{R_{\edge{v}{u}}: ||R_{\edge{v}{u}}||_{1} \le L, ||R_{\edge{v}{u}}||_{\infty} \le d, \edge{v}{u} \in E\}$. Further, $|\TJ|$ is at most~\eqref{eq:dlow-ub-app}. Finally, there exists an instance $I \in \cI$ such that
\begin{align*}
|\TJ| &\ge \frac{1}{2^{|V|}} \left( \uprod{C_i \in C(G)} L d^{|C_i| - 2}  \right)  \cdot \left( \left(\frac{L}{d}\right)^{|S_1(G)|} \cdot d^{|T_1(G)|} \right) \cdot L^{|S_2(G)|} \cdot d^{|\rho(G)|} \numberthis \label{eq:dlow-lb-app}.
\end{align*}	
\end{thm}
We first prove~\eqref{eq:dlow-lb-app},  followed by proving an upper bound of~\eqref{eq:dlow-ub-app} for $|\TJ|$ and the runtime of Algorithm~\ref{algo:generic-ub} follows as a corollary.
\subsubsection{Proof of~\eqref{eq:dlow-lb}} \label{app:dlow-lb}
We construct our lower bound instance for $G$. We define $L'$ and $d'$ to be a power of two in $[\frac{L}{2}, L]$ and $[\frac{d}{2}, d]$ respectively.\footnote{We would like note there that in our arguments in the section, we use the undirected edge $\edgeud{v}{u}$ instead of the standard directed version $\edge{v}{u}$ used in the rest of the paper. Note that this doesn't break correctness since for $\ell_{1}$-norm bounds, the direction does not matter and we (implicitly) assume that all tuples in the relation $R_{\edgeud{v}{u}}$ are directed from $v$ to $u$ (i.e., in the direction of the $\ell_{\infty}$ constraint).}

For each $\edge{v}{u}$ in $C(G)$, we define $R_{\edgeud{v}{u}}$ to be (note we used $E$ here and not $E(S_1,T_1)$):
\begin{align*}
\left \{ [ (j-1) \cdot d' + 1,  j \cdot d']  \times [ (j - 1) \cdot d' + 1, j \cdot d' ] : j \in \left[\frac{L'}{(d')^2} \right]  \right \}. \numberthis \label{eq:scc-lb}
\end{align*}
Further, for every $\edge{s_1}{t_1} \in E: s_1 \in S_1(G), t_1 \in T_1(G)$, we define 
\begin{align*}
R_{\edgeud{s_1}{t_1}} & = \left \{ \left[\frac{L'}{d'} \right] \times [d'] \right \} \numberthis \label{eq:s1-t1}
\end{align*}
and for every edge $\edge{s_2}{t_2} \in E: s_2 \in S_2(G), t_2 \in T_2(G)$ and
\begin{align*}
R_{\edgeud{s_2}{t_2}} & = \{ [L'] \times [1] \} \numberthis \label{eq:s2-t2}.
\end{align*}
All nodes $u$ that have not been assigned a domain yet get $\Dom(u)=[d']$ and all the unasigned relations are set as
\begin{align*}
R_{\edgeud{v}{u}} & =\Dom(v)\times \Dom(u)  \numberthis \label{eq:rem-lb}.
\end{align*}
We make the following assumption based on our construction above:
\begin{align*}
||R_{\edgeud{v}{u}}||_{1} & \le L \quad \forall \edgeud{v}{u} \in E \numberthis \label{eq:r-edge-1}\\
||R_{\edge{v}{u}}||_{\infty} & \le d \quad \forall \edge{v}{u} \in E \numberthis \label{eq:r-edge-2}.
\end{align*}
Assuming the above is true, we argue our lower bound. For each $C_i \in C(G)$, we have a lower bound of
\begin{align*}
\frac{L'}{(d')^2} \uprod{u \in V(C_i)} d' & = \frac{L'}{(d')^2} (d')^{|C_i|} \\
& \ge \frac{1}{2^{|C_i|}} L d^{|C_i| - 2}
\end{align*}
on the size of the join projected down to relations in $C_i$. Further, note that these bounds are `independent' for each $C_i$ and once we fix the projection of output tuples in $C(G)$, the rest is just a Cartesian product of the domain size for each vertex (which are $L'/d',d',L',1$ and $d$ for vertices in $S_1(G), T_1(G), S_2(G), T_2(G)$ and $\rho(G)$ respectively). Combining these bounds together, we get the desired lower bound of
\begin{align*}
\frac{1}{2^{|V|}} \left(  \uprod{C_i \in C(G)} L d^{|C_i| - 2}  \right)  \left(\frac{L}{d}\right)^{|S_1(G)|}  d^{|T_1(G)|} L^{|S_2(G)|} d^{|\rho(G)|},
\end{align*}
as required.

To complete the proof , we still need to argue~\eqref{eq:r-edge-1} and~\eqref{eq:r-edge-2}. For each $C_i \in C(G)$, we have (by definition) for every $\edge{v}{u} \in E(C_i)$:
\begin{align*}
||R_{\edgeud{v}{u}}||_{1} & = \usum{j \in \left[\frac{L'}{(d')^2} \right]} (d')^2 \le L \\
||R_{\edge{v}{u}}||_{\infty} & = \umax{j \in \left[\frac{L'}{(d')^2} \right]} d' \le  d.
\end{align*}
Further, for every $\edge{s_1}{t_1} \in E: s_1 \in S_1(G), t_1 \in T_1(G)$, we have
\begin{align*}
||R_{\edgeud{s_1}{t_1}}||_{1} = \frac{L'}{d'} \times d' \le L, \quad ||R_{\edge{s_1}{t_1}}||_{\infty}  = \umax{j \in \left[\frac{L'}{d'}\right]} d'  \le d
\end{align*}
and for every edge $\edge{s_2}{t_2} \in E: s_2 \in S_2(G), t_2 \in T_2(G)$, we have
\begin{align*}
|| R_{\edgeud{s_2}{t_2}} ||_{1}  \le L, \quad || R_{\edge{s_2}{t_2}} ||_{\infty} = 1 \le d.
\end{align*}
Next, we reason about the crossing edges. For edges $\edge{v}{t_1}$ between $C_i$ (for some $C_i \in C(G)$) and $T_1(G)$, we have
\begin{align*}
||R_{\edgeud{v}{t_1}}||_{1} & = \usum{j \in \left[\frac{L'}{(d')^2} \right]} d'  =  \frac{L'}{d'} \le L \\ ||R_{\edge{v}{t_1}}||_{\infty} &= \umax{j \in \left[\frac{L'}{(d')^2} \right]} d' \le d
\end{align*}
and $\edge{v}{t_2}$ between $C_i$ (for some $C_i \in C(G)$) and $T_2(G)$, we have
\begin{align*}
||R_{\edgeud{v}{t_2}}||_{1} = \usum{j \in \left[\frac{L'}{(d')^2} \right]} 1  \le L, \quad ||R_{\edge{v}{t_2}}||_{\infty} = \umax{j \in \left[\frac{L'}{(d')^2} \right]} 1  \le d.
\end{align*}
Finally, for each remaining edge, we have
\begin{align*}
||R_{\edgeud{v}{u}}||_{1} = d'^2\le d^2 \le L, \quad ||R_{\edge{v}{u}}||_{\infty} = \umax{j \in [d']}d' \le d.
\end{align*}
This completes the proof.

\subsubsection{Proof of~\eqref{eq:dlow-ub}} \label{app:dlow-ub}
We start by recalling the notation we defined for this section. Recall that each edge $\edgeud{v}{u}$ has $\ell_{\infty}$-norm constraints of the form $L_{\edgeinfty{v}{u}}$ and as a result, we consider $G$ to be the directed graph on each $\edge{v}{u} \in E$. Now, we decompose vertices in $V(G)$ into four buckets -- $(1)$ Set of non-trivial source Strongly Connected Components (SCCs) (i.e., source SCCs with at least two vertices) $C(G)$, $(2)$ The remaining sources (source SCCs with one vertex) $S(G)$, $(3)$ Set of vertices $T(G)$, where each vertex is connected by at least one vertex in $S(G)$ (through an incoming edge) and $(4)$ The remaining set of vertices $\rho(G)$. We consider the induced subgraph on vertices $S(G) \cup T(G)$, which we use to further decompose it into induced subgraphs $(S_1(G), T_1(G))$ and $(S_2(G), T_2(G))$ respectively.

We consider the induced subgraph on vertices $S(G) \cup T(G)$. We further partition $S$ into $S_1(G)$ and $S_2(G)$ and $T$ into $T_1(G)$ and $T_2(G)$ as follows. We choose a subset $E(S_1(G), T_1(G)) \subset E(G)$ to be a (disjoint) set of stars\footnote{A star with $n$ vertices is where one vertex has degree $n-1$ (which we call the {\em center}) and the remaining vertices have degree $1$ (which we call {\em leaves}).} with each $s_1 \in S_1$ as the center (in the undirected sense) and each $t_1 \in T_1(G)$ s.t. $\edge{s_1}{t_1} \in E(S_1(G), T_1(G))$ (for the fixed $s_1$) as a leaf. Similarly, we define $E(S_2, T_2) \subset E$ to be another (disjoint) set of disjoint stars such with each $t_2 \in T(G_2)$ is a center (in the undirected sense) and each $s_2 \in S_{2}(G)$ with $\edge{s_2}{t_2} \in E(S_2(G), T_2(G))$ (for the fixed $t_2$) as a leaf. We pick $S_i(G),T_i(G)$ and $E(S_i,T_i)$ for $i\in [2]$ that {\em minimizes} the size of this star cover, i.e. minimizes $\abs{E(S_1,T_1)}+\abs{E(S_2,T_2)}=\abs{T_1(G)}+\abs{S_2(G)}$. We show in the section below as to how to compute an optimal star cover of this kind using the AGM LP and also argue why minimizing the star cover size also minimizes our bound in Theorem~\ref{thm:dsmall}.

For each $C_i \in C(G)$, we fix an arbitrary edge $\edge{v_i}{u_i} \in E(C_i)$ and drop all incoming/outgoing edges from $v_i$ and all incoming edges from $u_i$. Note that each $C_i$ has a single source vertex. For non-source non-trivial SCCs, we can drop the minimal subset of edges to make them acyclic. For each $C_i$, we will be treating the edge $\edge{v_i}{u_i}$ differently compared to the other ones. Since we are working with a spanning subgraph of $G$ to prove our upper bound~\eqref{eq:dlow-ub-app}, note that we are answering Question~\ref{ques:acyclic} in affirmative in this process, in this setting.

We are now ready to restate~\eqref{eq:dlow-ub} here and we would like to prove $|\TJ|$ is at most
\begin{align*}
2^{2(|V| + |C(G)| + |S_1(G)| + |T_1(G)|)} \left( \uprod{C_i \in C(G)} L d^{|C_i| - 2}  \right) \cdot \left(  \left(\frac{L}{d}\right)^{|S_1(G)|} \cdot d^{|T_1(G)|} \right) \cdot L^{|S_2(G)|} \cdot d^{|\rho(G)|}.
\end{align*}
Invoking Theorem~\ref{claim:gen-acyclic} on $G$, we get
\begin{align*}
& \B(\bard, G) \\
& \le  \left( \uprod{\edge{v_i}{u_i}: C_i \in C(G)} \frac{2 \cdot L_{\edged{v_i}{u_i}}}{d_{\edge{v_i}{u_i}}} d_{\edge{v_i}{u_i}} \cdot \uprod{u \in V(C_i) \setminus \{v_i, u_i\}} \quad \umin{\edge{v}{u} \in E(C_i)} d_{\edge{v}{u}} \right) \\
&  \quad \cdot \left( \uprod{s_1 \in S_1(G)} 
\umin{\edge{s_1}{t_1} \in E(S_1, T_1)} \frac{2 \cdot L_{\edge{s_1}{t_1}}}{d_{\edged{s_1}{t_1}}} \right)  \cdot \left( \uprod{t_1 \in T_1(G)} \umin{\edge{s_1}{t_1} \in E(S_1, T_1)} d_{\edge{s_1}{t_1}}  \right) \\
& \quad \cdot \left( \uprod{s_2 \in S_2(G)} \umin{\edge{s_2}{t_2} \in E(S_2, T_2)} \frac{2 \cdot L_{\edged{s_2}{t_2}}}{d_{\edge{s_2}{t_2}}} \right) \cdot \left( \uprod{t_2 \in T_2(G)} \umin{\edge{s_2}{t_2} \in E(S_2, T_2)} d_{\edge{s_2}{t_2}}  \right) \\
& \quad  \left( \uprod{u \in \rho(G)} \umin{\edge{v}{u} \in E} d_{\edge{v}{u}} \right) \\
& \le 2^{|C(G)| + |S_1(G)| + |S_2(G)|}   \cdot \left( \uprod{\edge{v_i}{u_i}: C_i \in C(G)}  L_{\edged{v_i}{u_i}} \cdot \uprod{u \in V(C_i) \setminus \{v_i, u_i\}} \sqrt[|\inc(u)|] { \uprod{\edge{v}{u} \in E(C_i)} d_{\edge{v}{u}} } \right) \\
& \quad \cdot \left( \uprod{s_1 \in S_1(G)} \frac{L_{\edged{s_1}{t_{s_1}}}}{d_{\edge{s_1}{t_{s_1}}}} \cdot d_{\edge{s_1}{t_{s_1}}}  
\cdot \uprod{t_1 \in T_1(G) \setminus t_{s_1}}\sqrt[|\inc(t_1)|]{ \uprod{\edge{s_1}{t_1} \in E(S_1, T_1)} d_{\edge{s_1}{t_1}}}   \right) \\
& \quad \cdot \left( \uprod{s_2 \in S_2(G)}  \umin{\edge{s_2}{t_2} \in E(S_2, T_2)} \frac{L_{\edged{s_2}{t_2}}}{d_{\edge{s_2}{t_2}}} \uprod{t_2 \in T_2(G)} \umin{\edge{s_2}{t_2} \in E(S_2, T_2)} d_{\edge{s_2}{t_2}}  \right) \\
& \quad \cdot \left( \uprod{u \in \rho(G)} \sqrt[\inc(u)|]{ \uprod{\edge{v}{u} \in E} d_{\edge{v}{u}}  } \right) \\
& \le 2^{|C(G)| + |S_1(G)| + |S_2(G)|}  \cdot  \left( \uprod{\edge{v_i}{u_i}: C_i \in C(G)} L_{\edged{v_i}{u_i}} \cdot \uprod{u \in V(C_i) \setminus \{v_i, u_i\}} \sqrt[|\inc(u)|] { \uprod{\edge{v}{u} \in E(C_i)} d_{\edge{v}{u}} } \right) \\
& \quad \cdot \left( \uprod{s_1 \in S_1(G)}  L_{\edged{s_1}{t_{s_1}}}\right) \cdot \left( \uprod{t_1 \in T_1(G) \setminus \{t_{s_1} : s_1 \in S_1(G)\} }\sqrt[\inc(t_1)|]{ \uprod{\edge{s_1}{t_1} \in E(S_1, T_1)} d_{\edge{s_1}{t_1}}}  \right)  \\
& \quad \cdot \left( \uprod{\edge{s_2}{t_2} \in (S_2(G), T_2(G))}  L_{\edged{s_2}{t_2}} \right) \cdot \left( \uprod{u \in \rho(G)} \sqrt[\inc(u)|]{ \uprod{\edge{v}{u} \in E} d_{\edge{v}{u}}  } \right).
\end{align*}
Here, the second inequality follows from the fact $|\inc(u)| \ge 1$ for every $u \in V \setminus ( \ucup{C_i \in C(G)} \{v_i\} \cup S_1(G) \cup S_2(G))$ (note that all these sets of vertices are source vertices by definition and our contruction). For each $s_1 \in S_1(G)$, note that there is at least one vertex $t_1 \in T_1(G)$ such that $\edge{s_1}{t_1} \in E(S_1, T_1)$ (by definition of the latter). We define $T_{S_1(G)} = \{t_{s_1} : s_1\in S(G), \edge{s_1}{t_1} \in E(S_1, T_1) \}$. It follows that $|T_{S_1}(G)| = |S_1(G)|$.

Summing up $\B(\bard, G)$ over all possible degree configurations $\bd$, we get
\begin{align*}
& \usum{\bard = (d_{\edge{v}{u}})_{\edge{v}{u} \in E(G)}} \B(\bard, G) \\
& \le \usum{\bard = (d_{\edge{v}{u}})_{\edge{v}{u} \in E(G)}}  2^{|C(G)| + |S_1(G)| + |S_2(G)|} \\
& \quad \cdot \left( \uprod{\edge{v_i}{u_i}: C_i \in C(G)} L_{\edged{v_i}{u_i}} \cdot \uprod{u \in V(C_i) \setminus \{v_i, u_i\}} \sqrt[|\inc(u)|] { \uprod{\edge{v}{u} \in E(C_i)} d_{\edge{v}{u}} } \right) \\
& \quad \left( \uprod{s_1 \in S_1(G)} L_{\edged{s_1}{t_{s_1}}}\right) \cdot \left( \uprod{t_1 \in T_1(G) \setminus T_{S_1(G)}}\sqrt[\inc(t_1)|]{ \uprod{\edge{s_1}{t_1} \in E(S_1, T_1)} d_{\edge{s_1}{t_1}}}  \right)   \\
& \quad  \cdot \left( \uprod{\edged{s_2}{t_2} \in E(S_2, T_2)}  L_{\edged{s_2}{t_2}} \right) \cdot \left( \uprod{u \in \rho(G)} \sqrt[\inc(u)|]{ \uprod{\edge{v}{u} \in E} d_{\edge{v}{u}}  } \right) \\
& \le 2^{|C(G)| + |S_1(G)| + |S_2(G)|}  \\
& \quad \cdot \left( \uprod{\edge{v_i}{u_i}: C_i \in C(G)} \left( \usum{d_{\edge{v_i}{u_i}} \le L_{\edged{v_i}{u_i}}} L_{\edged{v_i}{u_i}} \right) \cdot \uprod{u \in V(C_i) \setminus \{v_i, u_i\}} \quad \uprod{\edge{v}{u} \in E(C_i)} \left( \usum{d_{\edge{v}{u}} \le d} d_{\edge{v}{u}} \right)^{\frac{1}{|\inc(u)|}} \right)  \\
& \quad \cdot  \left( \usum{d_{\edge{s_1}{t_1}} \le L: \edge{s_1}{t_1} \in E(S_1, T_1)} \quad \uprod{s_1 \in S_1(G)}  L_{\edged{s_1}{t_{s_1}}} \right) \\
& \quad \cdot \left( \uprod{t_1 \in T_1(G) \setminus T_{S_1}(G)} \quad \uprod{\edge{s_1}{t_1} \in E(S_1, T_1)} \left(\usum{d_{\edge{s_1}{t_1}} \le d}d_{\edge{s_1}{t_1}} \right)^{\frac{1}{|\inc(t_1)|}}  \right) \\
& \quad \cdot \left( \usum{d_{\edge{s_2}{t_2}} \le L: \edge{s_2}{t_2} \in E(S_2, T_2)} \uprod{\edge{s_2}{t_2} \in E(S_2, T_2)}  L_{\edged{s_2}{t_2}} \right) \cdot \left( \uprod{u \in \rho(G)}  \uprod{\edge{v}{u} \in E} \left( \usum{d_{\edge{v}{u}} \le d} d_{\edge{v}{u}} \right)^{\frac{1}{|\inc(u)|}} \right) \numberthis \label{eq:dsmall-holder-1}  \\
& \le 2^{|C(G)| + |S_1(G)| + |S_2(G)|} \cdot \left( \uprod{\edge{v_i}{u_i}: C_i \in C(G)} L_{\edge{v_i}{u_i}} \cdot \left( \uprod{u \in V(C_i) \setminus \{v_i, u_i\}} \quad \uprod{\edge{v}{u} \in E(C_i)} (2d)^{\frac{1}{|\inc(u)|}} \right) \right) \\
& \quad \cdot \left( \usum{d_{\edge{s_1}{t_1}} \le L: \edge{s_1}{t_1} \in E(S_1, T_1)} \quad \uprod{s_1 \in S_1(G)} L_{\edged{s_1}{t_1}} \right) \cdot  \left( \uprod{t_1 \in T_1(G) \setminus T_{S}(G)} \quad \uprod{\edge{s_1}{t_1} \in E(S_1, T_1)}  (2d)^{\frac{1}{|\inc(t_1)|}} \right) \\
& \quad  \cdot \left( \usum{d_{\edge{s_2}{t_2}} \le L: \edgeud{s_2}{t_2} \in E(S_2, T_2)}  \uprod{\edge{s_2}{t_2} \in E(S_2, T_2)}  L_{\edged{s_2}{t_2}} \right) \cdot \left( \uprod{u \in \rho(G)}  \uprod{\edge{v}{u} \in E} (2d)^{\frac{1}{|\inc(u)|}}   \right) \numberthis \label{eq:dsmall-holder-2} \\
& = 2^{|C(G)| + |S_1(G)| + |S_2(G)|} \cdot \left( \uprod{\edge{v_i}{u_i}: C_i \in C(G)} L_{\edge{v_i}{u_i}} \cdot \left( \uprod{u \in V(C_i) \setminus \{v_i, u_i\}} 2d \right)  \right) \\
& \quad \cdot \left( \uprod{s_1 \in S_1(G)} \usum{d_{\edge{s_1}{t_1}} \le L: \edge{s_1}{t_1} \in E(S_1, T_1)} L_{\edge{s_1}{t_1}} \right) \cdot  \left( \uprod{t_1 \in T_1(G)}  2d \right) \\
& \quad \cdot  \quad  \left( \usum{d_{\edge{s_2}{t_2}} \le L: \edge{s_2}{t_2} \in E(S_2, T_2)}  \uprod{\edge{s_2}{t_2} \in E(S_2, T_2)}  L_{\edged{s_2}{t_2}} \right) \cdot \left( \uprod{u \in \rho(G)}  2d \right)  \\
& \le 2^{|V| + |C(G)| + |S_1(G)| + |S_2(G)|} \left( \uprod{C_i \in C(G)} L d^{|C_i| - 2} \right) \cdot  
\left( L^{|S_1(G)|} \cdot d^{|T_1(G)|- |T_{S_1(G)}|} \right) \\
& \quad \cdot  \left( \usum{d_{\edge{s_2}{t_2}} \le L: \edge{s_2}{t_2} \in E(S_2, T_2)}  \uprod{\edge{s_2}{t_2} \in E(S_2, T_2)}  L_{\edged{s_2}{t_2}} \right) \cdot d^{|\rho(G)|} \\
& \le 2^{|V| + |C(G)| + |S_1(G)| + |S_2(G)|} \left( \uprod{C_i \in C(G)} L d^{|C_i| - 2} \right) \cdot  \left( L^{|S_1(G)|} d^{|T_1(G)| - |T_{S_1}(G)|} \right) \\
& \quad \cdot \left(\uprod{\edge{s_2}{t_2} \in E(S_2, T_2)}  \usum{d_{\edge{s_2}{t_2}} \le L} L_{\edged{s_2}{t_2}}  \right) d^{\rho(G)} \\
& \le 2^{2(|V| + |C(G)| + |S_1(G)| + |S_2(G)|)} \left( \uprod{C_i \in C(G)} L d^{|C_i| - 2} \right) \cdot  L^{|S_1(G)|} d^{ |T_1(G)| - |S_1(G)|} \cdot L^{|S_2(G)|} \cdot d^{\rho(G)} \numberthis \label{eq:dsmall-holder-4}. 
\end{align*}
Here~\eqref{eq:dsmall-holder-1} follows by a direct application of H\"{o}lder's inequality using the fact that $|\inc(u)| \ge 1$ for every non-source vertex $u \in V$. As a result, we can push the sums in and then,~\eqref{eq:dsmall-holder-2} follows from the definition of $d_{\edge{v}{u}}$s as powers of two and the sum from $1$ to $d$ is at most $2d$. Finally,~\eqref{eq:dsmall-holder-4} follows by applying $\usum{d_{\edge{s_2}{t_2}} \le L} L_{\edged{s_2}{t_2}} = L$ and noting that $|T_{1}(G)|- |T_{S_1(G)}| = |T_{1}(G)|- |S_{1}(G)|$. This completes the proof.

\subsubsection{Mapping Edge Cover of our Decomposition to Optimal Fractional Edge Cover} \label{sec:agm-edge-cover}
We first restate the fractional covering linear program for (the undirected version) of $E(S, T)$ and with a slight abuse of notation assume $E(S, T)$ to be containing the undirected set.
\begin{align*}
\min \quad  \usum{e \in E(S, T)} x_{e} \tag{($\LP'$)}\\
\usum{e \ni v} x_{e} \ge 1 \quad \forall v \in V \\
x_{e} \ge 0 \quad \forall e \in E.
\end{align*}
We state a well-known result based on the LP above (where we use the fact that $E(S,T)$ is bipartite):
\begin{thm}[Implicit in~\cite{NPRR}]
	There exists an optimal solution for $\LP'$ on $(S, T)$ that can decomposed into a union of disjoint stars, where $x_{e} = \{0, 1\}$ for every $e \in E(S, T)$. Further, the set of edges $e$ with $x_e=1$ forms the disjoint set of stars.
\end{thm}
We then pick $E(S_1, T_2)$ and $E(S_2, T_2)$ from the stars above. The following corollary holds.
\begin{cor}
	For each edge $\edge{v}{u} \in E(S_1, T_2), E(S_2, T_2)$, we have $x_{\edge{v}{u}} = 1$.
\end{cor}
This implies $E(S_1, T_2)$ and $E(S_2, T_2)$ are minimum integral edge coverings for $S, T$.

To complete this section, we argue that the min star cover size of $(S, T)$ is the same as the min bound achieved by our upper bound in Theorem~\ref{thm:dsmall}. For this, we start by restating the portion of upper bound for the $(S, T)$ part.
\begin{align*}
& \left(  \left(\frac{L}{d}\right)^{|S_1(G)|} \cdot d^{|T_1(G)|} \right) \cdot L^{|S_2(G)|} \\
& = \inparen{\frac{L}{d}}^{\abs{S_1(G)}+\abs{S_2(G)}}\cdot d^{\abs{T_1(G)}+\abs{S_2(G)}} \\
& = \inparen{\frac{L}{d}}^{\abs{S(G)}} \cdot d^{\abs{T_1(G)}+\abs{S_2(G)}}. 
\end{align*}
In our upper bound, $\abs{S(G)}$ is fixed so miniming the above part is the same as minimizing $\abs{T_1(G)}+\abs{s_2(G)}$, which is exactly minimizing the size of star/integral edge cover, as desired.

\subsection{Missing Details in Section~\ref{sec:dhigh}} \label{app:dhigh}
In this section, our goal is to prove~\ref{eq:dhigh-ub}. 
\subsubsection{Notation and Existing Results}
We restate $\LP^{(+)}$ for this scenario, where $x_{\edgeud{v}{u}}$ corresponds to the $\ell_1$-norm bounds\footnote{Note here that we use $\edgeud{v}{u}$ here instead of the standard $x_{\edge{v}{u}}$. We would like to note here that this is fine since for $\ell_{1}$-norm bounds, the direction does not matter and we (implicitly) assume that all tuples in the relation $R_{\edgeud{v}{u}}$ are directed from $v$ to $u$.} and $z_{\edge{v}{u}}$ correspond to the $\ell_{\infty}$-norm bounds.
\begin{align*}
& \min \left( \usum{(v, u) \in E} x_{\edgeud{v}{u}} \log(L) +  \usum{\edge{v}{u} \in E} z_{\edge{v}{u}} \log(d) \right) \tag{$\LP^{+}$} \\
& \text{ s.t. } \left( \usum{e = (v, u) \ni u} x_{v, u}  \right) + \left( \usum{\edge{v}{u} \in E} z_{\edge{v}{u}} \right) \ge 1 \quad \forall u \in V \numberthis \label{eq:dhigh-covering} \\
& x_{\edgeud{v}{u}}, z_{\edge{v}{u}} \ge 0 \quad \forall \edge{v}{u} \in E.
\end{align*}
We will be using a result similar to Theorem~\ref{thm:opt-decomp} on $\LP^{(+)}$ (the proof is in Appendix~\ref{sec:dhigh-struct}).
\begin{corollary} \label{cor:dhigh-struct}
For any $G = (V, E)$, there exists an optimal basic feasible solution $(\vx^*, \vz^*) = (x^*_{\edgeud{v}{u}}, z^*_{\edge{v}{u}})_{\edgeud{v}{u}, \edge{v}{u} \in E}$ to $\LP^{(+)}$ on $G$ that can be decomposed into a disjoint union of $t$ (some $t > 0$) connected components (in the undirected sense) $G_i = (V_i, E_i)$ with 
\begin{align*}
|V_i| - 1 \le |Q(E(G_i))| \le |V_i|, \text{ where } 
\end{align*}
$Q(E(G_i))$ is the set of non-zero values $x_{\edgeud{v}{u}}$ and $z_{\edge{v}{u}}$ for every $\edgeud{v}{u}$ and $\edge{v}{u}$ in $E(G_i)$. Further, $$(\mathbf{x}^*_{i}, \mathbf{z}^*_{i}) = \left( x_{\edgeud{v}{u}}, z_{\edge{v}{u}} \right)_{\edgeud{v}{u}, \edge{v}{u} \in E(G_i)}$$ is an optimal basic feasible solution for $\LP^{(+)}$ on $G_i$ for every $i \in [t]$ and $\cup_{i=1}^{t}V(G_i) = V$ and $V(G_i) \cap V(G_j) = \emptyset$ \space $\forall i, j \in [t], i \neq j$. The following is true -- $\TJ = \times_{i \in [t]} J_{G_i}^{(I)}$.
\end{corollary}
For our arguments, in addition to the above, we need an extremal property (Property~\ref{assump:dhigh}) as well, where consider a specific subclass of optimal basic feasible solutions to $\LP^{(+)}$ on $G$. 

\subsubsection{Proof of~\eqref{eq:dhigh-ub}} \label{sec:dhigh-ub}
We make the following claim.
\begin{lemm} \label{lemma:dhigh-dag}
Each $G_i: i \in [t]$ is a DAG.
\end{lemm}
If the above result holds, note that we can invoke Theorem~\ref{thm:gen-main} on each $G_i$ to prove~\eqref{eq:dhigh-ub}. The rest of this section will focus on proving Lemma~\ref{lemma:dhigh-dag} and we assume $G_i$ is cyclic and since $|V(G_i)| - 1 \le |E(G_i) \le |V(G_i)|$, we have that it has exactly one cycle, which we denote by $C$.
\subsubsection{Proof of Lemma~\ref{lemma:dhigh-dag}} \label{app:dhigh-dag}
We make the following claim on $G_i$ on $(\vx^*, \vz^*)$.
\begin{claim} \label{claim:dhigh-path}
For some $k \ge 2$, there exists a path (in the undirected sense) $P = \{u, \{v_i\}_{i \in [k - 1]}, w\}$ such that (for all $\in [k - 1]$ and $u, v_1 \in V(C)$)
\begin{align*}
x^*_{\edgeud{v_i}{v_{i + 1}}} > 0, z^*_{\edge{u}{v_1}} > 0, z^*_{\edge{v_k}{w}} > 0 .
\end{align*}
Further, for edges $\edgeud{u}{v_1}$ and $\edgeud{w}{v_k}$, we have $x_{\edgeud{u}{v_1}} = x_{\edgeud{w}{v_k}} = 0$ and for every $\edgeud{v_i}{v_{i + 1}}: i \in [t]$, we have $z_{\edge{v_i}{v_{i + 1}}} = 0$. 
\end{claim} 
Assuming the above claim is true, we prove Lemma~\ref{lemma:dhigh-dag} by a contradiction. In particular, we will construct an alternative optimal solution to $\LP^{(+)}$ on $G_i$ with one of the following properties:
\begin{property} \label{prop:et-odd}
When $|E(P)|$ is odd (i.e., $k$ is even), our solution has a smaller objective value than the optimal solution $(\vx^*_i, \vz^*_i)$ we started with, contradicting the optimality of $(\vx^*_i, \vz^*_i)$.
\end{property} 
\begin{property} \label{prop:et-even}
When $|E(P)|$ is even (i.e., $k$ is odd), there exists an optimal basic feasible solution $(\vx',\vz')$ and at least one edge $\edgeud{v}{u} \in E$ such that $x'_{\edgeud{v}{u}} = z'_{\edgeud{v}{u}} = 0$, contradicting Assumption~\ref{assump:dhigh}.
\end{property}
As discussed above, in both cases, we would end up in a contradiction, proving Lemma~\ref{lemma:dhigh-dag}, as required.

In the remaining section, we  will sketch the proofs of Properties~\ref{prop:et-odd} and~\ref{prop:et-even}. Let $\epsilon$ be defined as the minimum of the following two expressions:
\begin{align*}
& \min \left( z^*_{\edge{u}{v_1}}, 1 - z_{\edge{u}{v_1}}, z^*_{\edge{w}{v_k}}, 1 - z^*_{\edge{w}{v_k}} \right) \\
& \umin{\edgeud{v_i}{v_{i + 1}} \in E(P), i \in [k - 1]} \left(x^*_{\edgeud{v}{u}}, 1 - x^*_{\edgeud{v}{u}} \right)
\end{align*}
The following results are true:
\begin{align*}
& z^*_{\edge{u}{v_1}} \pm \epsilon \in [0, 1], z^*_{\edge{w}{v_k}} \pm \epsilon \in [0, 1] \numberthis \label{eq:z-e-1} \\
& x^*_{\edgeud{v_i}{v_{i + 1}}} \pm \epsilon \in [0, 1] \quad \forall i \in [k - 1] \numberthis \label{eq:x-e-1}
\end{align*}
and at least one of the following is true:
\begin{align*}
z^*_{\edge{u}{v_1}} = \epsilon \text{ (or) } z^*_{\edge{u}{v_1}} = 1 - \epsilon \\
x^*_{\edgeud{v_i}{v_{i + 1}}} = \epsilon \text{ (or) } x^*_{\edgeud{v_i}{v_{i + 1}}}  = 1 - \epsilon \text{ for some } i \in [k - 1] \\
z^*_{\edge{w}{v_k}} = \epsilon \text{ (or) } z^*_{\edge{w}{v}} = 1 - \epsilon.
\end{align*}
Here~\eqref{eq:z-e-1},~\eqref{eq:x-e-1} and the above follow by the definition of $(\vx^*, \vz^*)$ 
and $\epsilon$.

We construct two alternative solutions, starting with some notation for all $i \in [k - 1]$ and $\alpha\in [-1,1]$:
\begin{align*}
& x^{\alpha}_{\edgeud{v_i}{v_{i + 1}}} =  x^*_{\edgeud{v_i}{v_{i + 1}}} +\alpha 
\end{align*}
and
\begin{align*}
& z^{\alpha}_{\edge{u}{v_1}} = z^*_{\edge{u}{v_1}} + \alpha 
,& z^{\alpha}_{\edge{w}{v_k}} = z^{*}_{\edge{w}{v_k}}  + \alpha 
\end{align*}
The first solution is where we decrease $z^{*}_{\edge{u}{v_1}}$ by $\epsilon$ and subsequently, increase $x_{\edgeud{v_1}{v_2}}$ by $\epsilon$ and so on. In particular, we would have a solution of the form
\begin{align*}
z^{- \epsilon}_{\edge{u}{v_1}},  \left( x^{(-1)^{i - 1} \epsilon}_{\edgeud{v_i}{v_{i + 1}}} \right)_{i \in [k - 1]}, z^{(-1)^{k - 1} \epsilon}_{\edge{w}{v_k}} \numberthis \label{eq:dhigh-sol1}.
\end{align*}
The second solution is where we increase $z^{*}_{\edge{u}{v_1}}$ by $\epsilon$ and we have a solution of the form.
\begin{align*}
z^{+ \epsilon}_{\edge{u}{v_1}},  \left( x^{(-1)^{i} \epsilon}_{\edge{v_i}{v_{i + 1}}} \right)_{i \in [k - 1]},
z^{(-1)^{k} \epsilon}_{\edge{w}{v_k}} \numberthis \label{eq:dhigh-sol2}.
\end{align*}
We claim the following based on~\eqref{eq:dhigh-sol1} and~\eqref{eq:dhigh-sol2}. 
\begin{claim} \label{claim:dhigh-bf}
The constraint~\eqref{eq:dhigh-covering} is tight on each vertex in $P = (u, v_1, \dots, v_k, w)$ in both~\eqref{eq:dhigh-sol1} and~\eqref{eq:dhigh-sol2}. 
\end{claim}
Assuming the above claim is true (the proof is in Appendix~\ref{sec:dhigh-bf}), we sketch the proofs of Properties~\ref{prop:et-odd} and~\ref{prop:et-even}.
It turns out that when $|E(P)|$ is odd,~\eqref{eq:dhigh-sol1} has a strictly smaller objective value than $(\vx^*, \vz^*)$, resulting in a contradiction. When $|E(P)|$ is even, there exists at least one edge with value $0$ in one of~\eqref{eq:dhigh-sol1} and~\eqref{eq:dhigh-sol2}, which contradicts an extremal property of $(\vx^*, \vz^*)$. 

\subsubsection{Proof of Properties~\ref{prop:et-odd} and~\ref{prop:et-even}}
In this section, our goal is to prove Properties~\ref{prop:et-odd} and~\ref{prop:et-even}. We recall some notation $P = \{u, v_1, \dots, v_k, w\}$ for all $i \in [k - 1]$:
\begin{align*}
& x^{- \epsilon}_{\edgeud{v_i}{v_{i + 1}}} =  x^*_{\edgeud{v_i}{v_{i + 1}}} - \epsilon,  x^{+ \epsilon}_{\edgeud{v_i}{v_{i + 1}}} =  x^*_{\edgeud{v_i}{v_{i + 1}}} + \epsilon
\end{align*}
and
\begin{align*}
& z^{- \epsilon}_{\edge{u}{v_1}} = z^*_{\edge{u}{v_1}} - \epsilon , z^{+ \epsilon}_{\edge{u}{v_1}} = z^{*}_{\edge{u}{v_1}} + \epsilon \\
& z^{- \epsilon}_{\edge{w}{v_k}} = z^{*}_{\edge{w}{v_k}}  - \epsilon, z^{+ \epsilon}_{\edge{w}{v_k}} = z^{*}_{\edge{w}{v_k}} + \epsilon,
\end{align*}
where $(\vx^*, \vz^*)$ is an optimal basic feasible solution to $\LP^{(+)}$ on $G$. We restate~\eqref{eq:dhigh-sol1} below.
\begin{align*}
z^{- \epsilon}_{\edge{u}{v_1}},  \left( x^{(-1)^{i - 1} \epsilon}_{\edgeud{v_i}{v_{i + 1}}} \right)_{i \in [k - 1]}, z^{(-1)^{k} \epsilon}_{\edge{w}{v_k}} \numberthis \label{eq:dhigh-sol1-app1}.
\end{align*}
Note that the remaining values are the same as $(\vx^*, \vz^*)$.

Next, we compute the ratio of the new objective value and $(\vx^*, \vz^*)$. We have
\begin{align*}
& \frac{ \left(\uprod{\edgeud{f_1}{f_2} \in E \setminus E(P)} L^{x^*_{\edgeud{f_1}{f_2}}}  d^{z^*_{\edge{f_1}{f_2}}} \right)  \cdot d^{z^{ - \epsilon}_{\edge{u}{v_1}} + z^{ (-1)^{k -1} \epsilon}_{\edge{w}{v_k}}} \uprod{i \in [k - 1]} L^{x^{(-1)^{i - 1} \epsilon}_{\edge{v_i}{v_{i + 1}}} }} { \left(\uprod{\edgeud{f_1}{f_2} \in E \setminus E(P)} L^{x^*_{\edgeud{f_1}{f_2}}}  \cdot d^{z^*_{\edge{f_1}{f_2}}} \right) d^{z^*_{\edge{u}{v_1}} + z^*_{\edge{w}{v_k}}} 
	\uprod{i \in [k - 1]} L^{x^*_{\edgeud{v_i}{v_{i + 1}}}}
}  \\
& = \frac{d^{z^{ - \epsilon}_{\edge{u}{v_1}} + z^{ (-1)^{k - 1} \epsilon}_{\edge{w}{v_k}}} \uprod{i \in [k - 1]} L^{x^{(-1)^{i - 1} \epsilon}_{\edge{v_i}{v_{i + 1}}} }  } {  d^{z^*_{\edge{u}{v_1}} + z^*_{\edge{w}{v_k}}} \uprod{i \in [k - 1]} L^{x^*_{\edgeud{v_i}{v_{i + 1}}}} } \\
& = \frac{d^{z^{ - \epsilon}_{\edge{u}{v_1}} + z^{ - \epsilon}_{\edge{w}{v_k}}} \uprod{i \in [k - 1]} L^{x^{(-1)^{i - 1} \epsilon}_{\edge{v_i}{v_{i + 1}}} }  } {  d^{z^*_{\edge{u}{v_1}} + z^*_{\edge{w}{v_k}}} \uprod{i \in [k - 1]} L^{x^*_{\edgeud{v_i}{v_{i + 1}}}} } \quad \text{ when } k \text { is even} \\
& = \frac{d^{z^*_{\edge{u}{v_1}}  - \epsilon} d^ {z^*_{\edge{w}{v_k}} - \epsilon} \uprod{i \in [k - 1]} L^{x^{(-1)^{i - 1} \epsilon}_{\edge{v_i}{v_{i + 1}}} } }{ d^{z^*_{\edge{u}{v_1}} + z^*_{\edge{w}{v_k}}} \uprod{i \in [k - 1]} L^{x^*_{\edgeud{v_i}{v_{i + 1}}}} } \\
& =\left( \frac{L}{d^2} \right)^{\epsilon} \\
& < 1
\end{align*}
where the final inequality follows from $d^2 > L$. This proves Property~\ref{prop:et-odd} since the bound obtained by~\eqref{eq:dhigh-sol1} is strictly better than that of $(\vx^*, \vz^*)$, resulting in a contradiction.

We now consider the case when $k$ is odd, where we have
\begin{align*}
\frac{d^{z^{ - \epsilon}_{\edge{u}{v_1}} + z^{ (-1)^{k - 1} \epsilon}_{\edge{w}{v_k}}} \uprod{i \in [k - 1]} L^{x^{(-1)^{i - 1} \epsilon}_{\edge{v_i}{v_{i + 1}}} }  } {  d^{z^*_{\edge{u}{v_1}} + z^*_{\edge{w}{v_k}}} \uprod{i \in [k - 1]} L^{x^*_{\edgeud{v_i}{v_{i + 1}}}} }
& = \frac{d^{z^{ - \epsilon}_{\edge{u}{v_1}} + z^{+ \epsilon}_{\edge{w}{v_k}}  \uprod{i \in [k - 1]} L^{x^{(-1)^{i - 1} \epsilon}_{\edge{v_i}{v_{i + 1}}} } }} {d^{z^*_{\edge{u}{v_1}} + z^*_{\edge{w}{v_k}}} \uprod{i \in [k - 1]} L^{x^*_{\edgeud{v_i}{v_{i + 1}}}} } \\
& = 1 \numberthis \label{eq:dhigh-zero1}.
\end{align*}
Note that the objective values are the same in this case. We do a similar computation for~\eqref{eq:dhigh-sol2} as well, this time starting with the case when $k$ is odd and first restate~\eqref{eq:dhigh-sol2}.
\begin{align*}
z^{+ \epsilon}_{\edge{u}{v_1}},  \left( x^{(-1)^{i} \epsilon}_{\edge{v_i}{v_{i + 1}}} \right)_{i \in [k - 1]},
z^{(-1)^{k} \epsilon}_{\edge{w}{v_k}} \numberthis \label{eq:dhigh-sol2-app2}.
\end{align*}
\begin{align*}
& \frac{ \left( \uprod{\edgeud{f_1}{f_2} \in E \setminus E(P)} L^{x^*_{\edgeud{f_1}{f_2}}} d^{z^*_{\edge{f_1}{f_2}}} \right)  \cdot d^{z^{ + \epsilon}_{\edge{u}{v_1}} + z^{ (-1)^{k} \epsilon}_{\edge{w}{v_k}}} \uprod{i \in [k - 1]} L^{x^{(-1)^{i} \epsilon}_{\edge{v_i}{v_{i + 1}}} }} { \left( \uprod{\edgeud{f_1}{f_2} \in E \setminus E(P)} L^{x^*_{\edgeud{f_1}{f_2}}} d^{z^*_{\edge{f_1}{f_2}}} \right) d^{z^*_{\edge{u}{v_1}} + z^*_{\edge{w}{v_k}}} 
	\uprod{i \in [k - 1]} L^{x^*_{\edgeud{v_i}{v_{i + 1}}}}
}  \\
& = \frac{d^{z^{ + \epsilon}_{\edge{u}{v_1}} + z^{ (-1)^{k} \epsilon}_{\edge{w}{v_k}}} \uprod{i \in [k - 1]} L^{x^{(-1)^{i} \epsilon}_{\edge{v_i}{v_{i + 1}}} }} {d^{z^*_{\edge{u}{v_1}} + z^*_{\edge{w}{v_k}}}  \uprod{i \in [k - 1]} L^{x^*_{\edgeud{v_i}{v_{i + 1}}}} } \\
& = \frac{d^{z^*_{\edge{u}{v_1}} + \epsilon}  d^{z^*_{\edge{w}{v_k}} - \epsilon}  \uprod{i \in [k - 1]} L^{x^{(-1)^{i} \epsilon}_{\edge{v_i}{v_{i + 1}}} } } {d^{z^*_{\edge{u}{v_1}} + z^*_{\edge{w}{v_k}}}  \uprod{i \in [k - 1]} L^{x^*_{\edgeud{v_i}{v_{i + 1}}}} } \\
& = 1.
\end{align*}
We are now ready to argue Property~\ref{prop:et-even} based on~\eqref{eq:dhigh-sol1-app1} and~\eqref{eq:dhigh-sol2-app2}. Recall the definition of $\epsilon$:
\begin{align*}
\epsilon = \min \left \{ z^*_{\edge{u}{v_1}}, 1 - z_{\edge{u}{v_1}}, z^*_{\edge{w}{v_k}}, 1 - z^*_{\edge{w}{v_k}} , \umin{\edgeud{v_i}{v_{i + 1}} \in E(P), i \in [k - 1]} \left \{ x^*_{\edgeud{v}{u}}, 1 - x^*_{\edgeud{v}{u}} \right \} \right \}.
\end{align*}
WLOG let $x^*_{\edgeud{v_1}{v_2}}$ be the value achieving the $\epsilon$ minimum. Then, one of the following conditions is true:
\begin{align*}
x^*_{\edgeud{v_1}{v_2}} = \epsilon \\
x^*_{\edgeud{v_1}{v_2}} = 1 - \epsilon.
\end{align*}
Assuming the first case, we would have added it by $\epsilon$ in~\eqref{eq:dhigh-sol1-app1} and subtracted it by $\epsilon$ in~\eqref{eq:dhigh-sol2-app2}. The latter would give us $x^{- \epsilon}_{\edgeud{v_1}{v_2}} = 0$, contradicting Property~\ref{assump:dhigh}. Now, consider the second case, where we would have added it by $\epsilon$ in~\eqref{eq:dhigh-sol2-app2} and subtracted it by $\epsilon$ in~\eqref{eq:dhigh-sol1-app1}. In the former, we would have $x^{ + \epsilon}_{\edgeud{v_1}{v_2}} = 1$, which in turn implies $x^{- \epsilon}_{\edgeud{v_2}{v_3}} = 0$ and $z^{- \epsilon}_{\edgeud{u}{v_1}} = 0$. This contradicts Property~\ref{assump:dhigh} as well and proves Property~\ref{prop:et-even}, as required.

For the sake of completenes, we consider the case when $k$ is even as well. However, it turns out that the objective value in this case is worse than that of $(\vx^*, \vz^*)$ and as a result, we can ignore it.
\begin{align*}
\frac{d^{z^{ + \epsilon}_{\edge{u}{v_1}} + z^{ (-1)^{k} \epsilon}_{\edge{w}{v_k}}} \uprod{i \in [k - 1]} L^{x^{(-1)^{i} \epsilon}_{\edge{v_i}{v_{i + 1}}} }} {d^{z^*_{\edge{u}{v_1}} + z^*_{\edge{w}{v_k}}}  \uprod{i \in [k - 1]} L^{x^*_{\edgeud{v_i}{v_{i + 1}}}} } & = \frac{d^{z^*_{\edge{u}{v_1}} + \epsilon}  d^{z^*_{\edge{w}{v_k}} + \epsilon}  \uprod{i \in [k - 1]} L^{x^{(-1)^{i} \epsilon}_{\edge{v_i}{v_{i + 1}}} } } {d^{z^*_{\edge{u}{v_1}} + z^*_{\edge{w}{v_k}}}  \uprod{i \in [k - 1]} L^{x^*_{\edgeud{v_i}{v_{i + 1}}}} } \quad \text{ when } k \text { is even}  \\
& = \left( \frac{d^2}{L} \right)^{\epsilon} \\
& > 1.
\end{align*}
This completes the proof.

\subsubsection{Proof of Claim~\ref{claim:dhigh-path}} \label{sec:dhigh-path}
We prove Claim~\ref{claim:dhigh-path} here. Our goal here is to argue that for some $k \ge 2$, there exists a path $P = \{u, v_1, \dots, v_{k} , w\}$ such that $u, v_1 \in V(C), z_{\edge{u}{v_1}} > 0$, $z^*_{\edge{v_k}{w}} > 0$ and $x^*_{\edgeud{v_i}{v_{i + 1}}} > 0$ for every $i \in [k - 1]$. 

We first claim that there always exists a path $P' = \{u, v_1, v_2\}$ in $G$ such that $u, v_1 \in V(C), z^*_{\edge{u}{v_1}} > 0, x^*_{\edgeud{v_1}{v_2}} > 0$. Note that $P'$ is a subpath of $P$.

Assuming the above claim is true, we construct $P$ from $P'$ as follows. Consider the vertex $v_2$. If it has an incoming edge $\edge{w}{v_2}$ with $z^*_{\edge{w}{v_2}} > 0$ and we have constructed a $P = \{P', w\}$, as required. If not, there are two possibilities -- $(1)$ $v_2$ has only outgoing edges $\edge{v_2}{y} \in E$ with $z^*_{\edge{v_2}{y}} > 0$. $(2)$ $v_2$ at least one more incident edge $\edgeud{v_2}{v_2} \in E$ with $x^*_{\edgeud{v_2}{v_3}} > 0$. 

In case of $(1)$, note that we have $x^*_{\edge{v_1}{v_2}} = 1$ (by basic feasibility of $(\vx^*, \vz^*)$) and this would violate the tightness of the constraint on $v_1$ i.e., we have
\begin{align*}
z^*_{\edge{u}{v_1}} + x^*_{\edge{v_1}{v_2}} > 1,
\end{align*}
contradicting our assumption that $(\vx^*, \vz^*)$ is basic feasible. If we fall back to $(2)$, then we update $P = P' \cup \{v_3\}$ and continue this process. There are two stopping conditions for this procedure -- $(3)$ we identify a $w$ as above and terminate with a valid $P$ and $(4)$ the $v_k$ for some $k > 2$ we hit in our procedure is a leaf (i.e., has degree $1$) in $G$. We show that $(4)$ cannot happen, in particular, if $v_k$ is a leaf, then it is true that $x^*_{\edge{v_{k - 1}}{v_k}} = 1$ (again by basic feasibility of $(\vx^*, \vz^*)$) and this would violate the tightness of the constraint on $v_{k - 1}$ i.e., we have
\begin{align*}
x^*_{\edge{v_{k - 2}}{v_{k - 1}}}  + x^*_{\edge{v_{k - 1}}{v_k}}  > 1,
\end{align*} 
contradicting the basic feasibility of $(\vx^*, \vz^*)$.As a result, we can always construct a valid $P$ as required.

To complete the proof, we reason about the cases where no such path of the kind $P'$ exists. 
\begin{itemize}
	\item {The obvious case is when the graph $G$ is a cycle $C$ and all edges $\edge{v}{u} \in E(G)$ have $z_{\edge{v}{u}} > 0$.} 
	\item {All vertices in the cycle $C$ have only outgoing edges $\edge{u}{w} \in E(G), u \in V(C), w \not \in V(C)$ with $z_{\edge{u}{w}} > 0$.}
	\item {There exists at least one vertex $u \in V(C)$ with at least one incoming edge $\edge{v}{u}$ from outside the cycle i.e., $v \not \in V(C)$.}
\end{itemize}
We argue the first and second cases in one shot and argue the third case separately. For the first two cases,
we construct an alternative basic feasible solution $(\vx', \vz')$ to $G$ as follows. We pick an arbitrary edge $\edgeud{v}{u} \in E(C)$ and set $x'_{\edgeud{v}{u}} = 1$ and $z'_{\edge{v}{u}} = z'_{\edge{w}{v}} = 0$. The remaining values in $(\vx', \vz')$ are the same as $(\vx^*, \vz^*)$ Computing the ratio of objective values given by $(\vx', \vz')$ and $(\vx^*, \vz^*)$ for edges in $C$, we have
\begin{align*}
\frac{  L d^{|V(C)| - 2}}{ d^{|V(C)|}} & = \frac{L}{d^2} \\
&  < 1,
\end{align*}
where the inequality follows from $d^2 > L$. Note that the remaining edges in $G$ have the same value in $\vx'$ and $\vx^*$ (by definition of $\vx'$). We now argue that $\vx'$ is basic feasible and to do so, we only need to focus on vertices $u$ and $v$ whose tightness could have been affected by our construction. Since we have $x'_{\edgeud{v}{u}}  = 1$ and $v$, $u$ can have only outgoing edges outside of $E(C)$, the constraints are still tight on $v$ and $u$ respectively. Thus, $(\vx', \vz')$ is basic feasible and has a smaller objective value than $(\vx^*, \vz^*)$. As a result, we can replace $\edge{v}{u}$ with $\edgeud{v}{u}$ and drop the edge $\edge{w}{v}$ to make $G$ a DAG.

Finally, we argue the third case. As in the argument above, we construct an alternative optimal solution $(\vx', \vz')$ as follows. We focus on $\edge{v}{u} \in E$ such that $v \not \in V(C)$ and $u \in V(C)$ and consider $\edge{w}{u} \in E(C)$. We set $z'_{\edge{v}{u}} = z^*_{\edge{v}{u}} + z^*_{\edge{w}{u}}$ and $z'_{\edge{w}{u}} = 0$. The remaining values in $(\vx', \vz')$ are the same as $(\vx^*, \vz^*)$ Computing the ratio of objective values given by $(\vx', \vz')$ and $(\vx^*, \vz^*)$ for $G$, we have
\begin{align*}
\frac{d^{z'_{\edge{v}{u}}}} d^{z'_{\edge{w}{u}}}{d^{z^*_{\edge{v}{u}}} d^{z^*_{\edge{w}{u}}} } & = \frac{d^{z'_{\edge{v}{u}}}} {d^{z^*_{\edge{v}{u}} + z^*_{\edge{w}{u}}}} \\
& = 1,
\end{align*}
where the equation follows from definitions of $z'_{\edge{v}{u}}$ and $z'_{\edge{w}{u}}$. Note that the remaining edges in $G$ have the same value in $\vx'$ and $\vx^*$ (by definition of $\vx'$). We now argue that $\vx'$ is basic feasible and to do so, we only need to focus on vertices $u$ whose tightness could have been affected by our construction. Since $z'_{\edge{v}{u}} = z^*_{\edge{v}{u}} + z^*_{\edge{w}{u}}$, the constraint is still tight on $u$. Thus, we have constructed a new solution $(\vx', \vz')$ that is basic feasible, has the same objective value as $(\vx^*, \vz^*)$ and contradicts Property~\ref{assump:dhigh}. As a result, a $G$ of this kind cannot exist.

This completes the proof we since in all the above cases, we could always construct an alternative solution,
where $G$ is DAG.

\subsubsection{Proof of Claim~\ref{claim:dhigh-bf}} \label{sec:dhigh-bf}
We recall some notation $P = \{u, v_1, \dots, v_k, w\}$ for all $i \in [k - 1]$:
\begin{align*}
& x^{- \epsilon}_{\edgeud{v_i}{v_{i + 1}}} =  x^*_{\edgeud{v_i}{v_{i + 1}}} - \epsilon,  x^{+ \epsilon}_{\edgeud{v_i}{v_{i + 1}}} =  x^*_{\edgeud{v_i}{v_{i + 1}}} + \epsilon
\end{align*}
and
\begin{align*}
& z^{- \epsilon}_{\edge{u}{v_1}} = z^*_{\edge{u}{v_1}} - \epsilon , z^{+ \epsilon}_{\edge{u}{v_1}} = z^{*}_{\edge{u}{v_1}} + \epsilon \\
& z^{- \epsilon}_{\edge{w}{v_k}} = z^{*}_{\edge{w}{v_k}}  - \epsilon, z^{+ \epsilon}_{\edge{w}{v_k}} = z^{*}_{\edge{w}{v_k}} + \epsilon.
\end{align*}
We recall the constraint~\eqref{eq:dhigh-covering} for every $u \in V$ as well for an optimal basic feasible solution $(\vx^*, \vz^*)$:
\begin{align*}
\left( \usum{e = (v, u) \ni u} x^*_{\edgeud{v}{u}}  \right) + \left( \usum{\edge{v}{u} \in E} z^*_{\edge{v}{u}} \right) = 1 \quad \forall u \in V \numberthis \label{eq:dhigh-covering-app},
\end{align*}
where the tightness of the constraint follows from the basic feasibility of $(\vx^*, \vz^*)$.

We prove this claim first for~\eqref{eq:dhigh-sol1}, followed by~\eqref{eq:dhigh-sol2} and start by restating~\eqref{eq:dhigh-sol1}.
\begin{align*}
z^{- \epsilon}_{\edge{u}{v_1}},  \left( x^{(-1)^{i - 1} \epsilon}_{\edgeud{v_i}{v_{i + 1}}} \right)_{i \in [k - 1]}, z^{(-1)^{k - 1} \epsilon}_{\edge{w}{v_k}} \numberthis \label{eq:dhigh-sol1-app}.
\end{align*}
Note that the variables $z^{- \epsilon}_{\edge{u}{v_1}}$ and $z^{(-1)^{k - 1}}_{\edge{w}{v_1}}$ do not belong in~\eqref{eq:dhigh-covering} for $u$ and $w$. As a result, the constraint~\eqref{eq:dhigh-covering} is still tight for $u$ and $w$. Next, we consider the set of vertices $(v1, \dots, v_k)$ and for every $v_i:i \in [2, k - 2]$, we have
\begin{align*}
x^{(-1)^{i - 2}\epsilon}_{\edgeud{v_{i - 1}}{v_{i}}} + x^{(-1)^{i - 1} \epsilon}_{\edgeud{v_i}{v_{i + 1}}} & = x^{*}_{\edgeud{v_{i  - 1}}{v_i}} + x^{*}_{\edgeud{v_{i}}{v_{i + 1}}},
\end{align*}
where the equation follows from the fact that $(-1)^{i - 2}$ and $(-1)^{i - 1}$ have different parities for every $i \in [2, k -2]$. Since we didn't modify values of the other edges incident on $v_i$ and~\eqref{eq:dhigh-covering-app} was true for $v_i$ to start with, our modification doesn't affect the tightness. We can do a similar argument as above for $v_1$ and $v_k$ as well, where we have
\begin{align*}
z^{- \epsilon}_{\edge{u}{v_1}} + x^{\epsilon}_{\edgeud{v_1}{v_{2}}} & = z^{*}_{\edge{u}{v_1}} + x^*_{\edgeud{v_1}{v_2}} \\
z^{(-1)^{k - 1} \epsilon}_{\edge{w}{v_k}} + x^{(-1)^{k - 2} \epsilon}_{\edgeud{v_{k - 1}}{v_{k}}} & =
z^{*}_{\edge{w}{v_k}} +  x^{*}_{\edgeud{v_{k - 1}}{v_k}},
\end{align*}
as required.

To complete the proof, we do a similar argument for~\eqref{eq:dhigh-sol2} as well, starting by restating it.
\begin{align*}
z^{+ \epsilon}_{\edge{u}{v_1}},  \left( x^{(-1)^{i} \epsilon}_{\edgeud{v_i}{v_{i + 1}}} \right)_{i \in [k - 1]},
z^{(-1)^{k} \epsilon}_{\edge{w}{v_k}} \numberthis \label{eq:dhigh-sol2-app}.
\end{align*}
Similar to the earlier argument, $z^{- \epsilon}_{\edge{u}{v_1}}$ and $z^{(-1)^{k}}_{\edge{w}{v_1}}$ do not belong in~\eqref{eq:dhigh-covering} for $u$ and $w$. As a result, the constraint~\eqref{eq:dhigh-covering} is still tight for $u$ and $w$. Next, we consider the set of vertices $(v1, \dots, v_k)$ and for every $v_i:i \in [2, k - 2]$, we have
\begin{align*}
x^{(-1)^{i - 1}\epsilon}_{\edgeud{v_{i - 1}}{v_{i}}} + x^{(-1)^{i} \epsilon}_{\edgeud{v_i}{v_{i + 1}}} & = x^{*}_{\edgeud{v_{i  - 1}}{v_i}} + x^{*}_{\edgeud{v_i}{v_{i + 1}}},
\end{align*}
where the equation follows from the fact that $(-1)^{i - 1}$ and $(-1)^{i}$ have different parities for every $i \in [2, k -2]$.  Since we didn't modify values of the other edges incident on $v_i$ and~\eqref{eq:dhigh-covering-app} was true for $v_i$ to start with, our modification doesn't affect the tightness. We can do a similar argument as above for $v_1$ and $v_k$ as well, where we have
\begin{align*}
z^{+ \epsilon}_{\edge{u}{v_1}} + x^{ -\epsilon}_{\edgeud{v_1}{v_2}} & = z^{*}_{\edge{u}{v_1}} + x^*_{\edgeud{v_1}{v_2}} \\
z^{(-1)^{k} \epsilon}_{\edge{w}{v_k}} + x^{(-1)^{k - 1} \epsilon}_{\edgeud{v_{k - 1}}{v_{k}}} & =
z^{*}_{\edge{w}{v_k}} +  x^{*}_{\edgeud{v_{k - 1}}{v_k}},
\end{align*}
as required. This completes the proof.

%% file: appendix_conclusion.tex
\section{Extensions to Hypergraphs and other Open Questions}
\subsection{Extending our results to (acyclic) hypergraphs} \label{sec:hyper}
In this section, we outline how your techniques (which have been so far only used for simple graphs) can be extended to work with hypergraphs. The aim of this section is not to present the most general result we can prove for hypergraphs but rather to show that the bottleneck in extending our arguments to the hypergraph case is {\em not} the notion of a degree configuration. Thus, to keep the notation simple(r), we will focus on the case of $\ell_1$  bounds only and we outline how we can recover the AGM bound. We also note that by our we have a `proof from the book' for the AGM bound and while the proof below can be considered `new' we do not see potential benefit of the argument {\em other} than the fact that it shows that our degree configuration based bounds can be extended to hypergraphs.  
Finally, we only consider the problem of computing an upper bound on the size of the join output (the arguments for matching lower bound and the corresponding worst-case optimal algorithms follow along the usual lines but again for less clutter we omit those here).

Before we dive into our argument, we first define a generalization of notion of degrees and degree configuration for hypergraphs (the definitions below are also valid for the case of $\ell_\infty$ bounds though we'll not consider the latter here).

We now consider the case when $G=(\cV,\cE)$ is a {\em hypergraph} and start by recalling the notion of 
{\em simple} degree constraints from~\cite{panda-pods20}, where we have for any 
subsets $X, Y \subseteq \cV$ of variables, a degree constraint $Y | X$ is considered simple if $|X| \le 1$. For our definition of degree constraint, we place an additional restriction that $X \cup Y = \cE$ as well. This would make our degree constraints guarded by the relation $R_{X \cup Y}$.
\begin{defn} [Simple and Guarded degree constraints for hypergraphs]
	For every $F\in \cE$ and $u\in F$, we use $L_{\edgeinfty{u}{F\setminus\inset{u}}}$ to denote the {\em degree constraint of $u$ in $F$}. More specifically we first define the {\em degree of $a\in\Dom(u)$} as
	\[d_{\edge{u}{F\setminus\inset{u}}}[a]=\abs{\inset{\vt|(a,\vt)\in R_F}}.\]
	Then the {\em degree constraint} $L_{\edgeinfty{u}{F\setminus\inset{u}}}$ implies that for every $a\in\Dom(u)$:
	\[d_{\edge{u}{F\setminus\inset{u}}}[a]\le L_{\edgeinfty{u}{F\setminus\inset{u}}}.\]
\end{defn}

It will be useful to define a {\em set} of degree constrains that basically keeps track of the set of simple degree constraint:
\begin{defn} \label{def:set-deg-constraint}
	A set of degree constraints for hypergraph $G=(\cV,\cE)$, is a subset $\cC(G)\subseteq \cV\times \cE$. We now overload notation by also referring to $L_{\edgeinfty{u}{F\setminus\inset{u}}}$ as $L_{(u,F)}$.
\end{defn}
Following~\cite{ngo-survey}, we denote the corresponding constraints acyclic (and we overload notation and call $G$) to be acyclic if the simple graph obtained by replacing the constraint $\edge{u}{F\setminus \set{u}}$ by binary directed edges $\inset{\edge{u}{v}|v\in F\setminus \inset{u}}$ is acyclic.

Since $\ell_1$ constraints do not have any inherent direction we impose such a direction. In particular, fix an ordering of vertices $u_1,\dots,u_n$ in $\cV$ and for each  edge $F\in \cE$ if $u$ is the first vertex in $F$ according to this order, then add $(u,F)$ to $\cC(G)$.

For the rest of the section fix $G=(\cV,\cE)$ and a set of degree constraints $\cC(G)$ (as defined above).
Recall we want to consider the setup where we have for each $F\in \cE$, an $\ell_1$ bound $N_F$ (i.e. size of $R_F$). We now (re)consider the AGM/edge cover LP  for hypergraphs:
\begin{align*}
& \min \usum{F \in \cE}  x_F\cdot  \log(N_F) \tag{$\LP^{(+)}$}\\ 
& \usum{u\ni F \in \cE}  x_{F} \ge 1 \quad \forall u \in V  \numberthis \label{eq:gen-covering-hyper-corr} \\
& x_{F} \ge 0 \quad \forall F \in E 
\end{align*}

We will argue that for any DB instance that satisfies the $\ell_1$ bounds, we can upper bound the join output size by
\[\uprod{F \in \cE} \inparen{N_F}^{x_F}.\] 
We argue this by induction on the size of $G$. 
As the base case consider the situation that $|\cV|=1$. In this case we have the set intersection problem and it is known that the above bound is true (this e.g. follows from the fact that the size of the output is bounded by the size of smallest set (and then using Lemma~\ref{lemma:zero-sum} along with~\eqref{eq:gen-covering-hyper-corr}). 

Let us assume that for every $a\in \Dom(a)$, we have
\begin{equation}
\label{eq:hyper-assum-corr}
d_{\edge{u}{F\setminus\inset{u}}}[a] = L_{(u,F)},
\end{equation}
for some value $L_{(u,F)}$.
(We'll later outline how we can get rid of this assumption.)

We'll use notation equivalent to one used in~\cite[Theorem 4.1]{ngo-survey} (though in there the recursion happens in the reverse topological ordering):
\begin{align*}
\partial(u_1)&=\inset{F|u_1\in F} &\\
\cE'&=\inset{F|F\cup\inset{u_1}\in\cE\text{ or } F\in \cE\setminus \partial(u_1)\wedge F\ne\emptyset} &\\
\cV'&= \cV\setminus\inset{u_1}&\\
G'&=\inparen{\cV',\cE'}&\\
\cC(G')&= \cC(G)\setminus\inset{(u_1,F)|(u_1,F)\in \cC(G)}&\\
R'_F &=\begin{cases}
R_F & \text{ if } F\in \cE\setminus \partial(u_1)\\
\pi_F\inparen{R_{F\cup\inset{u_1}}} & F\cup\inset{u_1}\in\cE
\end{cases}
&F\in\cE'\\
N'_F &=\begin{cases}
N_F & \text{ if } F\in \cE\setminus \partial(u_1)\\
L_{(u_1,F)} & F\cup\inset{u_1}\in\cE
\end{cases}
&F\in\cE'\\
x'_F&=\begin{cases}
x_F &\text{ if } F\in \cE\setminus \partial(u_1)\\
x_{F\cup\inset{u_1}}& F\cup\inset{u_1}\in\cE
\end{cases}
&F\in\cE'\\
\end{align*}

Note that by arguments similar to those in Section~\ref{sec:overview} the number of choices of $a\in\Dom(u_1)$ that can be in the output is upper bounded by
\begin{equation}
\label{eq:u-1-ub}
\min_{u_1\ni F\in\cE} \inset{\frac{L_F}{L_{(u,F)}}} \le 
\prod_{u_1\ni F\in\cE} \inparen{\frac{L_F}{L_{(u,F)}}}^{x_F}.
\end{equation}

Note that for each such choice of $a$ for $u_1$ above, we have an instance of the join problem on $G'$ (as defined above) and by induction hypothesis the size of each of these instance is upper bounded by (note that $\vx'$ is indeed a feasible solution for $\LP^{(+)}$ on $G'$):
\begin{equation}
\label{eq:G'-ub}
\uprod{F \in \cE'} \inparen{N'_F}^{x'_F}
=\uprod{F \in \cE\setminus \partial(u_1)} \inparen{N'_F}^{x'_F}\cdot \uprod{F\cup\inset{u_1}\in\cE} \inparen{L_{(u_1,F)}}^{x_{F\cup\inset{u_1}}}
\end{equation}

Multiplying the bounds in~\eqref{eq:u-1-ub} and~\eqref{eq:G'-ub}, we re-prove the AGM bound as desired.

We quickly address the assumption in~\eqref{eq:hyper-assum-corr}. First we get around the exact degree assumption by assuming that all degrees are within buckets of type $[2^i,2^{i+1})$ for some $i$. Again, if we assumed the final AGM bound for all degree configurations, we will incur an extra $O(\log^{|\cE|}{N})$ multiplicative factor. However, we can use the trick of `pushing' the sum over all degree configuration inside the product using H\"{o}lder's inequality as we have done in other cases.

We conclude this section by noting that if we `unroll' the recursion above (as we have done implicitly in our arguments), then it turns out that for each $F\in\cE$ not only do we have to monitor the degree constraint $(u,F)$ (where again $u$ is the first according to some global ordering) but if say $F=(u=u_0,u_1,\dots,u_\ell)$ then we need to keep track of `degrees' for $\edge{u_i}{(u_{i+1},\dots,u_\ell)}$ for all $0\le i\le \ell-1$. We leave the question of working out the details as an tantalizing possibility for future work.

\subsection{More Open Questions} \label{sec:app-concl}
We leave the following open questions for future and discuss current limitations to achieve them.
\begin{ques} \label{ques:gen-bounds}
Can we generalize our results for same $\ell_1$ and $\ell_{\infty}$ bounds for all edges to different $\ell_1$ and $\ell_\infty$ bounds, potentially resolving Question~\ref{ques:acyclic} completely in affirmative for all simple graphs?
\end{ques}
Our results in this regime currently rely on very specific structural decompositions based on the value of $d$. Extending them to different $\ell_1$, $\ell_{\infty}$ values could need more finer decompositions. Finally, we note that in most of the above cases, our degree configuration is more robust and extensible but our structural results are what cause the limitation. 

Our assumption that we fix $p$ upfront in Question~\ref{ques:main} leads to the following question:
\begin{ques} \label{ques:gen-norm}
What is the optimal $p$ that must be chosen for each relation for a given query and input data? 
\end{ques}
We would like to emphasize here that Algorithm~\ref{algo:generic-ub} is $p$-agnostic and uses only a global ordering of attributes in the query. As a result, we conclude by conjecturing that the algorithmremains the same in both the above scenarios, while the analysis could be more technically involved. 

Finally, we consider the following question.
\begin{ques} \label{ques:gen-p}
Can we extend our results on simple graphs for $p \in [1, 2]$ to the case when $p > 2$?
\end{ques}
It turns out that when $p > 2$, the optimal hard instances are not necessarily Cartesian product-based. To see this, consider the case when $p = 3$, $G$ is a triangle with a cyclic orientation and all $\ell_3$-norm bounds are upper bounded by $L$. The best Cartesian product based lower bound that can be achieved in this case is $L^{9/4}$ (using our LP-based result). However, we can get a (trivial) lower bound of $L^3$ on $|\TJ|$ using the below instance: 
\begin{align*}
|| R_{\edge{A}{B}} ||_{3} = || R_{\edge{B}{C}} ||_{3} = || R_{\edge{C}{A}} ||_{3}  = \{ (i, i): i \in [L^3] \}.
\end{align*}
We present a slightly more formal result below and end with a conjecture, as we did for the hypergraph case.
\subsubsection{Extending our results to $p \in [1, \infty]$ for all simple graphs}
In this section, we outline how some of our techniques for $p\in [1,\infty)$ (which have so far been used only for graphs with girth at least $p + 1$) can be extended to work for all simple graphs. The aim of this section is not to present the most general result we can prove for this case but rather to show the bottleneck in extending our arguments to the case of arbitrary simple graphs. To this end, we first fix $p > |V| - 1$. We will focus on the case when $G$ is a triangle with a cyclic orientation, all $\ell_{p}$-norm bounds are upper bounded by $L$ and outline how we can prove tight bounds (up to a poly-logarithmic factor in $|E|$) in this case. We mainly consider the problem of proving an optimal lower bound on the size of the join output and briefly discuss a plan of attack for the upper bound. 

We consider the problem of obtaining the best lower bound in this case when $G$ is a triangle with a cyclic orientation. As we saw in Section~\ref{sec:concl}, our current Cartesian product based lower bounds do not hold when $p > 2$. We define the following LP.
\begin{align*}
& \max \quad  y_{A} + y_{B} + y_{C}  + m \tag{$\LP_{p}^{(**)}$} \\
& \text{s.t. } \frac{y_{A}}{p}   + y_{B}  + \frac{m}{p} \le \log (L)   \\ 
& \frac{y_{B}}{p}   + y_{C}  + \frac{m}{p} \le \log (L)   \\ 
&  \frac{y_{C}}{p}   + y_{A}  + \frac{m}{p} \le \log (L). 
\end{align*}
Consider an optimal solution $\vy^* = (y^*_A, y^*_B, y^*_C, m^*)$ to the LP above. We construct the following instance $I$ for $\TJ$, where we define (assuming $D_A = \floor{2^{y^*_A}}$, $D_B = \floor{2^{y^*_B}}$, $D_C = \floor{2^{y^*_C}}, M = \floor{2^{m^*}}$)
\begin{align*}
R_{\edge{A}{B}} & = \{ [(i - 1)\cdot D_A + 1, i\cdot D_A] \times [(i - 1)\cdot D_B + 1, i\cdot D_B] : i \in [M] \} \\
R_{\edge{B}{C}} & = \{ [(i - 1)\cdot D_B + 1, i\cdot D_B] \times [(i - 1)\cdot D_C + 1, i\cdot D_C] : i \in [M] \} \\
R_{\edge{C}{A}} & = \{ [(i - 1)\cdot D_C + 1, i\cdot D_C] \times [(i - 1)\cdot D_A + 1, i\cdot D_A] : i \in [M] \}.
\end{align*} 
Based on above, we have
\begin{align}
\norm{R_{\edge{A}{B}}}_{p} & = \sqrt[p]{D_{A} \cdot D_{B}^{p} \cdot M} \\
& \le 2^{\frac{y^*_{A}}{p} + y^*_{B} + \frac{m^*}{p}} \\ 
& \le 2^{\log(L)} = L,
\end{align}
where the first inequality follows from the definition of $D_A, D_B$ and $M$ and the second inequality follows from dual constraint above (in $\LP^{(**)}$). Using symmetry, we can do a similar argument for 
$\norm{R_{\edge{B}{C}}}_{p}$ and $\norm{R_{\edge{C}{A}}}_{p}$ as well. 

Our instance is valid and we get a lower bound of
\begin{align*}
|\TJ| & \ge D_{A} \cdot D_{B} \cdot D_{C} \cdot M \\
& \ge \frac{1}{8} 2^{y^*_A + y^*_B + y^*_C + m^*} 
\end{align*}
from it.

We now attempt to see if we can obtain an upper bound based on $\LP_{p}^{(**)}$ as well. To do this, we come up with an upper bound $\B'(\bard, G)$ (for a fixed degree configuration $\bd = (d_{\edge{v}{u}})_{d_{\edge{v}{u}} \le L}$) as the minimum of the following three expressions:
\begin{align*}
& \frac{2^{p} \cdot L^p}{d_{\edge{A}{B}}^p} \min \left \{  d_{\edge{A}{B}}, \frac{2^p \cdot L^p}{d_{\edge{B}{C}}^p} \right \}  \min \left \{  d_{\edge{B}{C}} ,\frac{2^p \cdot L^p}{d_{\edge{C}{A}}^p} \right \} \\
& \frac{2^{p} \cdot L^p}{d_{\edge{B}{C}}^p} \min \left \{  d_{\edge{B}{C}}, \frac{2^p \cdot L^p}{d_{\edge{C}{A}}^p} \right \}  \min \left \{  d_{\edge{C}{A}}, \frac{2^p \cdot L^p}{d_{\edge{A}{B}}^p} \right \} \\
& \frac{2^{p} \cdot L^p}{d_{\edge{C}{A}}^p} \min \left \{  d_{\edge{C}{A}}, \frac{2^p \cdot L^p}{d_{\edge{A}{B}}^p} \right \}  \min \left \{  d_{\edge{A}{B}}, \frac{2^p \cdot L^p}{d_{\edge{B}{C}}^p} \right \}.
\end{align*}
One natural option is to use $\LP^{(*)}$, which we used to upper bound $\B'(\bard, G)$ for the case when $p \le |V| - 1$.
\begin{align*}
& \max z_{C} \tag{$\LP^{(*)}$} \\
&   z_{C} \le z_{\edgep{A}{B}} + y'_{B} + y'_{C} \\
&   z_{C} \le z_{\edgep{B}{C}} + y'_{A} + y'_{A} \\ 
&   z_{C} \le z_{\edgep{C}{A}} + y'_{A} + y'_{B} \\  
&  \frac{z_{\edgep{A}{B}}}{p} + y'_{B} \le \log(L) \\
&  \frac{z_{\edgep{B}{C}}}{p} + y'_{C} \le \log(L) \\
& \frac{z_{\edgep{C}{A}}}{p} + y'_{A} \le \log(L) \\
& y'_{A} \le z_{\edge{A}{B}}, y'_B \le z_{\edge{B}{C}}, y'_C \le z_{\edge{C}{A}} \\
& y'_{A}, y'_{B}, y'_{C}, z'_{A}, z'_{B}, z'_{C} \ge 0. 
\end{align*}
To argue that we can upper bound $|\TJ|$ by $\LP_{p}^{(**)}$ as well, we first need to prove the below result.
\begin{lemma}
	\begin{align*}
	\B'(\bard, G) & \le 2^{\LP^{(*)}}.
	\end{align*}
\end{lemma}
This result can be proved with appropriate definition of a feasible solution $(z_{\edge{A}{B}}, z_{\edge{B}{C}}, z_{\edge{C}{A}})$ and $(y'_{A}, y'_{B}, y'_{C})$ based on $\B'(\bard, G)$. For the next step, we state the below conjecture (which we believe to be true):
\begin{conjecture}
	\begin{align*}
	2^{\LP^{(*)}} & \le 2^{\LP_{p}^{(**)}} 
	\end{align*}
\end{conjecture}
Note that if it is indeed true, by strong duality, we would get tight bounds on $|\TJ|$ up to a polylogarithmic factor of $\log(L)^{3}$ (i.e., number of degree configurations). 